\documentclass[10pt,journal,compsoc]{IEEEtran}
\usepackage{amsmath,amsfonts}
\usepackage{array}
\usepackage{textcomp}
\usepackage{url}
\usepackage{verbatim}
\usepackage{times}
\usepackage{soul}
\usepackage{url}
\usepackage[hidelinks]{hyperref}
\usepackage[utf8]{inputenc}
\usepackage[small]{caption}
\usepackage{graphicx}
\usepackage{amsthm}
\usepackage{booktabs}
\usepackage{tabularx}
\usepackage{multirow}
\usepackage{color}
\usepackage[normalem]{ulem}

\def\BibTeX{{\rm B\kern-.05em{\sc i\kern-.025em b}\kern-.08em
    T\kern-.1667em\lower.7ex\hbox{E}\kern-.125emX}}
\usepackage{balance}

\usepackage{algorithm}
\usepackage[noend]{algpseudocode}
\usepackage{amssymb}
\usepackage{amsfonts}
\usepackage{amsthm}
\usepackage{mathtools}
\usepackage{bbding}
\usepackage{multicol}
\usepackage{multirow}
\usepackage{scalerel}
\usepackage{subcaption}
\usepackage{color}
\usepackage{xcolor}
\usepackage{colortbl}
\usepackage{diagbox}
\usepackage{stfloats}
\usepackage{xpatch}
\usepackage{tablefootnote}

\makeatletter
\newif\ifdiagbox@cellEmpty@

\define@key{diagbox}{widest}{%
  \ifdiagbox@cellEmpty@
    \setkeys{diagbox}{innerwidth=\widthof{#1}}%
    \PackageWarning{diagbox}{innerwidth is set to \the\diagbox@wd}%
  \fi}

\define@key{diagbox}{highest}{%
  \ifdiagbox@cellEmpty@
    \setkeys{diagbox}{height=#1}%
  \fi}

\define@key{diagbox}{text}{%
  \def\diagbox@text{#1}}

\define@key{diagbox}{align}{%
  \in@{#1}{l, c, r}%
  \ifin@
    \def\diagbox@align{#1}%
  \else
    \PackageError{diagbox}%
      {The value of key "text" must be one of "l", "c", and "r".}%
  \fi}

\xpatchcmd{\diagbox@double}{%
  \setkeys{diagbox}{dir=NW,#1}%
}{%
  \if\relax\detokenize{#2}\relax
    \if\relax\detokenize{#3}\relax
      \diagbox@cellEmpty@true
      \setkeys{diagbox}{highest=1\line, align=l, text=\@empty}%
    \fi
  \fi
  \setkeys{diagbox}{dir=NW, #1}%
  \ifdiagbox@cellEmpty@
    \rlap{\makebox
      [\dimexpr\diagbox@wd-\diagbox@insepl-\diagbox@insepr\relax]%
      [\diagbox@align]%
      {\diagbox@text}}%
  \fi
}{}{\ddt}
\makeatother

\newfloat{figtab}{htb}{fgtb}
\makeatletter
  \newcommand\figcaption{\def\@captype{figure}\caption}
  \newcommand\tabcaption{\def\@captype{table}\caption}
\makeatother

\newtheorem{assumption}{Assumption}
\newtheorem{theorem}{Theorem}
\newtheorem{lem}{Lemma}
\newtheorem{definition}{Definition}
\newtheorem{prop}{Proposition}
\newtheorem{rmk}{Remark}

\newcommand{\calA}{\mathcal{A}}
\newcommand{\calD}{\mathcal{D}}

\newcommand{\calU}{\mathcal{U}}

\newcommand{\calX}{\mathcal{X}}
\newcommand{\calY}{\mathcal{Y}}

\newcommand{\calN}{\mathcal{N}}

\newcommand{\calG}{\mathcal{G}}

\newcommand{\RR}{\mathbb{R}}

\newcommand{\gray}[1]{{\color{gray}#1}}
\newcommand{\red}[1]{{\color{red}#1}}

\usepackage{pifont}
\title{FedAdOb: Privacy-Preserving  \underline{Fed}erated Deep Learning with \underline{Ad}aptive \underline{Ob}fuscation  }

\author{Hanlin Gu, Jiahuan Luo, Yan Kang, Yuan Yao, ~\IEEEmembership{Member,~IEEE,} Gongxi Zhu, Bowen Li,\\Lixin Fan,~\IEEEmembership{Member,~IEEE,} 
Qiang Yang,~\IEEEmembership{Fellow. ~IEEE}
\IEEEcompsocitemizethanks{
\IEEEcompsocthanksitem Hanlin Gu, Jiahuan Luo, Yan Kang and Lixin Fan are with WeBank AI Lab, WeBank, China. E-mail: \{allengu, jihuanluo, yangkang, lixinfan\}@webank.com. 
\IEEEcompsocthanksitem Yuan Yao is with the Department of math, , Hong
Kong University of Science and Technology, Hong Kong. E-mail: yuany@ust.hk.
\IEEEcompsocthanksitem Gongxi Zhu is with the School of Information and Software Engineering, University of Electronic Science and Technology of China, Chengdu, China. E-mail: gx.zhu@foxmail.com.
\IEEEcompsocthanksitem Bowen Li is with the Department of Computer Science and Engineering, Shanghai Jiao Tong University, Shanghai 200240, China. E-mail: li-bowen@sjtu.edu.cn.
\IEEEcompsocthanksitem  Qiang Yang is with the Department of Computer Science and Engineering, Hong
Kong University of Science and Technology, Hong Kong and  WeBank AI Lab, WeBank, China. E-mail: qyang@cse.ust.hk.}
\thanks {Corresponding author: Lixin Fan.}}

\begin{document}

\IEEEtitleabstractindextext{
\begin{abstract}
Federated learning (FL) has emerged as a collaborative approach that allows multiple clients to jointly learn a machine learning model without sharing their private data. 
The concern about privacy leakage, albeit demonstrated under specific conditions \cite{zhu2019dlg}, has triggered numerous follow-up research in designing powerful attacking methods and effective defending mechanisms aiming to thwart these attacking methods. Nevertheless, privacy-preserving mechanisms employed in these defending methods invariably lead to compromised model performances due to a \textit{fixed} obfuscation applied to private data or gradients. In this article, we, therefore, propose a novel \textit{adaptive} obfuscation mechanism, coined FedAdOb, to protect private data without yielding original model performances. Technically, FedAdOb utilizes passport-based adaptive obfuscation to ensure data privacy in both horizontal and vertical federated learning settings. The privacy-preserving capabilities of FedAdOb, specifically with regard to private features and labels, are theoretically proven through Theorems 1 and 2. Furthermore, extensive experimental evaluations conducted on various datasets and network architectures demonstrate the effectiveness of FedAdOb by manifesting its superior trade-off between privacy preservation and model performance, surpassing existing methods.


\end{abstract}
}

\maketitle

\IEEEdisplaynontitleabstractindextext

\IEEEpeerreviewmaketitle

\section{Introduction}

Federated Learning (FL) 
offers a privacy-preserving framework that allows multiple organizations to jointly build global models without disclosing private datasets \cite{konevcny2015federated,konevcny2016federated,mcmahan2017communication,yang2019federated}.
Two distinct paradigms have been proposed in the context of FL \cite{yang2019federated}: Horizontal Federated Learning (HFL) and Vertical Federated Learning (VFL).
HFL focuses on scenarios where multiple entities have similar features but different samples. It is suitable for cases where data sources are distributed, such as healthcare institutions contributing patient data for disease prediction. On the other hand, VFL addresses situations where entities hold different attributes or features of the same samples. This approach is useful in scenarios like combining demographic information from banks with call records from telecom companies to predict customer behavior. 

Since the introduction of HFL and VFL, studies have highlighted the existence of privacy risks in specific scenarios. For instance, Zhu et al. reported that semi-honest attackers can potentially infer \textit{private features} from the released model and updated gradients \cite{zhu2019dlg} in HFL. The leakage risk of \textit{private labels} has also been identified by Jin et al. \cite{jin2021cafe} and Fu et al. \cite{fu2022label} in VFL. 
To mitigate these concerns, a variety of defense methods have been proposed to enhance the privacy-preserving capabilities of Federated Learning. Specifically, for HFL, defense approaches such as differential privacy \cite{abadi2016deep}, gradient compression \cite{lin2018deep}, SplitFed \cite{thapa2020splitfed}, and mixup \cite{huang2020instahide,zhang2018mixup} have been put forward. For VFL, defense methods encompass differential privacy techniques like noise addition and random response \cite{fu2022label,liu2021defending,ghazi2021deep,yang2022differentially}, as well as gradient discretization \cite{dryden2016communication} and gradient sparsification \cite{aji2017sparse}. Nevertheless, our comprehensive analysis and empirical investigation (refer to Sect. \ref{sec:exp}) reveal that all the aforementioned defense mechanisms experience a certain level of performance degradation.



The existing privacy defense methods \cite{fu2022label,liu2020secure,dryden2016communication,aji2017sparse,lin2018deep,huang2020instahide} can be fundamentally understood as obfuscation mechanisms that offer privacy assurances through the utilization of an obfuscation function $g(\cdot)$ applied to private feature (such as transferred model weights in HFL and forward embedding in VFL). Take the HFL as one example (see details for both HFL and VFL in Sect. \ref{sec:FedAdOb}), the process is depicted as follows:
\begin{equation}
\begin{split}
    \label{eq:fix-obfuscation}
        x \stackrel{g(\cdot)}
        \longrightarrow g(x) \stackrel{F_{\omega}}\longrightarrow \ell \longleftarrow 
y, \\
\end{split}
\end{equation}
in which $F$ represents the model parameterized by $\omega$ and $\ell$ is the loss. Noted that in Eq. \eqref{eq:fix-obfuscation}, the level of privacy-preserving capability is determined by the \textit{extent of obfuscation} (the distortion $g()$), which is controlled by a fixed hyperparameter. As examined in previous studies \cite{zhang2022trading,kang2022framework}, substantial obfuscation inevitably results in the loss of information in $g(x)$ and $g(\omega)$, thereby impacting the model's performance (as observed in the fundamental analysis in Sect. \ref{sec:update-AO} and empirical investigation in Sect. \ref{sec:exp}). We contend that a fixed obfuscation strategy fails to consider the dynamic nature of the learning process and serves as the fundamental cause of model performance deterioration.  
As a remedy to shortcomings of the fixed obfuscation,  we propose to \textbf{adapt obfuscation function} $g_{\omega}$ during the learning of model $F_\omega$, such that the obfuscation itself is also optimized during the learning stage. 
That is to say, the learning of the obfuscation function also aims to preserve model performance by tweaking model parameters $\omega$:
\vspace{-0.4em}
\begin{equation}
\begin{split}
    \label{eq:ada-obfuscation}
        x \stackrel{g_\red{\omega}(\cdot)}
        \longrightarrow g(x) \stackrel{F_\omega}\longrightarrow \ell \longleftarrow 
y, \\
\end{split}
\end{equation}
We regard this adaptive obfuscation in both HFL and VFL as the gist of the proposed method, called \textbf{FedAdOb}. In this work, we implement the adaptive obfuscation based on the passport technique, which was originally designed for protecting the intellectual property of deep neural networks (DNN) \cite{fan2021deepip,li2022fedipr}. More specifically, we decompose the entire network into a bottom model, which is responsible to extract features, and a top model which outputs labels.  Both the bottom model and top model embed private passports to protect, respectively,  privacy of features and labels. The proposed framework FedAdOb
has three advantages: i) Private passports embedded in bottom and top layers unknown to adversaries, which prevents attackers from inferring features and labels. It is \textit{exponentially hard} 
to infer features by launching various attacks, while attackers are defeated by a \textit{non-zero recovery error} when attempting to infer private labels (see Sect. \ref{sec:analysis}).
ii) Passport-based obfuscation is learned in tandem with the optimization of model parameters, thereby preserving model performance (see Sect. \ref{sec:fedpass} and Sect. \ref{sec:update-AO}). iii) The learnable obfuscation is efficient, with only minor computational costs incurred since no computationally extensive encryption operations are needed (see Sect. \ref{sec:exp}). Our contributions are summarized as follows:


\begin{itemize}
    \item We propose a general privacy-preserving FL framework by leveraging the adaptive obfuscation, named FedAdOb, which could simultaneously preserve the model performance, data privacy, and training efficiency. The proposed framework can be applied in the HFL and VFL.
    \item Theoretical analysis demonstrates that the proposed FedAdOb could protect features and labels against various privacy attacks.
    \item Extensive empirical evaluations conducted on diverse image and tabular datasets, along with different architecture settings, demonstrate that FedAdOb achieves a trade-off between privacy protection and model performance that is comparable to near-optimal levels, surpassing existing fixed obfuscation defense methods.
\end{itemize}

\section{Related Work} 

Federated learning has three categories \cite{yang2019federated} from the perspective of how data is distributed among participating parties:
\begin{enumerate}
\item \textit{horizontal federated learning}: datasets share the same feature space but different space in samples; 
\item \textit{vertical federated learning}: two datasets share the same sample ID space but differ in feature space; 
\item \textit{federated transfer learning}: two datasets differ not only in samples but also in feature space.
\end{enumerate} 
We focus on HFL and VFL settings. In the following, we review privacy attacks and defense mechanisms proposed in the literature for HFL and VFL.

\subsection{Horizontal Federated Learning}
The horizontal federated learning (HFL) was proposed by McMahan et al.~\cite{mcmahan2017communication}, aiming to build a machine learning model based on datasets that are distributed across multiple devices \cite{konevcny2016federated, mcmahan2017communication} without sharing private data with the server and other devices.

\noindent\textbf{Privacy Attacks.} 
There are mainly two types of privacy attacks in HFL: gradient inversion (GI)~\cite{zhu2019dlg,geiping2020inverting} and model inversion (MI)~\cite{he2019model}. GI tries to infer the private data by minimizing the distance between the estimated and observed gradients. MI recovers the private data by comparing the estimated feature output and the observed one.

\noindent\textbf{Defense Mechanisms.} Differential privacy (DP) \cite{abadi2016deep}, gradient compression (GC) \cite{lin2018deep,gu2021federated}, homomorphic encryption (HE) \cite{hardy2017private}, secure multi-party computation (MPC) \cite{SecShare-Adi79}, and mixup \cite{zhang2018mixup,huang2020instahide} are widely used privacy defense mechanisms. The cryptographic techniques HE and MPC can guarantee data privacy but have a high computational and communication expense. DP and GC distort shared model updates to protect privacy, often dramatically worsening model performance. Mixup protects data privacy by mixing-up private images with other images randomly sampled from private and public datasets, which also degrades the model's performance.

\subsection{Vertical Federated Learning} 

Vertical Federated Learning (VFL) is proposed to address enterprises’ demands of leveraging features dispersed among multiple parties to achieve better model performance, compared with the model trained by a single party, without jeopardizing data privacy~\cite{yang2019federated}. In VFL, privacy is a paramount concern because participants of VFL typically are companies whose data may contain valuable and sensitive user information. We review privacy attacks and defense mechanisms of VFL as follows.


\noindent\textbf{Privacy Attacks in VFL.} Privacy attacks can be categorized into Feature Inference (FI) attacks and Label Inference (LI) attacks in VFL. FI attacks are initiated by the active party and aim to discover the features of the passive party. Model inversion~\cite{he2019model} and CAFE~\cite{jin2021cafe} are the two most well-known FI attacks in VFL. LI attacks, on the other hand, are initiated by the passive party and aim to infer labels possessed by the active party. LI attacks have two types: gradient-based and model-based. The former~\cite{oscar2022split} calculates the norm or direction of gradients back-propagated to the passive party to determine labels, while the latter~\cite{fu2022label} first pre-trains an attacking model and then leverages this attacking model to infer labels.

\noindent\textbf{Defense Mechanisms.}
Cryptography-based defense mechanisms such as Homomorphic Encryption (HE) and Multi-Party Computation (MPC) are widely adopted in VFL to protect data privacy for their high capability of preserving privacy. For example, Hardy et al.~\cite{hardy2017private} proposed vertical logistic regression (VLR) using homomorphic encryption (HE) to protect feature privacy, and Zhou et al.~\cite{secureboost} proposed the SecureBoost, a VFL version of XGBoost, that leverages HE to protect gradients exchanged among parties. However, Cryptography-based defense mechanisms typically have high computation and communication costs. Thus, they are often applied to shallow models (e.g., LR and decision trees) compared with deep neural networks.

Non-cryptography-based defense mechanisms protect data privacy typically by distorting the model information to be disclosed to adversaries. For example, Differential Privacy~\cite{abadi2016deep} adds noise to disclosed model information while Sparsification~\cite{fu2022label,lin2018deep} compresses disclosed one to mitigate privacy leakage. Specialized defense techniques such as MARVELL~\cite{oscar2022split} and Max-Norm~\cite{oscar2022split} are designed to thwart gradient-based label inference attacks by applying optimized or heuristic noise to gradients. Additionally, InstaHide~\cite{huang2020instahide} and Confusional AutoEncoder~\cite{zou2022defending} are defensive mechanisms that encode private data directly to enhance data privacy. 

\begin{table*}[htbp] 
\center 
\caption{Threat model we consider in this work.}
\begin{tabular}{c c c c c } 
\hline
  Setting & Adversary & Attacking Target & Attacking Method & Adversary's Knowledge  \\ 
   \hline \hline
    \multirow{5}{*}{VFL}& \multirow{3}{*}{\shortstack{Semi-honest \\ passive party}}  & \multirow{3}{*}{\shortstack{\textit{Labels} owned by \\ the active party}} & PMC \cite{fu2022label} & A few labeled samples \\
   &&  & NS \cite{oscar2022split} & No prior \\
    &&  & DS \cite{oscar2022split} & No prior \\
   \cline{2-5}
   & \multirow{2}{*}{\shortstack{Semi-honest \\ active party}} & \multirow{2}{*}{\shortstack{\textit{Features} owned by \\ a passive party}} & BMI \cite{he2019model} & Some labeled samples \\ 
   & & & WMI \cite{he2019model,jin2021cafe} & Passive models \\ \hline\
   
     \multirow{4}{*}{HFL}&\multirow{4}{*}{\shortstack{Semi-honest \\ server}}  & \multirow{4}{*}{\textit{Private data} of clients} & WMI \cite{he2019model} & Model parameters and output \\
     
   &&&BMI \cite{he2019model}& Model output and a few auxiliary samples\\
   
 &&&WGI \cite{zhu2019dlg}& Model parameters and gradients \\
 &&&BGI \cite{zhu2019dlg}& Model gradients \\ \hline
\end{tabular}
\vspace{-0.6em}
\label{tab_threat_model}
\end{table*}

\section{Preliminary}

In this section, we introduce the two main federated learning settings: horizontal federated learning and vertical federated learning. 

\subsection{Horizontal Federated Learning}
Consider a Horizontal Federated Learning (HFL) consisting of $K$ clients who collaboratively train a HFL model $\omega$ to optimize the following objective:
\begin{equation} \label{eq:objective}
    \min_{\omega} \sum_{k=1}^K\sum_{i=1}^{n_k}\frac{\ell(F_\omega(x_{k,i}), y_{k,i})}{n_1+\cdots+n_K},
\end{equation}
where $\ell$ is the loss, e.g., the cross-entropy loss, $\calD_k=\{(x_{k,i}, y_{k,i})\}_{i=1}^{n_k}$ is the dataset with size $n_k$ owned by client $k$. 

In each communication round, client $k$ uploads their own model weights $\omega_k$ ($k=1,\cdots, K$) to the server and then the server aggregates all uploaded model weights through $\omega = \frac{1}{K}\sum_{k=1}^K\omega_k$. Next, the server distributes the aggregated weights to all clients (see Fig. \ref{fig:Framework}(a)).

\noindent\textbf{Threat Model:} We assume the server might be \textit{semi-honest} adversaries such that they do not submit any malformed messages but may launch \textit{privacy attacks} on exchanged information from other clients to infer clients' private data. 

We consider four types of threat models, summarized in Tab. \ref{tab_threat_model}: i) White-box Model Inversion (WMI), where the adversary knows the model parameters and feature output to restore private data via model inversion attack \cite{he2019model}; ii) Black-box Model Inversion (BMI), where the adversary only knows feature output and a few auxiliary samples to restore private data via model inversion attack \cite{he2019model}; iii) White-box Gradient Inversion (WGI), where the adversary knows the model parameters and gradients to restore private data via gradient inversion attack \cite{zhu2019deep}; IV) Black-box Gradient Inversion (BGI), where the adversary only knows the model parameters to recover private data via gradient inversion attack \cite{zhu2019deep}.

\subsection{Vertical Federated Learning} \label{sec:pre-vfl}
Consider a Vertical Federated Learning (VFL) setting consisting of one active party $P_0$ and $K$ passive parties $\{P_1, \cdots, P_K\}$  who collaboratively train a VFL model $\Theta=(\theta, \omega)$ to optimize the following objective: 
\begin{equation}\label{eq:loss-VFL}
\begin{split}
        \min_{\omega, \theta_1, \cdots, \theta_K} &\frac{1}{n}\sum_{i=1}^n\ell(F_{\omega} \circ (G_{\theta_1}(x_{1,i}),G_{\theta_2}(x_{2,i}), \\
        & \cdots,G_{\theta_K}(x_{K,i})), y_{i}),
\end{split}
\end{equation}
where passive party $P_k$ owns features $\calD_k = (x_{k,1}, \cdots, x_{k,n}) \in \mathcal{X}_k$ and the passive model $G_{\theta_k}$, the active party owns the labels $y \in \mathcal{Y}$ and active model $F_\omega$, $\mathcal{X}_k$ and $\mathcal{Y}$ are the feature space of party $P_k$ and the label space respectively. Each passive party $k$ transfers its forward embedding $H_k$ to the active party to compute the loss. The active model $F_\omega$ and passive models $G_{\theta_k},k \in \{1,\dots,K\}$ are trained based on backward gradients (See Fig. \ref{fig:Framework}(b) for illustration). Note that, before training, all parties leverage Private Set Intersection (PSI) protocols to align data records with the same IDs. 

\noindent\textbf{Threat Model:} We assume all participating parties are \textit{semi-honest} and do not collude with each other. An adversary (i.e., the attacker) $P_k,k=0,\cdots,K$ faithfully executes the training protocol but may launch privacy attacks to infer the private data (features or labels) for other parties. 

We consider two types of threat models summarized in Tab. \ref{tab_threat_model}: \romannumeral1) The active party wants to reconstruct the private features of a passive party through the WMI \cite{he2019model} or BMI~\cite{jin2021cafe, he2019model}. \romannumeral2) A passive party wants to infer the private labels of the active party through the Passive Model Completion (PMC)~\cite{fu2022label}, Norm-based Scoring (NS) function and Direction-based Scoring (DS) function \cite{oscar2022split}.

\section{The Proposed Method: FedAdOb}
\label{sec:FedAdOb}



We first introduce the adaptive obfuscation module and then apply this adaptive obfuscation module to the HFL and VFL settings.




\subsection{Adaptive Obfuscation Module} \label{sec:fedpass}
This section illustrates two critical steps of the Adaptive Obfuscation Module, i) embedding private passports to adapt obfuscation; ii) generating passports randomly to improve privacy-preserving capability.

\begin{figure}[!ht]
\centering
\includegraphics[width=0.49\textwidth]{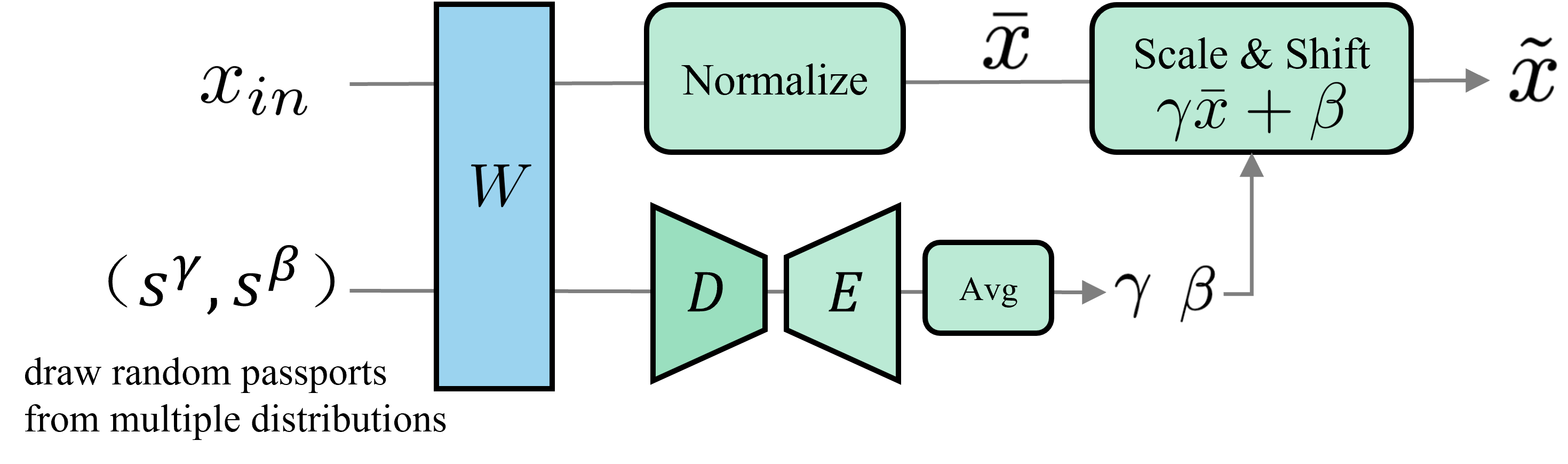}
\caption{Adaptive obfuscation ($g_W(\cdot)$). We implement $g_W(\cdot)$ by inserting a passport layer into a normal neural network layer. }
\label{fig:passport_layer}
\vspace{-3pt}
\end{figure}

\subsubsection{Random Passport Generation}

How passports are generated is crucial in protecting data privacy. 
Specifically, when the passports are embedded in a convolution layer or linear layer with $c$ channels\footnote{For the convolution layer, the passport $s\in \RR^{c\times h_1 \times h_2}$, where $c$ is channel number, $h_1$ and $h_2$ are height and width; for the linear layer, $s\in \RR^{ c\times h_1}$, where $c$ is channel number, $h_1$ is height.}, 
for each channel $j\in[c]$, 
the passport $s{(j)}$ (the $j_{th}$ element of vector $s$) is randomly generated as follows:
\begin{equation}\label{eq:sample-pst}
    s{(j)} \sim \calN(\mu_j, \sigma^2), \quad
    \mu_j \in \calU(-N, 0),
\end{equation}
where all $\mu_j, j=1,\cdots, c$ are different from each other, $\calU$ represents the uniform distribution, $\sigma^2$ is the variance of Gaussian distribution and $N$ is the \textit{passport mean range}, which are two crucial parameters of FedAdOb. The strong privacy-preserving capabilities rooted in such a random passport generation strategy are justified by theoretical analysis in  Theorems \ref{thm:thm1} and \ref{thm2} as well as experiment results in Sect. \ref{sec:exp}.

\subsubsection{Embedding Private Passports}
In this work, we adopt the DNN passport technique proposed by \cite{fan2019rethinking,fan2021deepip} as an implementation for the adaptive obfuscation framework of FedAdOb. 
Specifically, the adaptive obfuscation is determined as follows (described in Algo. \ref{alg:AO}):
\begin{equation}\label{eq:pst1}
\begin{split}
        g_W(x_{in}, s) = &
        \gamma(Wx_{in}) + \beta, \\
         \gamma=&\text{Avg}\Big( D\big(E(Ws^\gamma)\big)\Big),\\
        \beta = &\text{Avg}\Big( D\big(E(Ws^\beta)\big)\Big),
\end{split}
\end{equation}
where $W$ denotes the model parameters of the neural network layer for inserting passports, $x_{in}$ is the input fed to $W$, $\gamma$ and $\beta$ are the scale factor and the bias term. Note that the determination of the crucial parameters $\gamma$ and $\beta$ involves the model parameter $W$ with private passports $s^\gamma$ and $s^\beta$, followed by a autoencoder (Encoder $E$ and Decoder $D$ with parameters $W'$) 
and a average pooling operation Avg($\cdot$). Learning adaptive obfuscation formulated in Eq. \eqref{eq:pst1} brings about two desired properties: 
\begin{itemize}
    \item Passport-based parameters $\gamma$ and $\beta$ provide strong privacy guarantee (refer to Sect.~\ref{sec:analysis}): without knowing passports, it is exponentially hard for the attacker to infer layer input $x_{in}$ from layer output $g_W(x_{in}, s)$, because attacker have no access to $\gamma$ and $\beta$ (see Theorem \ref{thm:thm1}).
    \item Learning adaptive obfuscation formulated in Eq. (\ref{eq:pst1}) optimizes the model parameter $W$ through three backpropagation paths via $\beta, \gamma, W$, respectively, which helps preserve model performance. This is essentially equivalent to adapting the obfuscation (parameterized by $\gamma$ and $\beta$) to the model parameter $W$ (more explanations in Sect. \ref{sec:update-AO}); This adaptive obfuscation scheme offers superior model performance compared to fixed obfuscation schemes (see Sect. \ref{sec:exp}). 

\end{itemize}
\begin{algorithm}[!ht]
\caption{Adaptive Obfuscation ($g_W(\cdot)$)}
\begin{algorithmic}[1]
  \Statex \textbf{Input:} Model parameters $W$ of the neural network layer for inserting passports, the input $x_{in}$ to which the adaptive obfuscation applies, passport keys  $s = (s^\gamma, s^\beta)$.
  \Statex \textbf{Output:} The obfuscated version of the input.
  \State Compute $
         \gamma=\text{Avg}\left( D\big(E(W* s^\gamma)\big)\right)$
\State Compute $ \beta = \text{Avg}\left( D\big(E(W* s^\beta)\big)\right)$ \\
\Return $\gamma(W* x_{in}) + \beta$
	\end{algorithmic}
	\label{alg:AO}
\end{algorithm}
\begin{figure*}[!htbp]
\centering
	\begin{subfigure}{0.48\textwidth}
  		 	\includegraphics[width=1\textwidth]{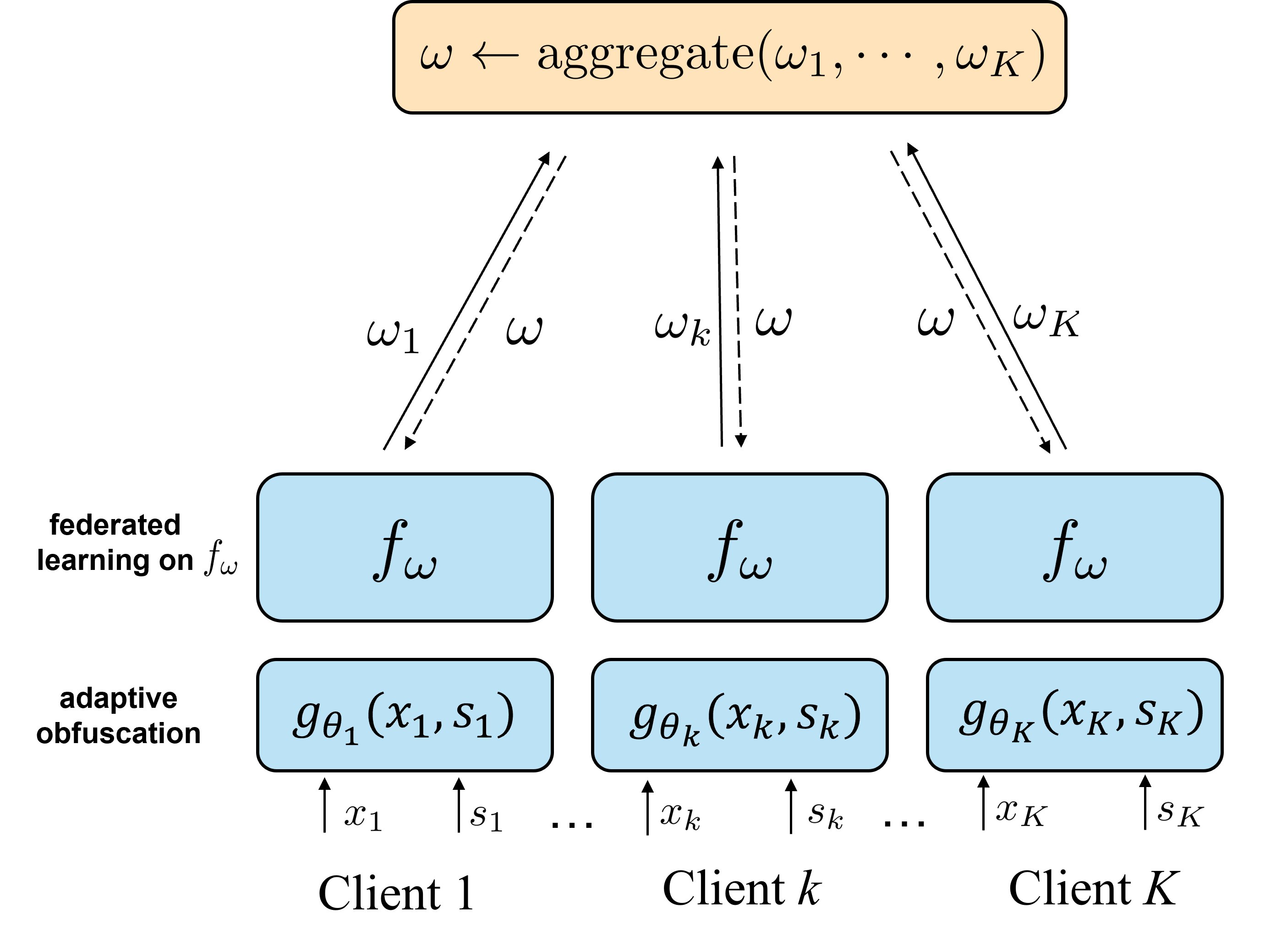}
      \subcaption{HFL}
    		\end{subfigure}
      \hspace{3pt}
      	\begin{subfigure}{0.46\textwidth}
  		 	\includegraphics[width=1\textwidth]{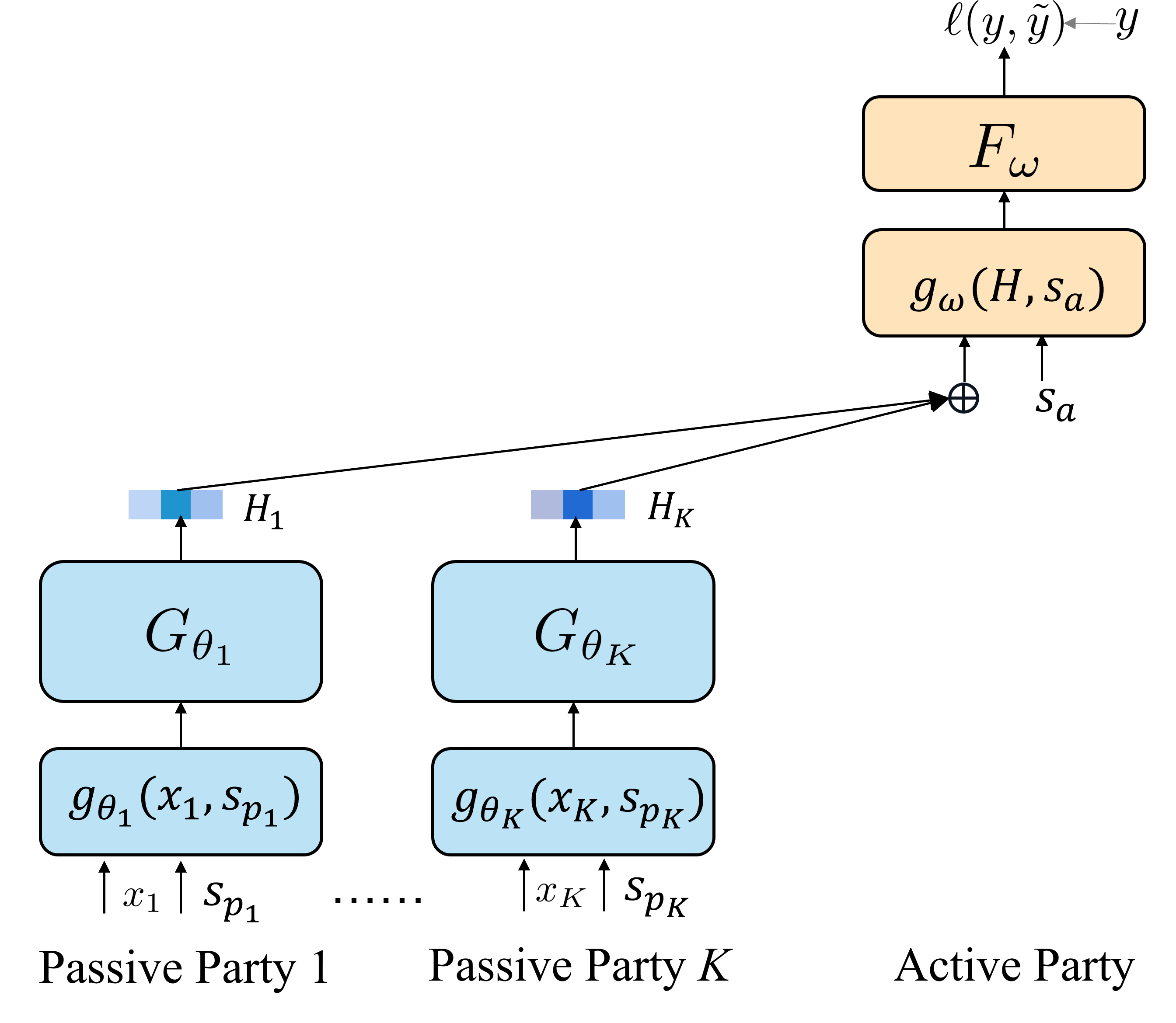}
      \subcaption{VFL}
    		\end{subfigure}
\caption{The left sub-figure illustrates the FedAdOb for HFL, including adaptive obfuscation $g_\theta$ and federated model $f_\omega$. The right sub-figure illustrates the FedAdOb for the VFL setting, in which multiple passive parties and one active party collaboratively train a VFL model, where passive parties only have the private features $x$, whereas the active party has private labels $y$. Both the active party and the passive party adopt adaptive obfuscation by inserting passports into their models to protect features and labels.}
\label{fig:Framework}
\vspace{-6pt}
\end{figure*}

\subsection{FedAdOb in HFL}
We first partition the neural network model into two components, namely the bottom layers represented by $G_\theta$ and the top layers represented by $f_\omega$. In this partitioning, the bottom layers are considered private, while the top layers are shared for aggregation in the federated learning process. Each client participating in the federated learning employs a private adaptive obfuscation technique within the bottom model to protect the client's private feature, as outlined in Eq. \eqref{eq:pst1}.

The training procedure of Federated Adaptive Obfuscation (FedAdOb) in the context of horizontal federated learning is depicted in Fig. \ref{fig:Framework}(a). Algo. \ref{alg:aof-hfl} provides a detailed description of this training procedure.
\begin{enumerate}
    \item Each client optimizes the adaptive obfuscation $g_{\theta_k}$ and federated model $f_\omega$ with its private passports $s_k$ and data $\calD_k$ according to the cross-entropy loss. Then, all clients upload the federated model $\omega_k$ to the server (line 4-8 of Algo. \ref{alg:aof-vfl});
    \item The central server aggregates federated models $\omega_k$ to obtain $\omega$ \cite{mcmahan2017communication} and distributes the aggregated model $\omega$ to all clients; (line 10-11 of Algo. \ref{alg:aof-vfl}); 
\end{enumerate}

The two steps iterate until the performance of the model does not improve. 

\begin{algorithm}[!ht]\vspace{-3pt}
	\caption{FedAdOb in HFL}
	\begin{algorithmic}[1]
	\Statex \textbf{Input:} Communication rounds $T$,  \# of clients $K$, learning rate $\eta$, the dataset $\calD_k=\{x_{k,i}, y_{k,i}\}_{i=1}^{n_k}$ owned by client $k$ and the passport mean range and variance $\{N^k, \sigma^k\}$ for client $k$.

    \Statex \textbf{Output:} $\theta_1, \cdots, \theta_K, \omega$ \vspace{4pt}
	    
	\State Initialize adaptive obfuscation  parameters $\theta_1, \cdots, \theta_K$ and public model parameter $\omega$. 
     \For{$t$ in communication round $T$} \vspace{2pt}
    \State \gray{$\triangleright$ \textit{Clients perform:}}
         \For{Client $k$ in $\{1,\dots,K\}$}:
                 \State $\omega_k \longleftarrow \omega$;
    \For{Batch $(B_x, B_y) \in \calD_k$}:
    \State Sample the passport tuple $s_k=(s_k^\gamma, s_{k}^\beta)$ via 
    \Statex \quad \qquad \qquad Eq. \eqref{eq:sample-pst} and $N^k, \sigma^k$; 
    \State Compute loss $\tilde{\ell}=\ell(f_{\omega_k} \circ g_{\theta_k}(s_k, B_x), B_y) $;
    \State $\theta_{k} \longleftarrow \theta_k - \eta \nabla_{\theta_k}\tilde{\ell}$; 
    \State $\omega_k \longleftarrow \omega_k - \eta \nabla_{\omega_k}\tilde{\ell}$;
\EndFor
\State  Upload $\omega_k$ to the server;
\EndFor
\State \gray{$\triangleright$ \textit{The server performs:}}
\State The server aggregates: $\omega = \frac{1}{K} (\omega_{1}+ \cdots + \omega_{K})$;

\State Distribute $\omega$ to all clients;
\EndFor \\
\Return $\theta_1, \cdots, \theta_K, \omega$
\end{algorithmic}\label{alg:aof-hfl}
\end{algorithm}

\subsection{FedAdOb in VFL}
As discussed in Sect. \ref{sec:pre-vfl}, the active parties in Vertical Federated Learning (VFL), who possess the bottom model, are concerned about the potential leakage of labels to passive parties. Conversely, the passive parties, who possess the top model, are keen to safeguard their private features from being reconstructed by the active parties.

To address the simultaneous protection of private features and labels, we introduce the adaptive obfuscation module in VFL, implemented separately for the bottom and top models. This module ensures the simultaneous protection of both labels and features. The framework of Federated Adaptive Obfuscation (FedAdOb) in vertical federated learning is depicted in Fig. \ref{fig:Framework}(b). The specific training procedure is illustrated as follows (described in Algo. \ref{alg:aof-vfl}):

\begin{enumerate}
    \item Each passive party $k$ applies the adaptive obfuscation to its private features with its private passports $s_{p_k}$ and then sends the forward embedding $H_k$ to the active party (line 3-9 of Algo. \ref{alg:aof-vfl});
    \item The active party sums over all $H_k, k \in \{1,\dots, K\}$ as $H$, and applies the adaptive obfuscation to $H$ with its private passports $s_a$, generating $\Tilde{H}$. Then, the active party computes the loss $\Tilde{\ell}$ and updates its model through back-propagation. Next, the active party computes gradients $\nabla_{H_{k}}\Tilde{\ell}$ for each passive party $k$ and sends $\nabla_{H_{k}}\Tilde{\ell}$ to passive party $k$ (line 10-19 of Algo. \ref{alg:aof-vfl});
    \item Each passive party $k$ updates its model $\theta_k $ according to $\nabla_{H_k}\Tilde{\ell}$ (line 20-22 of Algo. \ref{alg:aof-vfl}). 
\end{enumerate} 

The three steps iterate until the performance of the joint model does not improve.



\begin{algorithm}[!ht]\vspace{-3pt}
	\caption{FedAdOb in VFL}
	\begin{algorithmic}[1]
	   \Statex \textbf{Input:} Communication rounds $T$, passive parties number $K$, learning rate $\eta$, batch size $b$, the passport mean range and variance $\{N_a, \sigma_a\}$ and $\{N_{p_k}, \sigma_{p_k}\}$ for the active party and passive party $k$ respectively, the feature dataset $\calX_k = (x_{k,1},\cdots, x_{k,n_k})$ owned by passive party $k$, the aligned label $\calY = (y_1, \cdots, y_{n_0})$ owned by the active party.

        \Statex \textbf{Output:} Model parameters $\theta_1, \cdots, \theta_K, \omega$ \vspace{4pt}
	    
	\State Initialize model weights $\theta_1, \cdots, \theta_K, \omega$.
     \For{$t$ in communication round $T$} \vspace{2pt}
        \State \gray{$\triangleright$ \textit{Passive parties perform:}}
         \For{Passive Party $k$ in $\{1,\dots,K\}$}: 
         \State Sample a batch $B_x=(x_{k,1},\cdots, x_{k,b})$ from the 
         \Statex \quad \qquad dataset $\calX_k$;
         \State Sample the passport tuple $s_{p_k} = (s_{p_k}^\gamma, s_{p_k}^\beta)$ accord-
         \Statex \quad \qquad ing to Eq. \eqref{eq:sample-pst} and $N_{p_k}, \sigma_{p_k}$;
          \State Compute $\Tilde{B_x} = g_{\theta_k}( B_x,s_{p_k})$;
\State   Compute  $H_{k} \gets G_{\theta_k}(\Tilde{B_x})$;
\State  Send  $H_{k}$  to the active party;
\EndFor
\State \gray{$\triangleright$ \textit{The active party performs:}}
\State Obtain a batch $B_y$ from the label $\calY$ related to $B_x$;
\State $H = \sum_{k=1}^KH_k$;
\State Sample the passport tuple $s_a=(s_{a}^\gamma, s_{a}^\beta)$ via Eq. \eqref{eq:sample-pst} 
\Statex \qquad and $N_a, \sigma_a$ ;
\State Compute $\Tilde{H} = g_{\omega}(H,s_a)$;
\State Compute cross-entropy loss:
$\Tilde{\ell} = \ell(F_{\omega}(\Tilde{H}),B_y)$
\State Update the active model as: $\omega = \omega - \eta \nabla_\omega\Tilde{\ell}$;
   \For {$k$ in $\{1,\dots,K\}$}:   
   \State Compute and send  $\nabla_{{H_{k}}}\tilde{\ell}$ to each passive party $k$;
   \EndFor		
\State \gray{$\triangleright$ \textit{Passive parties perform:}}
   \For{Passive Party $k \in \{1,\dots,K\}$}: 
      \State  Update $\theta_k$ by $\theta_k = \theta_k - \eta [\nabla_{H_k}\tilde{\ell}] [\nabla_{\theta_k}H_k]$
\EndFor
\EndFor
\Return $\theta_1, \cdots, \theta_K, \omega$
	\end{algorithmic}\label{alg:aof-vfl}
\end{algorithm}

\subsection{Update Procedure by Adaptive Obfuscation} \label{sec:update-AO}
Consider a neural network $f_\Theta(x):\calX \to \RR$, where $x \in \calX$, $\Theta = (\omega, \theta_1, \cdots, \theta_K)$ denotes top model and bottom model parameters. Then we can reformulate the loss of FedAdOb as the  following:
\begin{prop} \label{prop:MP-AO}
The loss of the FedAdOb with adaptive obfuscation in the bottom model can be written as:
\begin{itemize}
\item For VFL:
\begin{equation} \label{eq:loss-aof-VFL}
\begin{split} \min_{\omega, \theta_1 ,\cdots, \theta_K} &\frac{1}{n}\sum_{i=1}^n\ell(F_{\omega} \circ (G'_{\theta_1}(x_{1,i}, s_{p_1})), \\
        &\cdots, G'_{\theta_K}(x_{K,i}, s_{p_K})) ,  y_{i})),
\end{split}
\end{equation}
\item For HFL:
\begin{equation} \label{eq:loss-aof-HFL}
    \min_{\theta_1, \cdots, \theta_K,\omega} \sum_{k=1}^K\sum_{i=1}^{n_k}\frac{\ell(F_\omega \circ (G'_{\theta_k}(x_{k,i}, s_{p_k}), y_{k,i})}{n_1+\cdots+n_K},
\end{equation}
\end{itemize}
where $G'_{\theta_k}()$ is the composite function $G_{\theta_k}\cdot g_{\theta_k}()$, $k=1, \cdots, K$. 
\end{prop}

Proposition \ref{prop:MP-AO} illustrates that FedAdOb is trained with the tuple \{(Feature, passport), label\} = $\{(x, s), y\}$ (see proof in Appendix B). Therefore, the training of FedAdOb can be divided two optimization steps. One is to maximize model performance with the data $(x,y)$ when the passport $s$ is fixed (including $W$ path). The second is also to maximize model performance with the passport $s$ when the data is fixed, i.e., the adaptive obfuscation is also optimized towards maximizing model performance (including $\gamma,\beta$ path). Specifically, we take the Eq. \eqref{eq:pst1} as one example, the the derivative for $g$ w.r.t. $W=\theta$ has the following three backpropagation paths via $\beta, \gamma, \theta$:
\begin{equation} \label{eq:upload-gradients-app1}
\frac{\partial g}{\partial \theta} = \left\{
\begin{aligned}
x_{in} \otimes diag(\gamma)^T + \beta \quad & \theta \text{ path}\\
(\theta x_{in})^T\frac{\partial \gamma}{\partial \theta}  \quad & \text{$\gamma$ path}\\
\frac{\partial \beta}{\partial \theta},  \quad & \text{$\beta$ path}\\
\end{aligned}
\right.
\end{equation}
where $\otimes$ represents Kronecker product.


\section{Security Analysis}

\label{sec:analysis}
We investigate the privacy-preserving capability of FedAdOb against feature reconstruction attacks and label inference attacks. 
We conduct the privacy analysis with linear regression models, for the sake of brevity. Proofs are deferred to Appendix C.
\begin{definition} \label{def:SplitFed}
Define the forward function of the bottom model $G$ and the top model $F$: 
\begin{itemize}
    \item For the bottom layer: $H = G(x) =  W_p s_p^\gamma \cdot W_p x + W_p s_p^\beta$;
    \item For the top layer: $y = F(H) =  W_a s_a^\gamma \cdot W_a  H + W_a s_a^\beta$.
\end{itemize}
where $W_p$, $W_a$ are 2D matrices of the bottom and top models; $\cdot$ denotes the inner product, $ s_p^\gamma,  s_p^\beta$ are passports embedding into the bottom layers, $ s_a^\gamma,  s_a^\beta$ are passports embedding into the top layers.
\end{definition}

\subsection{Hardness of Feature Restoration with FedAdOb}
Consider the two strong feature restoration attacks,  White-box Gradient Inversion (WGI) \cite{zhu2019dlg} and White-box Model Inversion (WMI) attack \cite{jin2021cafe,he2019model}, which aims to recover features $\hat{x}$ approximating original features $x$ according to the model gradients and outputs respectively. Specifically, for WMI, the attacker knows the bottom model parameters $W_p$, forward embedding $H$, and the way of embedding passports, but does not know the passport. For WGI, the adversary knows the bottom model gradients $\nabla W_p$, and the way of embedding passports, but does not know the passport.

\begin{theorem}\label{thm:thm1}
    Suppose the client protects features $x$ by inserting the $s_{\beta}^p$. The probability of recovering features by the attacker via WGI and WMI attack is at most $\frac{\pi^{m/2}\epsilon^m}{\Gamma(1+m/2)N^m}$ such that the recovering error is less than $\epsilon$, i.e., $\|x-\hat{x}\|_2\leq \epsilon$, where $m$ denotes the dimension of the passport via flattening, $N$ denotes the passport mean range formulated in Eq. \eqref{eq:sample-pst} and $\Gamma(\cdot)$ denotes the Gamma distribution.
\end{theorem}

Theorem \ref{thm:thm1} demonstrates that the attacker's probability of recovering features within error $\epsilon$ is exponentially small in the dimension of passport size $m$. The successful recovering probability is inversely proportional to the passport mean range $N$ to the power of $m$. 
\subsection{Hardness of Label Recovery with FedAdOb}
Consider the passive model competition attack \cite{fu2022label} that aims to recover labels owned by the active party. The attacker (i.e., the passive party) leverages a small auxiliary labeled dataset $\{x_i, y_i\}_{i=1}^{n_a}$ belonging to the original training data to train the attack model $W_{att}$, and then infer labels for the test data. Note that the attacker knows the trained passive model $G$ and forward embedding $H_i = G(x_i)$. Therefore, the attacker optimizes the attack model $W_{att}$ by minimizing $ \sum_{i=1}^{n_a}\|W_{att}H_i-\vec{y}_i\|_2$\footnote{$\vec{y}_i$ represents the one-hot vector of label $y_i$; we use the mean square error loss for the convenience of analysis.}.
\begin{assumption}\label{assum1}
Suppose the original main algorithm of VFL achieves zero training loss. For the attack model, we assume the error of the optimized attack model $W^*_{att}$ on test data $\tilde{\ell}_t$ is larger than that of the auxiliary labeled dataset $\Tilde{\ell}_a$.
\end{assumption}
\begin{theorem} \label{thm2}
Suppose the active party protects $y$ by embedding $s^a_\gamma$, and adversaries aim to recover labels on the test data with the error $\Tilde{\ell}_t$ satisfying:
    \begin{equation}
        \Tilde{\ell}_t \geq \min_{W_{att}} \sum_{i=1}^{n_a}\|(W_{att}- T_i)H_i\|_2,
    \end{equation}
    where $T_i =diag(W_as_{\gamma,i}^a) W_a$ and $s_{\gamma,i}^a$ is the passport for the label $y_i$ embedded in the active model. Moreover, if $H_{i_1} = H_{i_2} = H$ for any $1\leq i_1,i_2 \leq n_a$, then
    \begin{equation} \label{eq:protecty}
        \Tilde{\ell}_t \geq \frac{1}{(n_a-1)}\sum_{1\leq i_1<i_2\leq n_a}\|(T_{i_1}-T_{i_2})H\|_2.
    \end{equation}
\end{theorem}
\begin{prop}\label{prop1}
Since passports are randomly generated and $W_a$ and $H$ are fixed, if the $W_a = I, H=\Vec{1}$, then it follows that:
\begin{equation}
    \Tilde{\ell}_t \geq \frac{1}{(n_a-1)}\sum_{1\leq i_1<i_2\leq n_a}\|s_{\gamma,i_1}^a-s_{\gamma,i_2}^a\|_2).
\end{equation}
\end{prop}

Theorem \ref{thm2} and Proposition \ref{prop1} show that the label recovery error $\Tilde{\ell}_t$ has a lower bound,
which deserves further explanation. First, when passports are randomly generated for all data, i.e., $s_{\gamma,i_1}^a \neq s_{\gamma,i_2}^a$, then a non-zero label recovery error is guaranteed no matter how adversaries attempt to minimize it. The recovery error thus acts as a protective random noise imposed on true labels. Second, the magnitude of the recovery error monotonically increases with the variance $\sigma^2$ of the Gaussian distribution passports sample from (in Eq. \eqref{eq:sample-pst}), which is a crucial parameter to control privacy-preserving capability (see Experiment results in Sect. \ref{sec:aba-tradeoff}) are in accordance with Theorem \ref{thm2}. 
Third, it is worth noting that the lower bound is based on the training error of the auxiliary data used by adversaries to launch PMC attacks. Given possible discrepancies between the auxiliary data and private labels, e.g., in terms of distributions and the number of dataset samples, the actual recovery error of private labels can be much larger than the lower bound. Again, this strong protection is observed in experiments (see Sect. \ref{sec:exp}). 


\section{Experiment}\label{experiment}
We present empirical studies of FedAdOb on various datasets using different model architectures.

\subsection{Experiment Setting}\label{sec:setup}
\subsubsection{Models \& Datasets} 
We conduct experiments on four image datasets and one tabular dataset:
\textit{MNIST} \cite{lecun2010mnist}, \textit{CIFAR10}, \textit{CIFAR100} \cite{krizhevsky2009learning}, ModelNet \cite{wu20153d} and Criteo \cite{wang2017deep}. 
In HFL, we adopt LeNet \cite{lecun1998gradient} on MNIST, \textit{AlexNet} \cite{NIPS2012_c399862d} on CIFAR10, \textit{ResNet18} \cite{he2016deep} on CIFAR100. 
In VFL, we adopt LeNet on MNIST and ModelNet, adopt AlexNet and ResNet18 on CIFAR10, and Deep \& Cross Network (DCN) \cite{wang2017deep} on Criteo. 

\subsubsection{Federated Learning Settings}
We partition a neural network into a bottom model and a top model for both HFL and VFL settings. In HFL, each client incorporates their passport information into the bottom model to safeguard the features. In VFL, the passive party embeds the passport information into the last layers of the bottom model to protect the features, while the active party inserts the passport information into the last layers of the top model to protect the labels. In HFL, each client possesses both private features and labels. In VFL, the passive party exclusively contributes private features, while the active party solely provides labels. The details of our HFL and VFL scenarios are summarized in TABLE \ref{table:models}. Please refer to Appendix A for details on the experimental settings.

\begin{table}[!h]
	\caption{Models for evaluation in HFL and VFL. \# P denotes the number of parties. FC: fully-connected layer. Conv: convolution layer. }
	\centering
	\footnotesize
	\begin{tabular}{c||c|c|c|c}
	        \hline
             \\[-1em]
		Scenario& \shortstack{Model \& \\ Dataset}  & \shortstack{Bottom \\ Model } & \shortstack{Top \\ Model}   & \# P \\
         \hline
         \hline
           \\[-1em]
            \multirow{3}{*}{HFL}&LeNet-MNIST & 1 Conv  &  2 Conv+ 3 FC  & 2 \\
		\cline{2-5}
       \\[-1em]
            & AlexNet-CIAFR10 & 1 Conv  &  4 Conv+ 1 FC  & 2 \\
		\cline{2-5}
		\\[-1em]
            & ResNet-CIFAR100 & 1 Conv  &   16 Conv+ 1 FC   & 2 \\
		\hline
        \hline
          \\[-1em]
              \multirow{5}{*}{VFL}& LeNet-MNIST & 2 Conv  &  3 FC  & 2 \\
		\cline{2-5}
          \\[-1em]
		&AlexNet-CIFAR10 & 5 Conv  &  1 FC & 2  \\
		\cline{2-5}
          \\[-1em]
		&ResNet18-CIAFR10 & 17 Conv & 1 FC & 2\\
		\cline{2-5}
     \\[-1em]
     	&LeNet-ModelNet & 2 Conv  &  3 FC  & 7 \\
		\cline{2-5}
       \\[-1em]
     	&DCN-Criteo  & 3 FC  &  1 FC  & 2 \\
		\hline
         
	\end{tabular}
\label{table:models}
\end{table}

\subsubsection{Privacy Attack Methods}
We consider privacy attacks in both HFL and VFL settings. For \textbf{HFL}, we investigate the effectiveness of FedAdOb against White-box and Black-box gradient inversion (WGI and BGI) attacks and White-box \cite{zhu2019dlg} and Black-box model inversion (WMI and BMI) attacks \cite{he2019model}.
For \textbf{VFL}, we evaluate the effectiveness of FedAdOb against feature reconstruction attacks (including CAFE attack~\cite{he2019model,jin2021cafe} and Model Inversion (MI) attack \cite{he2019model} ) and label inference attacks (including  Passive Model Completion (PMC) attack \cite{fu2022label}, Norm-based scoring attack and Direction-based scoring attack \cite{oscar2022split}). The details of the attacks are shown in Appendix A.

\subsubsection{Baseline Defense Methods} 
In both HFL and VFL scenarios, we compare FedAdOb with FedAVG \cite{mcmahan2017communication}, the baseline with no defense, and other three defense methods including \textbf{Differential Privacy (DP)} \cite{abadi2016deep}, \textbf{Sparse} \cite{lin2018deep}, and \textbf{InstaHide} \cite{huang2020instahide}. Besides, \textbf{SplitFed} \cite{thapa2020splitfed} is used for comparison with FedAdOb in HFL and \textbf{Confusional AutoEncoder} (CAE)~\cite{zou2022defending}, \textbf{Marvell} \cite{oscar2022split}, \textbf{GradPerturb} \cite{yang2022differentially} and \textbf{LabelDP} \cite{ghazi2021deep} are additional baseline defense methods in VFL (see details in Appendix). 

\subsubsection{Evaluation Metrics} 
We use data (feature or label) recovery error and main task accuracy \footnote{AUC is used for binary classification dataset Criteo.} to evaluate defense mechanisms. We adopt the ratio of incorrectly labeled samples by a label inference attack to all labeled samples to measure the performance of that label inference attack. We adopt Mean Square Error (MSE) \cite{zhu2019dlg} between original images and images recovered by a feature reconstruction attack to measure the performance of that feature reconstruction attack. MSE is widely used to assess the quality of recovered images. A higher MSE value indicates a higher image recovery error. In addition, we leverage Calibrated Averaged Performance (CAP) \cite{fan2020rethinking} to quantify the trade-off between main task accuracy and data recovery error. CAP is defined as follows:
\begin{definition}[Calibrated Averaged Performance (CAP)] \label{def:cap}
For a given Privacy-Preserving Mechanism $g_s\in \calG$ ($s$ denotes the controlled parameter of $g$, e.g., the sparsification level, noise level, and passport range) and attack mechanism $a \in \calA$, the Calibrated Averaged Performance is defined as:
\begin{equation}
    \text{CAP}(g_s, a) = \frac{1}{m}\sum_{s=s_1}^{s_m} Acc(g_s, x) * Rerr(x, \hat{x}_s),
\end{equation}
where $Acc(\cdot)$ denotes the main task accuracy and $Rerr(\cdot)$ denotes the recovery error between original data $x$ and estimated data $\hat{x}_s$ via attack $a$.
\end{definition}



\subsection{Comparison with Other Defending Methods} \label{sec:exp}
\subsubsection{Defending against Feature Restoration Attack in HFL}

We evaluate FedAdOb's performance by comparing the averaged clients' accuracy with 5 baselines illustrated in Sect. \ref{sec:setup}. Fig. \ref{fig:star} compares FedAdOb with 4 defense methods in terms of their trade-offs between main task accuracy and data recovery error against WGI and WMI attacks. We have the following three observations: 
\begin{itemize}
    \item FedAdOb facilitates excellent trade-offs between main task accuracy and privacy protection. For instance, model accuracy degradations of FedAdOb are less than 2\% for all 6 scenarios shown in Fig. \ref{fig:star} while attackers fail to recover original data (with recovery errors more than 0.045 for WGI attacks and 0.1 for WMI attacks). In contrast, although it achieves high model accuracy, FedAVG suffers from significant privacy leakage.
    \item The trade-off curves of FedAdOb (dotted dashed line in blue) are near the optimal trade-off towards the top-right corner, and it outperforms all baseline defense methods by large margins.
    \item Among all baseline defense methods, InstaHide tends to sacrifice a great deal of model accuracy for high privacy-preserving capability. More specifically, for InstaHide, the best cases of performance degradation on CIFAR10 and CIFAR100 are 18.6\% and 10.66\%, respectively, while the worst cases reach 52.12\% and 16.65\%, respectively, when mixing up with 4 samples. On the other hand, DP and Sparsification can achieve high model performance but at the cost of high privacy leakage.
\end{itemize}  

In addition, we compare FedAdOb with SplitFed against BGI and BMI attacks. TABLE \ref{tab:blackbox_attack} reports the results, showing that FedAdOb outperforms SplitFed by a large data recovery error on MNIST, CIFAR10, and CIFAR100.

\begin{figure*}
\centering
	\centering
      		\begin{subfigure}{0.3\textwidth}
  		 	\includegraphics[width=1\textwidth]{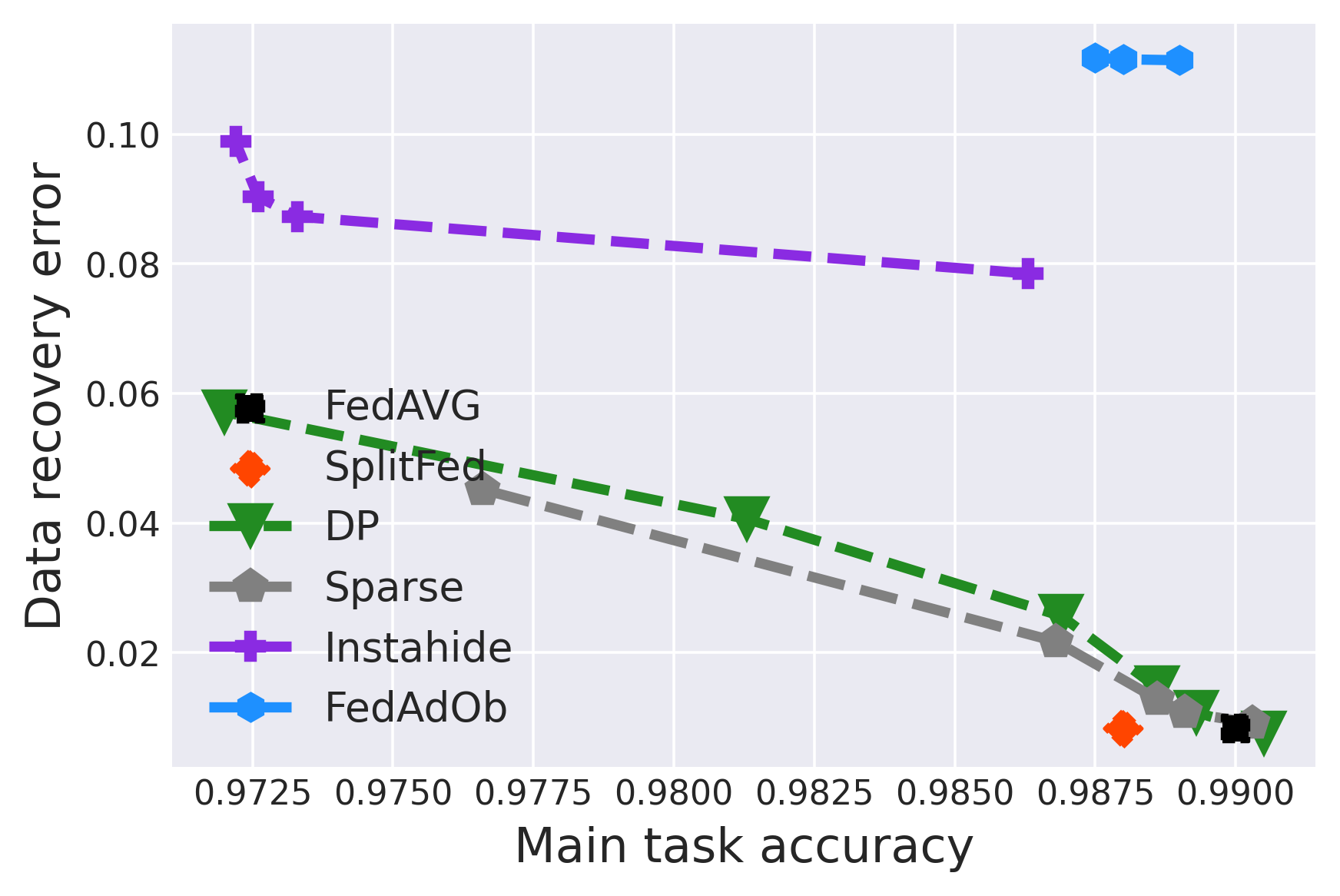}
      \subcaption{LeNet-MNIST}
    		\end{subfigure}
    	\begin{subfigure}{0.3\textwidth}
  		 	\includegraphics[width=1\textwidth]{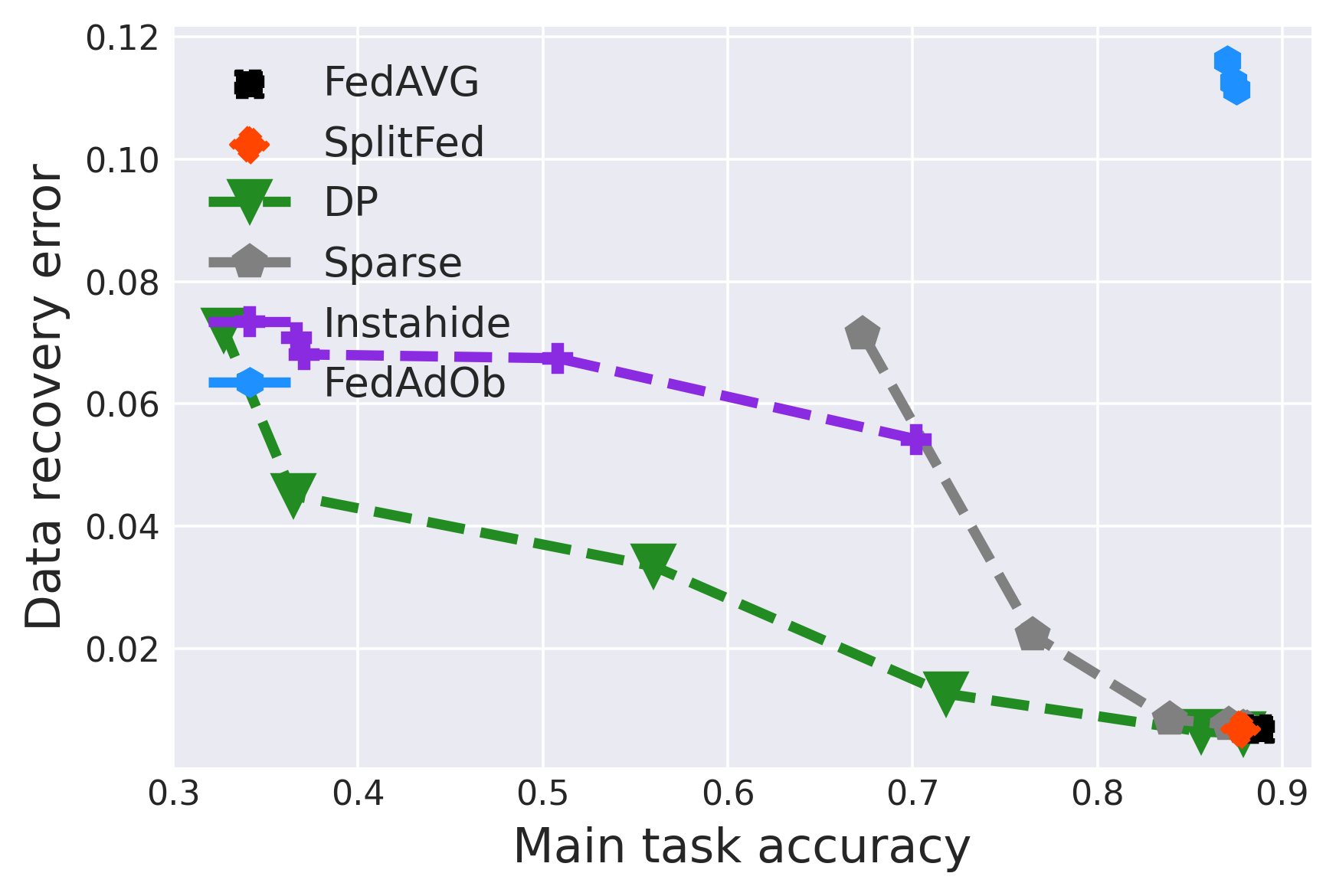}
            \subcaption{AlexNet-CIFAR10}
    		\end{subfigure}
   \begin{subfigure}{0.3\textwidth}
			\includegraphics[width=1\textwidth]{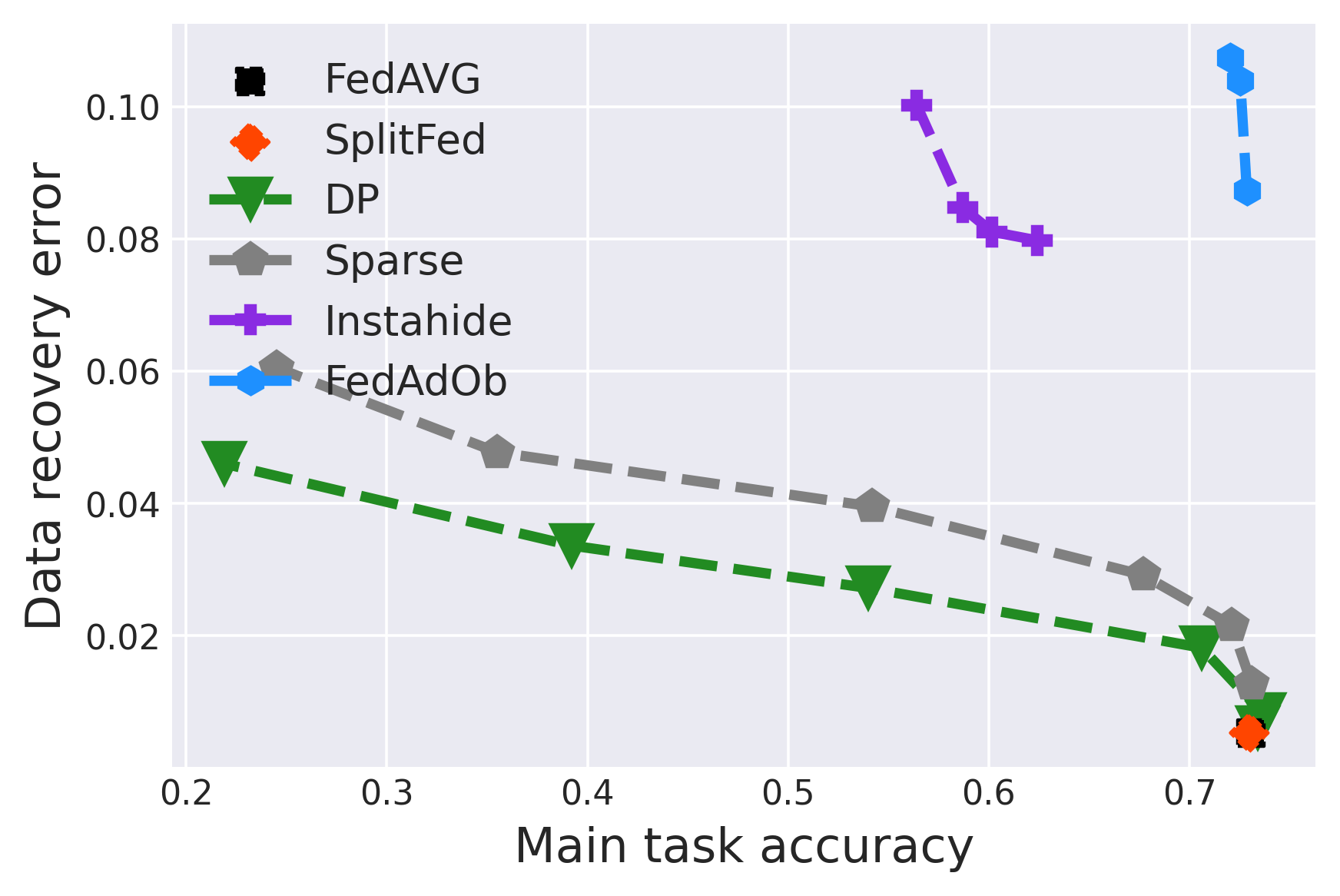}
         \subcaption{ResNet-CIFAR100}
		\end{subfigure}

      		\begin{subfigure}{0.3\textwidth}
  		 	\includegraphics[width=1\textwidth]{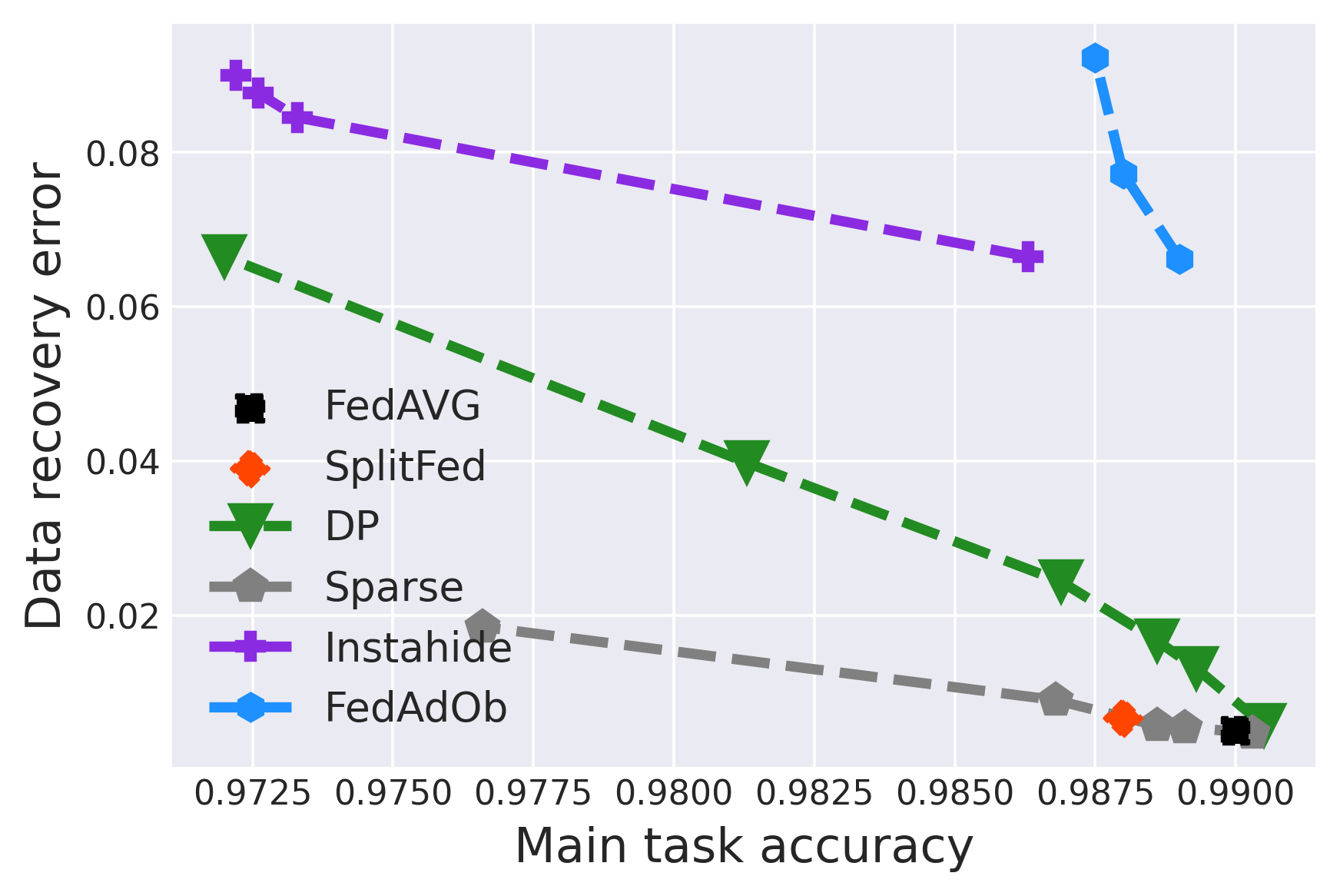}
      \subcaption{LeNet-MNIST}
    		\end{subfigure}
    	\begin{subfigure}{0.3\textwidth}
  		 	\includegraphics[width=1\textwidth]{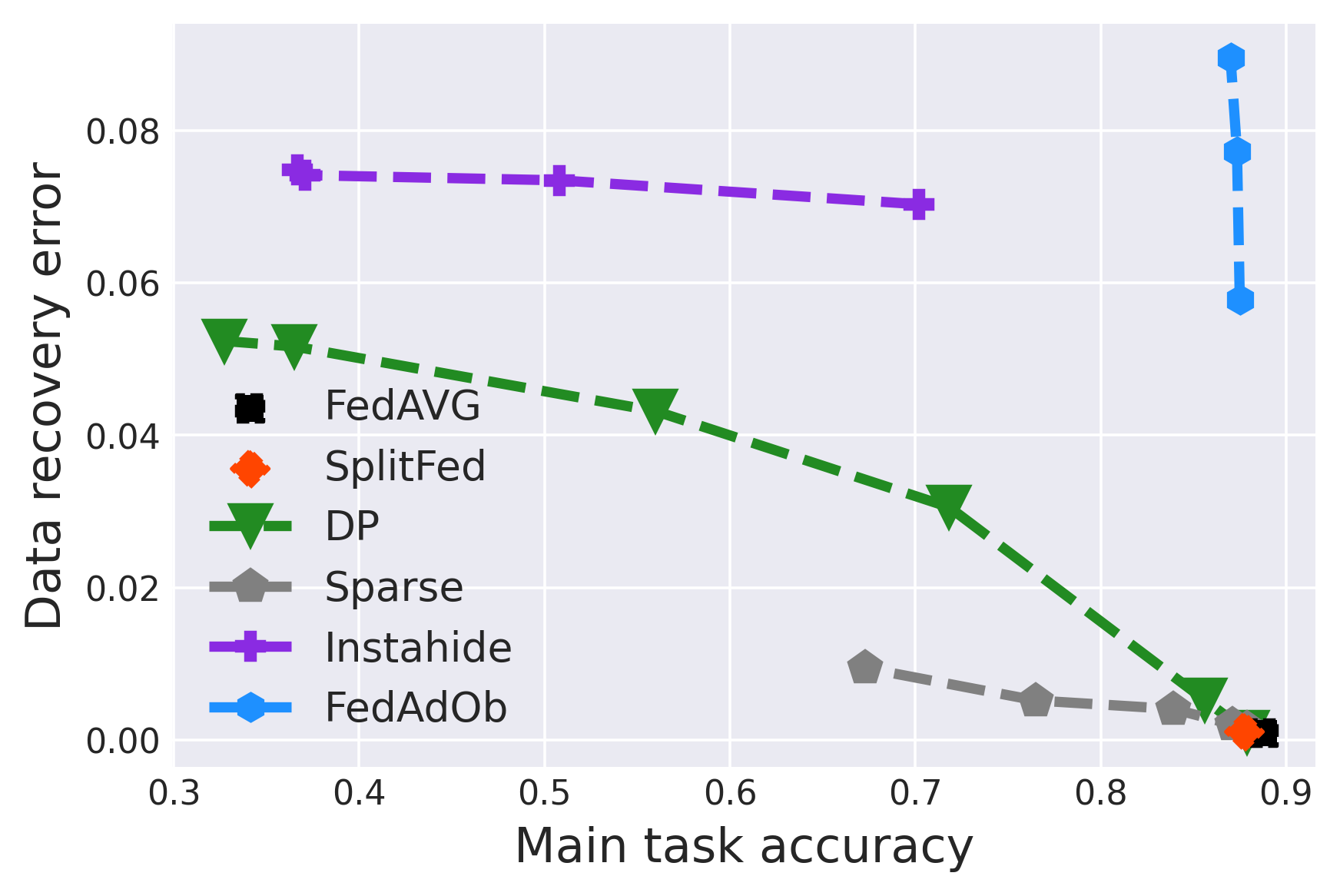}
            \subcaption{AlexNet-CIFAR10}
    		\end{subfigure}
   \begin{subfigure}{0.3\textwidth}
			\includegraphics[width=1\textwidth]{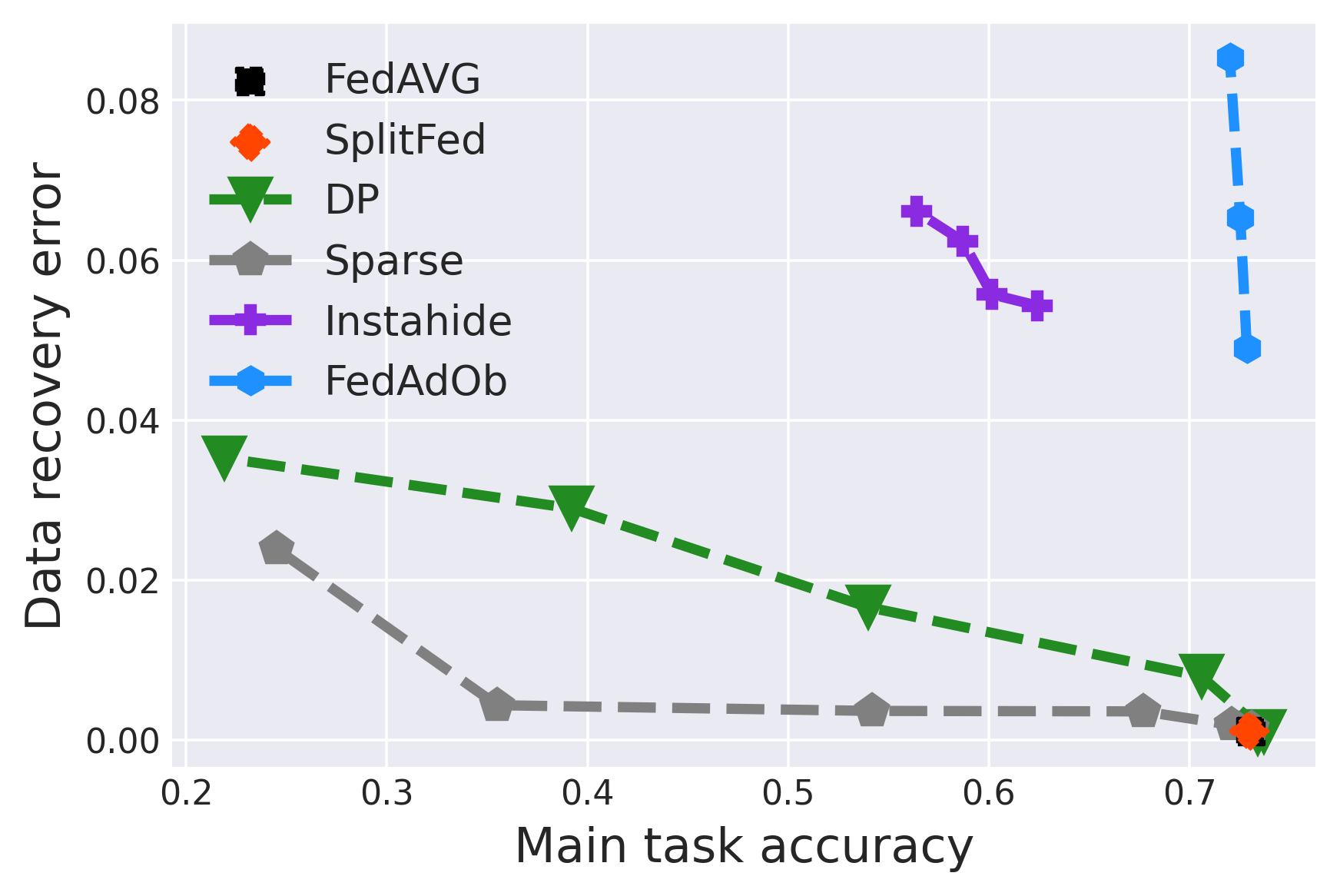}
         \subcaption{ResNet-CIFAR100}
		\end{subfigure}
\centering
\caption{\textbf{HFL tradeoff.} Comparison of different defense methods in terms of their trade-offs between main task accuracy and data recovery error against \textbf{WGI} attack \cite{zhu2019dlg} (the first line) and \textbf{WMI} attack \cite{he2019model} (the second line) on LeNet-MNIST, AlexNet-CIFAR10, and ResNet-CIFAR100, respectively. \textit{Trade-off curves near the top right corner are preferred to faraway ones.}}
\label{fig:star}
\end{figure*}

\begin{table}[!htbp]
\caption{Data recovery error of SplitFed and FedAdOb ($N=2$) under \textbf{BGI} and \textbf{BMI} attacks in \textbf{HFL} setting.}
\centering
\footnotesize
\begin{tabular}{c|c|c|c|c}
\hline
Method                    & Attack & MNIST & CIFAR10 & CIFAR100 \\ \hline
\multirow{2}{*}{SplitFed} 
                          & BGI       & 0.111$\pm$0.002 & 0.110$\pm$0.003  & 0.087$\pm$0.002   \\ \cline{2-5} 
                          & BMI       & 0.052$\pm$0.005 & 0.028$\pm$0.004  & 0.027$\pm$0.004   \\ \hline
\multirow{2}{*}{FedAdOb}    
                          & BGI       & 0.111$\pm$0.002 & 0.116$\pm$0.002  & 0.107$\pm$0.001   \\ \cline{2-5}
                          & BMI       & 0.071$\pm$0.003 & 0.057$\pm$0.004  & 0.055$\pm$0.003   \\ \hline
\end{tabular}
\label{tab:blackbox_attack}
\end{table}

\subsubsection{Defending against the Feature Reconstruction Attack in VFL}

The trade-off between feature recovery error (y-axis) and main task accuracy (x-axis) of FedAdOb, along with baseline methods, is compared in the first and second columns of Fig. \ref{fig:tradeoff_result} against BMI and WMI attacks on four different models. We further compare the recovered images against WMI on Fig. \ref{fig:vis-whitebox}. The following observations are made:
\begin{itemize}
    \item Differential Privacy (DP) and Sparsification methods exhibit a trade-off between high main task performance and low feature recovery error (indicating low privacy leakage). For instance, ResNet-CIFAR10 achieves a main task performance of $\geq 0.90$ and a feature recovery error as low as $\leq 0.06$ with DP and Sparsification. On the other hand, these methods can achieve a feature recovery error of $\geq 0.11$ but obtain a main task performance of $\leq 0.70$ for ResNet-CIFAR10.
    \item InstaHide is generally ineffective against BMI and WMI attacks. Even when more data is mixed in, InstaHide still results in a relatively small feature recovery error but significantly degrades the main task performance.
    \item The trade-off curves of FedAdOb are located near the top-right corner under both attacks for all models. This indicates that FedAdOb performs best in preserving feature privacy while maintaining the model performance. For example, under MI and CAFE attacks on ResNet-CIFAR10, FedAdOb achieves a main task accuracy of $\geq 0.91$ and a feature recovery error of $\geq 0.12$. The results presented in Table 1 also demonstrate that FedAdOb offers the best trade-off between privacy and performance under BMI and WMI attacks. 
    \item The reconstructed images under the WMI attack appear as random noise, indicating the successful mitigation of the WMI attack by FedAdOb. Furthermore, the model performance of FedAdOb is nearly indistinguishable from the original model (without defense) across all datasets, as shown in Fig. \ref{fig:tradeoff_result}. This demonstrates the superior trade-off achieved by FedAdOb, in stark contrast to existing methods. Notably, approaches such as Differential Privacy (DP) and Sparsification, even at high protection levels ($r5$ and $r7$), result in a significant deterioration of model performance compared to the original model.
\end{itemize}

\subsubsection{Defending against the Label Inference Attack in VFL}

\begin{figure*}[!h]
	\centering
      \begin{subfigure}{0.3\textwidth}
  		\includegraphics[width=1\textwidth]{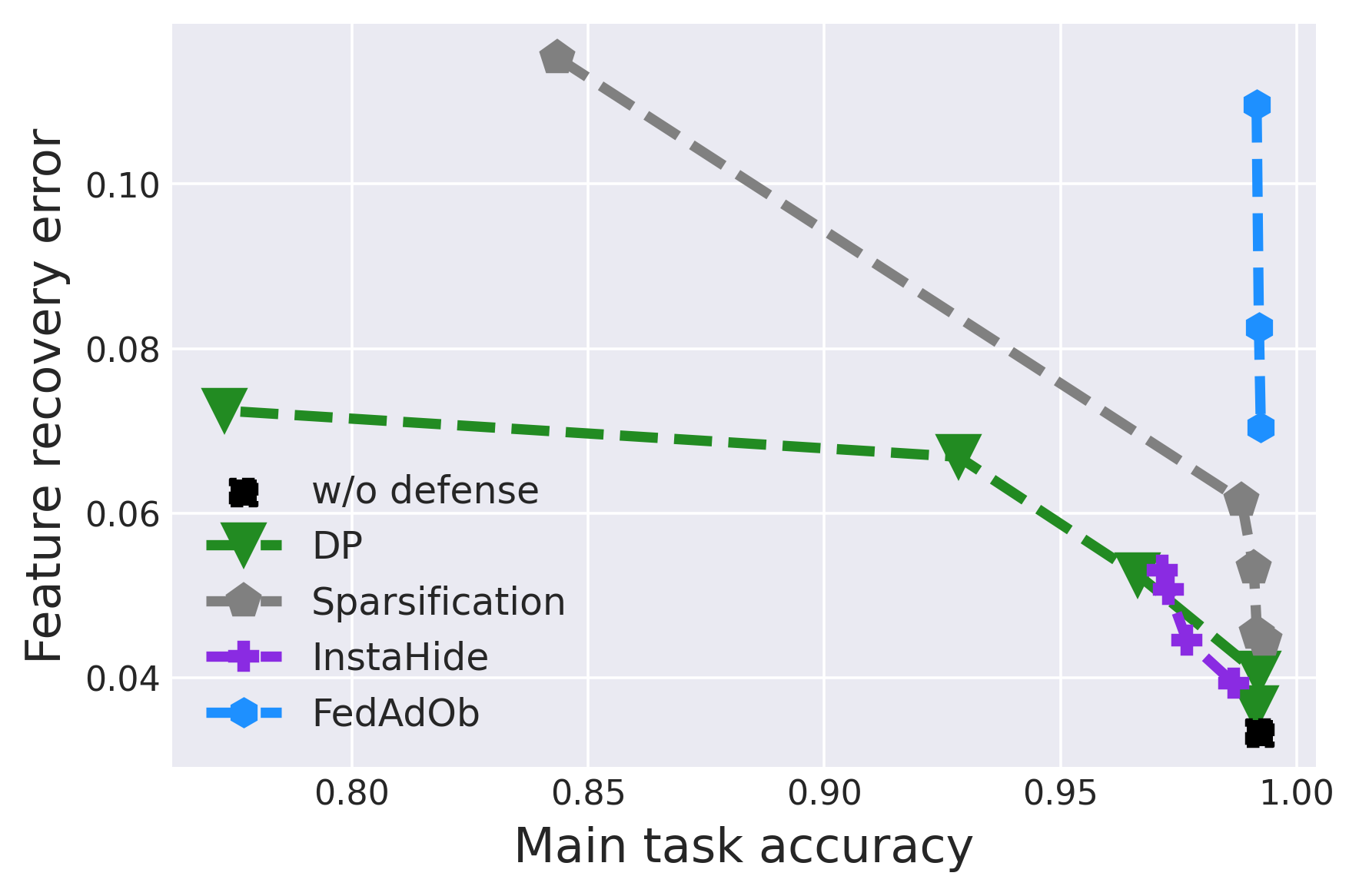}
    \subcaption{LeNet-MNIST}
    		\end{subfigure}
    	\begin{subfigure}{0.3\textwidth}
  		 \includegraphics[width=1\textwidth]{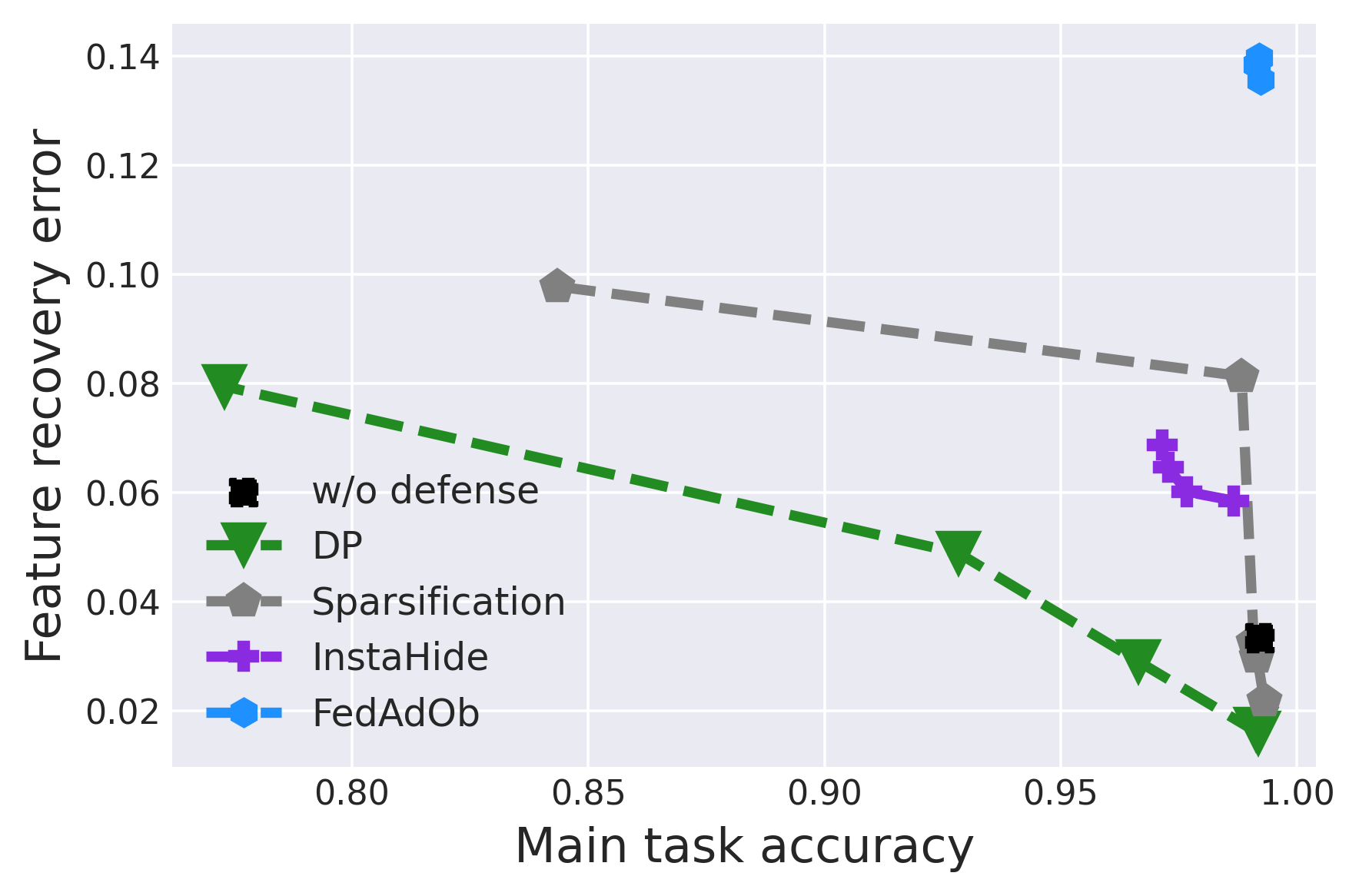}    \subcaption{LeNet-MNIST}
    		\end{subfigure}
   \begin{subfigure}{0.3\textwidth}
		\includegraphics[width=1\textwidth]{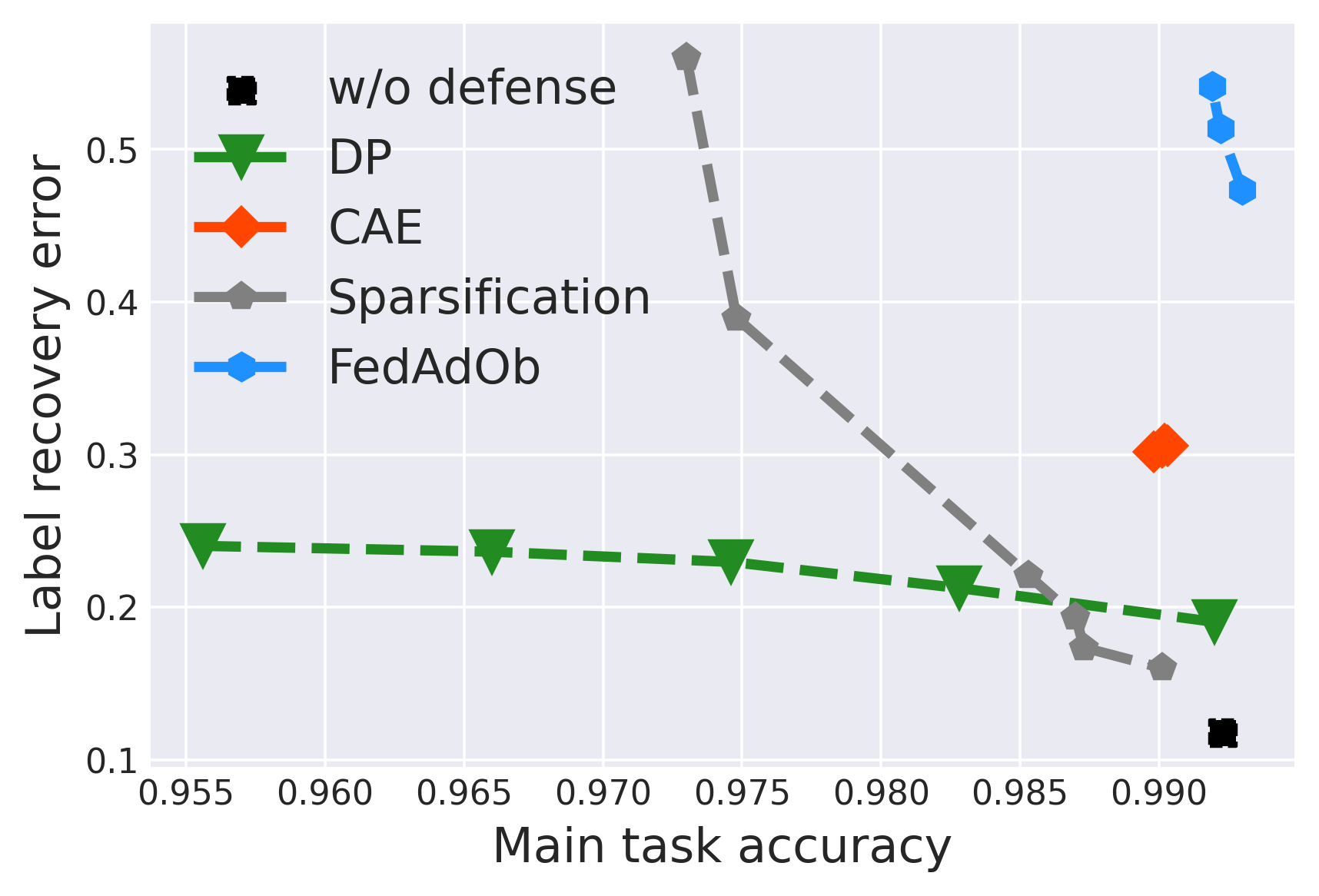}    \subcaption{LeNet-MNIST}
		\end{subfigure}

      \begin{subfigure}{0.3\textwidth}
  		 \includegraphics[width=1\textwidth]{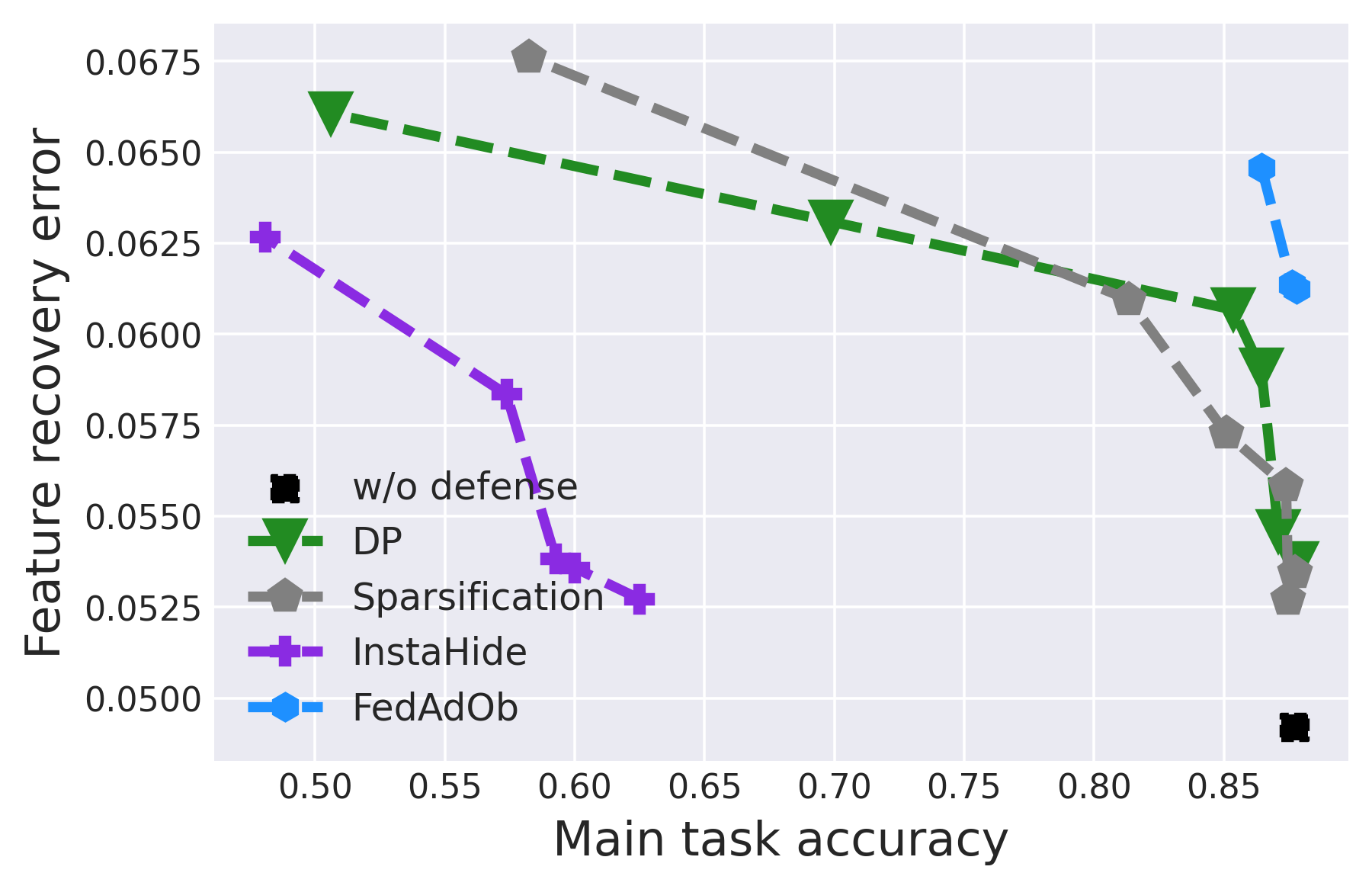}  \subcaption{AlexNet-CIFAR10}
    		\end{subfigure}
    	\begin{subfigure}{0.3\textwidth}
  		 \includegraphics[width=1\textwidth]{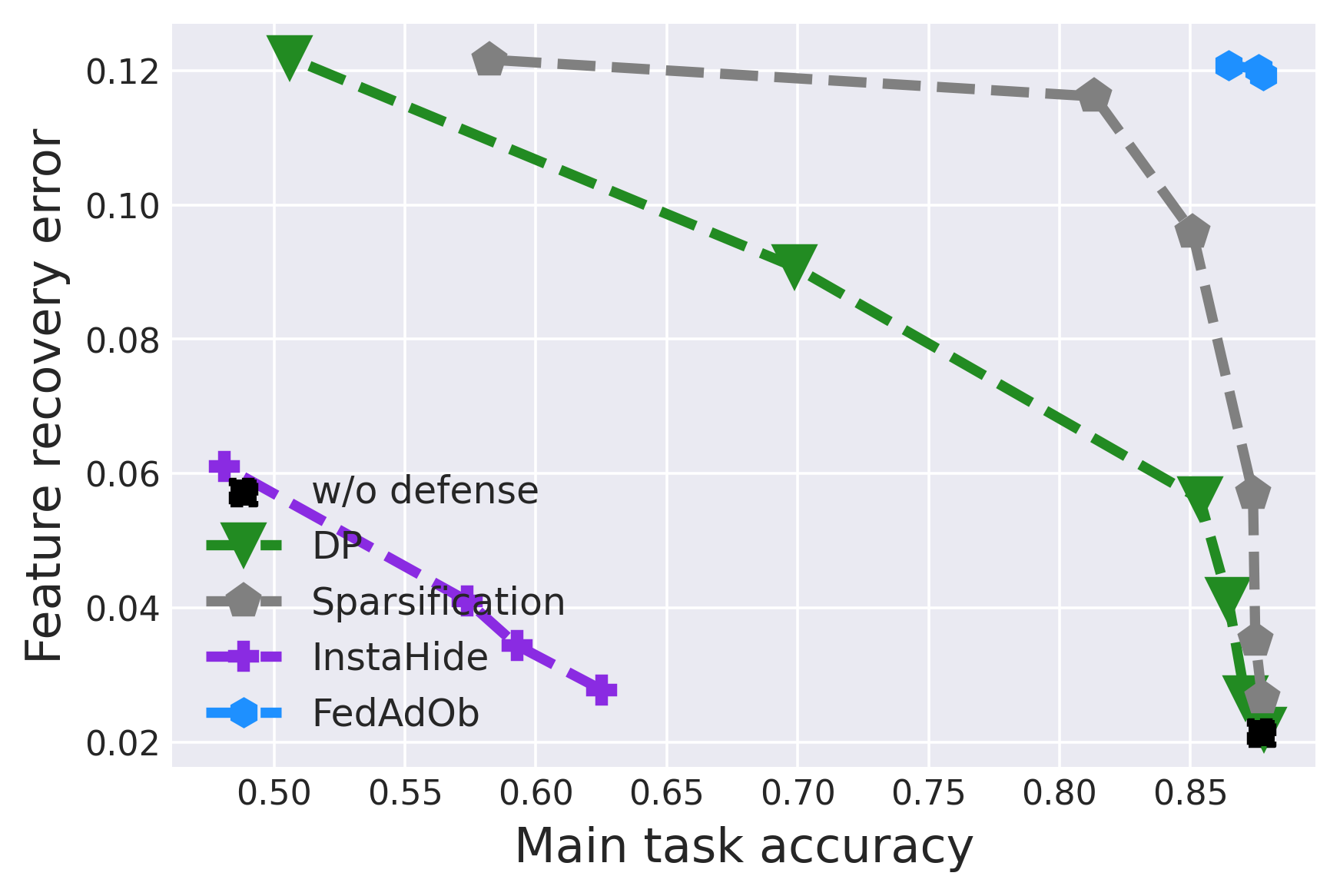}\subcaption{AlexNet-CIFAR10}
    		\end{subfigure}
   \begin{subfigure}{0.3\textwidth}
		\includegraphics[width=1\textwidth]{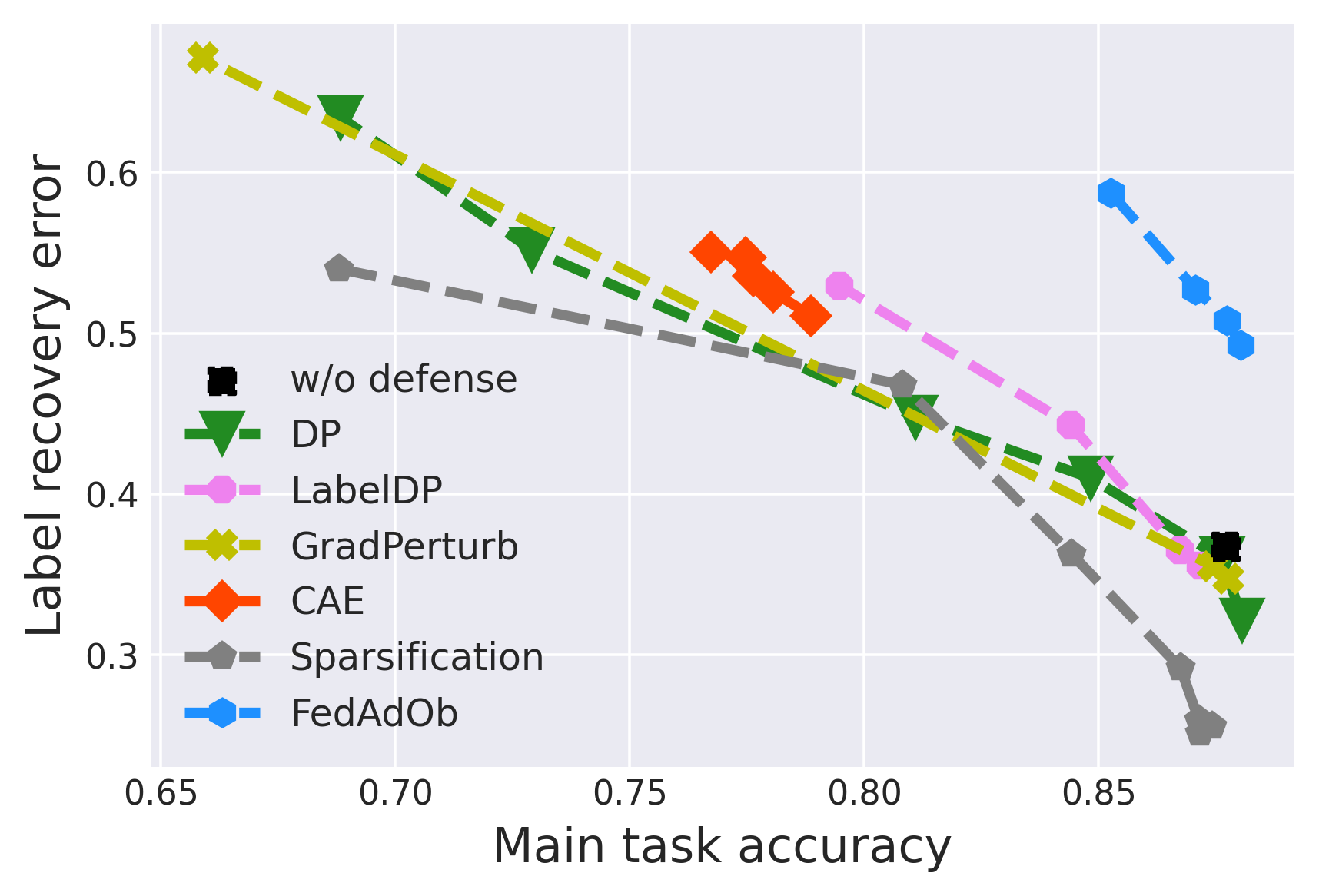}\subcaption{AlexNet-CIFAR10}
		\end{subfigure}

      \begin{subfigure}{0.3\textwidth}
  		 \includegraphics[width=1\textwidth]{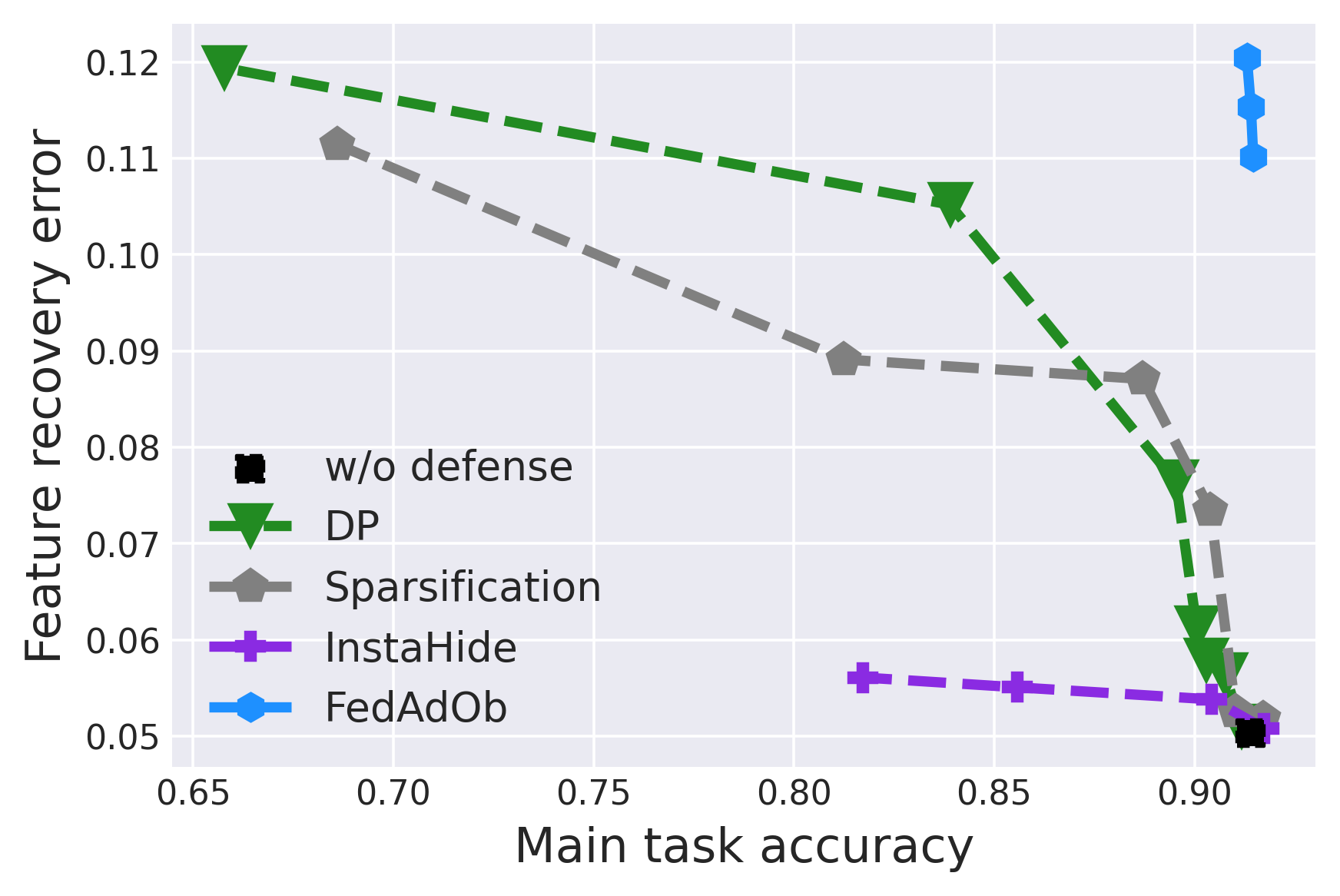}\subcaption{ResNet-CIFAR10}
    		\end{subfigure}
    	\begin{subfigure}{0.3\textwidth}
  		 \includegraphics[width=1\textwidth]{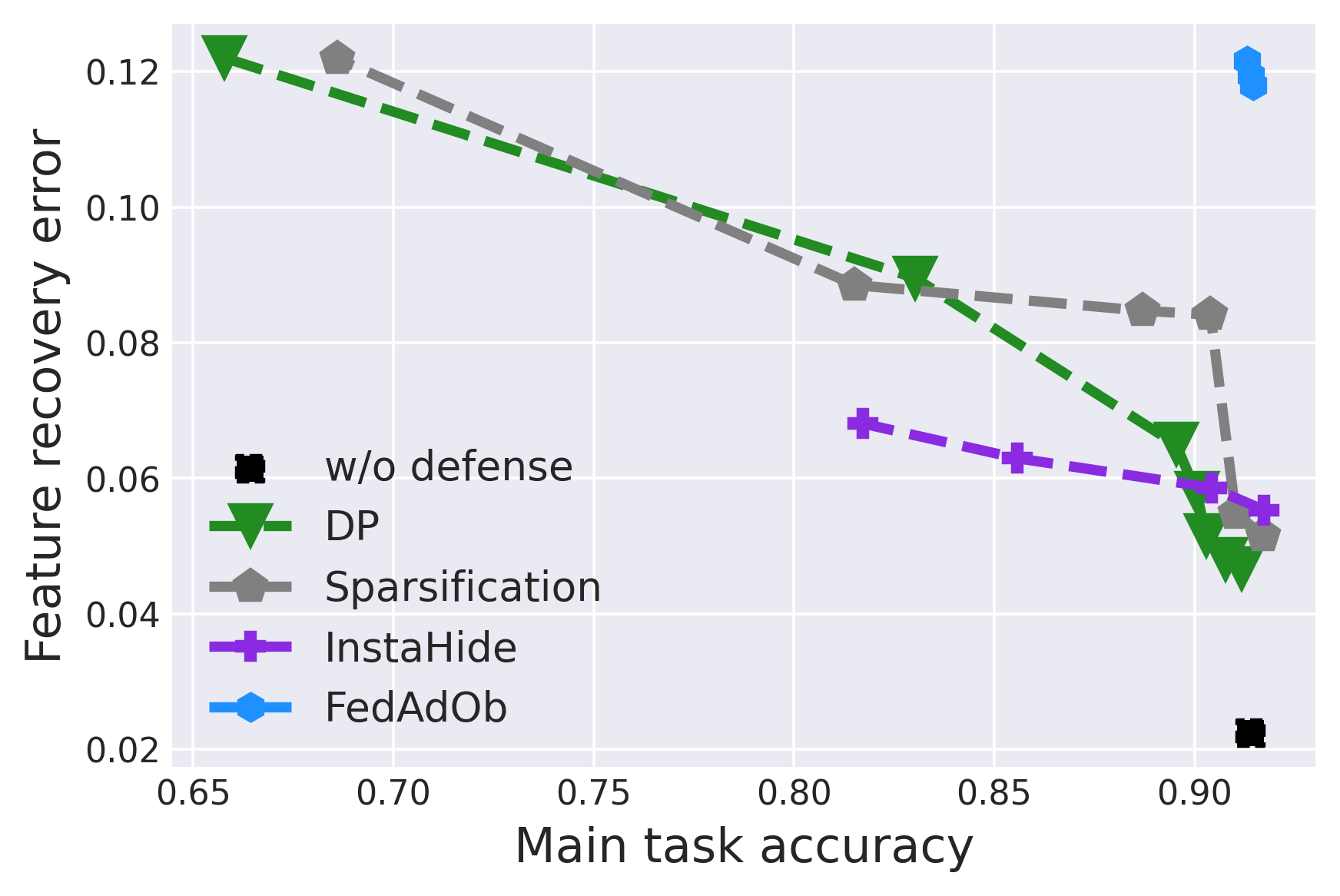}\subcaption{ResNet-CIFAR10}
    		\end{subfigure}
   \begin{subfigure}{0.3\textwidth}
		\includegraphics[width=1\textwidth]{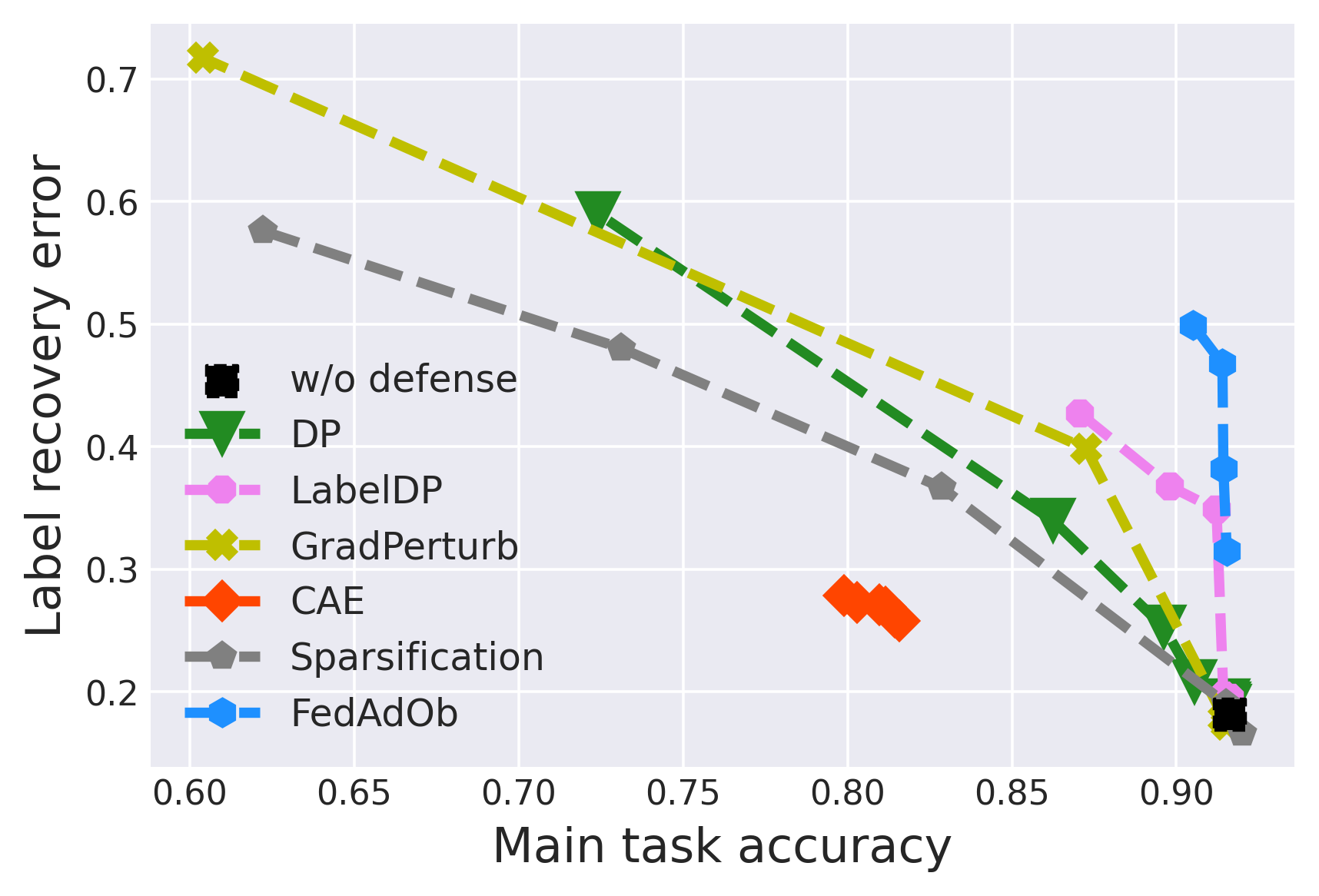}\subcaption{ResNet-CIFAR10}
		\end{subfigure}

     \begin{subfigure}{0.3\textwidth}
  		 \includegraphics[width=1\textwidth]{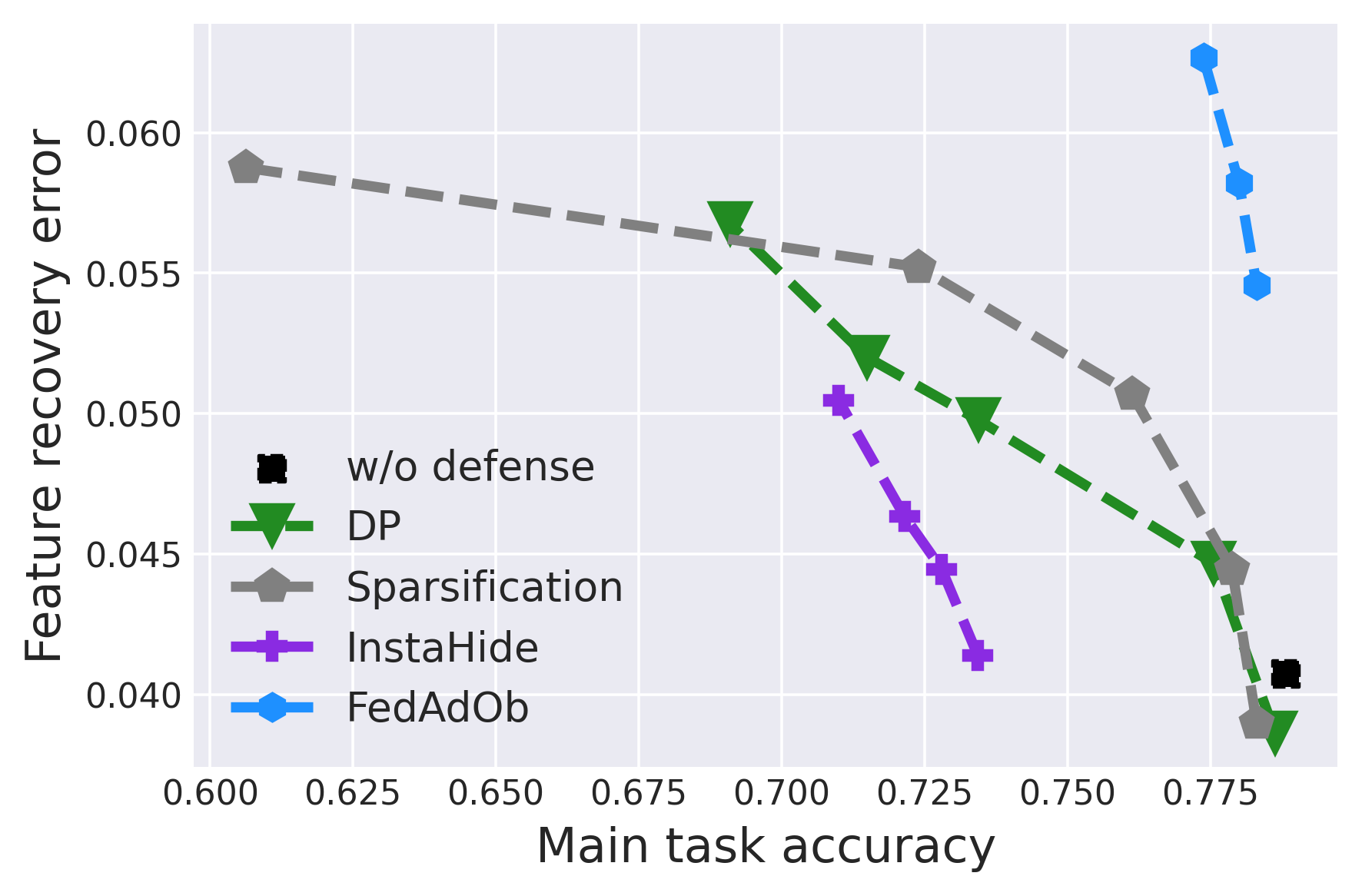}
      \subcaption{LeNet-ModelNet}
    		\end{subfigure}
    	\begin{subfigure}{0.3\textwidth}
  		 \includegraphics[width=1\textwidth]{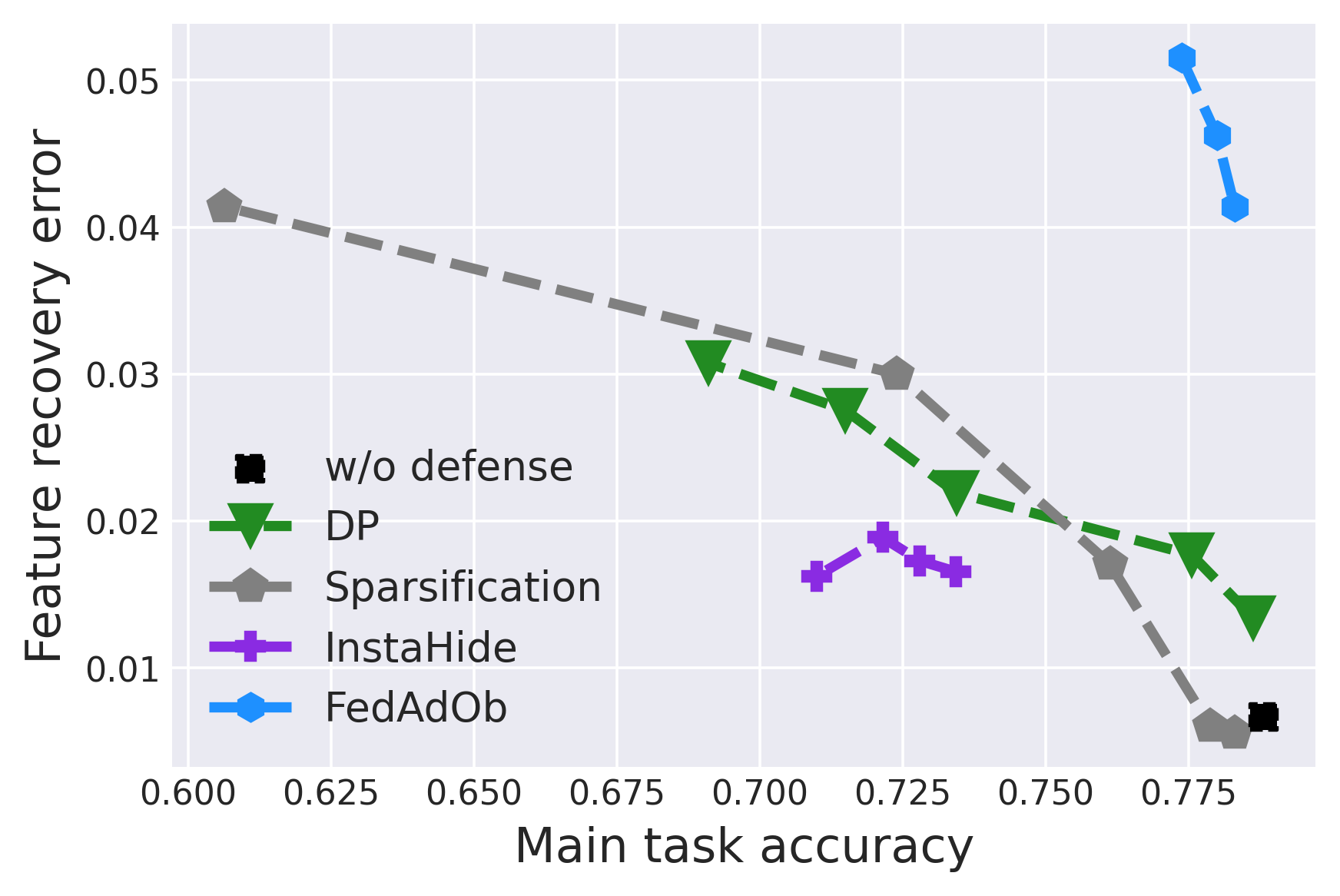}
            \subcaption{LeNet-ModelNet}
    		\end{subfigure}
   \begin{subfigure}{0.3\textwidth}
		\includegraphics[width=1\textwidth]{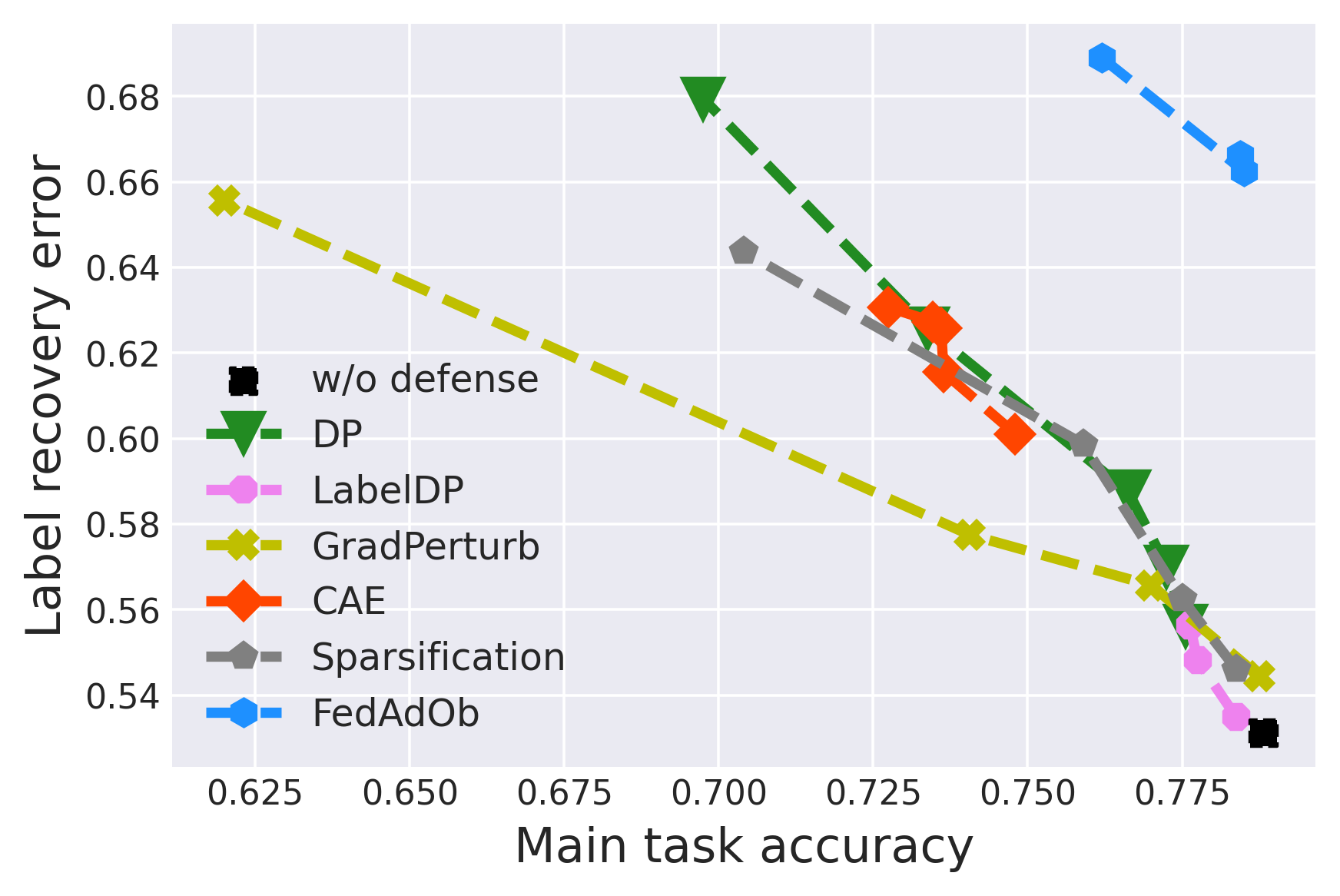}
         \subcaption{LeNet-ModelNet}
		\end{subfigure}
 \caption{\textbf{VFL tradeoff.} Comparison of different defense methods in terms of their trade-offs between main task accuracy and data (feature or label) recovery error against three attacks on LeNet-MNIST, AlexNet-CIFAR10, ResNet-CIFAR10, respectively and LeNet-ModelNet. \textbf{BMI} (the first column) and \textbf{WMI} (the second column) are feature reconstruction attacks, whereas \textbf{Passive Model Completion} (the third column) is a label inference attack. \textit{A better trade-off curve should be more toward the top-right corner of each figure}.}
	\label{fig:tradeoff_result}
\end{figure*}


\begin{table*}[!ht]
\caption{The Calibrated averaged performance (CAP) for different defense mechanisms against BMI, WMI, and PMC attacks in \textbf{VFL}.}
\center
\footnotesize
\setlength{\tabcolsep}{1.5mm}
\begin{tabular}{c|c||c|c|c|c|c|c|c|c}
\hline
\multicolumn{2}{c||}{\diagbox[dir=SE]{Attack}{Defense}} & w/o defense & CAE             &GradPerturb  & LabelDP       & Sparsification  & DP              & InstaHide       & FedAdOb       \\ \hline \hline
\multirow{3}{*}{CAFE}                    &LeNet-MNIST        & 0.033       & \textemdash     & \textemdash & \textemdash   & 0.049$\pm$0.026 & 0.033$\pm$0.018 & 0.061$\pm$0.004 & \textbf{0.137$\pm$0.002} \\ 
                                         & AlexNet-CIFAR10      & 0.019       & \textemdash     & \textemdash & \textemdash   & 0.058$\pm$0.026 & 0.042$\pm$0.017 & 0.023$\pm$0.004 & \textbf{0.105$\pm$0.001} \\
                                         & ResNet-CIFAR10       & 0.021       & \textemdash     & \textemdash & \textemdash   & 0.067$\pm$0.014 & 0.057$\pm$0.014 & 0.053$\pm$0.002 & \textbf{0.109$\pm$0.001} \\ 
                                         & LeNet-ModelNet       & 0.005     & \textemdash     & \textemdash & \textemdash   & 0.014$\pm$0.009 & 0.016$\pm$0.004 & 0.012$\pm$0.001 & \textbf{0.036$\pm$0.003} \\\hline
\multirow{3}{*}{MI}                      &LeNet-MNIST        & 0.033       & \textemdash     & \textemdash & \textemdash   & 0.060$\pm$0.020 & 0.049$\pm$0.010 & 0.046$\pm$0.005 & \textbf{0.087$\pm$0.001} \\
                                         & AlexNet-CIFAR10      & 0.043       & \textemdash     & \textemdash & \textemdash   & 0.047$\pm$0.003 & 0.046$\pm$0.006 & 0.032$\pm$0.001 & \textbf{0.054$\pm$0.001} \\
                                         & ResNet-CIFAR10       & 0.046       & \textemdash     & \textemdash & \textemdash   & 0.065$\pm$0.012 & 0.063$\pm$0.015 & 0.047$\pm$0.001 & \textbf{0.105$\pm$0.004} \\ 
                                         & LeNet-ModelNet       &  0.032     & \textemdash     & \textemdash & \textemdash   & 0.035$\pm$0.003 & 0.036$\pm$0.003 & 0.033$\pm$0.002 & \textbf{0.046$\pm$0.002} \\\hline
\multirow{3}{*}{PMC}                      &LeNet-MNIST      & 0.117        & 0.302$\pm$0.002 & 0.195$\pm$0.008 & 0.314$\pm$0.103 & 0.277$\pm$0.140  & 0.216$\pm$0.015 & \textemdash  & \textbf{0.506$\pm$0.028} \\
                                         & AlexNet-CIFAR10     & 0.322        & 0.415$\pm$0.008 & 0.353$\pm$0.063 & 0.355$\pm$0.045 & 0.283$\pm$0.065  & 0.358$\pm$0.051 & \textemdash  & \textbf{0.460$\pm$0.025} \\
                                         & ResNet-CIFAR10      & 0.166        & 0.217$\pm$0.004 & 0.276$\pm$0.119 & 0.276$\pm$0.081 & 0.268$\pm$0.088  & 0.237$\pm$0.087 & \textemdash  & \textbf{0.379$\pm$0.065}  \\ 
                                         & LeNet-ModelNet      & 0.419       & 0.457$\pm$0.004    & 0.425$\pm$0.011 & 0.426$\pm$0.005 & 0.443$\pm$0.011 & 0.451$\pm$0.015 & \textemdash & \textbf{0.523$\pm$0.002} \\\hline
\end{tabular}
\label{tab:cap}
\end{table*}

The third column of Fig. \ref{fig:tradeoff_result} (g)-(i) compares the trade-offs between label recovery error (y-axis) and main task accuracy (x-axis) of FedAdOb with those of baselines in the face of the PMC attack on four models. It is observed DP and its variants (GradPerturb, Label DP), Sparsification, and CAE fail to achieve the goal of obtaining a low level of privacy leakage while maintaining model performance, whereas FedAdOb is more toward the top-right corner, indicating that FedAdOb has a better trade-off between privacy and performance. TABLE \ref{tab:cap} reinforces the observation that FedAdOb achieves the best trade-off between privacy and performance under PMC attack.

Furthermore, in the case of NS and DS attacks, which are specifically applied in binary classification problems, we compare the performance of FedAdOb with four other baselines, as depicted in Fig. \ref{fig:ds-ns}. The results demonstrate that FedAdOb achieves comparable performance to Marvel and outperforms the other three methods.


\begin{figure*}[!ht]
	\centering
 \captionsetup[subfigure]{labelformat=empty}
	{$r1$\begin{subfigure}{0.23\textwidth}
\includegraphics[width=1\textwidth]{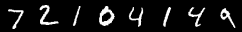}
		\end{subfigure}}
    	{
     \begin{subfigure}{0.23\textwidth}	 	\includegraphics[width=1\textwidth]{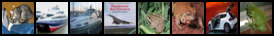}
    		\end{subfigure}}
    	{
    		 \begin{subfigure}{0.23\textwidth}
	\includegraphics[width=1\textwidth]{imgs/exp/reconstruct_whitebox_new/cifar_8_oriimg.png}
    		\end{subfigure}}
          	{
    		\begin{subfigure}{0.23\textwidth}
	\includegraphics[width=1\textwidth]{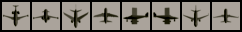}
    		\end{subfigure}}
\vspace{5pt}
      
     {$r2$\begin{subfigure}{0.23\textwidth}
\includegraphics[width=1\textwidth]{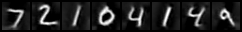}
		\end{subfigure}}
      {
		\begin{subfigure}{0.23\textwidth}
\includegraphics[width=1\textwidth]{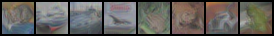}
		\end{subfigure}}
      {
		\begin{subfigure}{0.23\textwidth}
\includegraphics[width=1\textwidth]{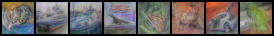}
		\end{subfigure}}
        {
		\begin{subfigure}{0.23\textwidth}
\includegraphics[width=1\textwidth]{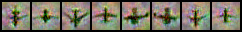}
		\end{subfigure}}
\vspace{5pt}

      {$r3$\begin{subfigure}{0.23\textwidth}
\includegraphics[width=1\textwidth]{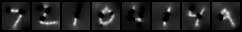}
		\end{subfigure}}
      {
		\begin{subfigure}{0.23\textwidth}
\includegraphics[width=1\textwidth]{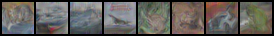}
		\end{subfigure}}
      {
		\begin{subfigure}{0.23\textwidth}
\includegraphics[width=1\textwidth]{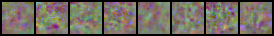}
		\end{subfigure}}
       {
		\begin{subfigure}{0.23\textwidth}
\includegraphics[width=1\textwidth]{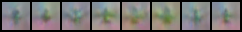}
		\end{subfigure}}
\vspace{5pt}
 
      {$r4$\begin{subfigure}{0.23\textwidth}
\includegraphics[width=1\textwidth]{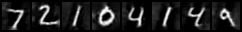}
		\end{subfigure}}
      {
		\begin{subfigure}{0.23\textwidth}
\includegraphics[width=1\textwidth]{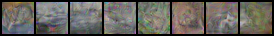}
		\end{subfigure}}
      {
		\begin{subfigure}{0.23\textwidth}
\includegraphics[width=1\textwidth]{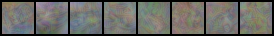}
		\end{subfigure}}
       {
		\begin{subfigure}{0.23\textwidth}
\includegraphics[width=1\textwidth]{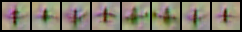}
		\end{subfigure}}
\vspace{5pt}

      {$r5$\begin{subfigure}{0.23\textwidth}
\includegraphics[width=1\textwidth]{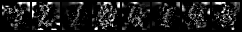}
		\end{subfigure}}
      {
		\begin{subfigure}{0.23\textwidth}
\includegraphics[width=1\textwidth]{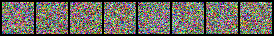}
		\end{subfigure}}
      {
		\begin{subfigure}{0.23\textwidth}
\includegraphics[width=1\textwidth]{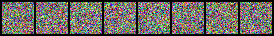}
		\end{subfigure}}
       {
		\begin{subfigure}{0.23\textwidth}
\includegraphics[width=1\textwidth]{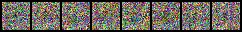}
		\end{subfigure}}
\vspace{5pt}

      {$r6$\begin{subfigure}{0.23\textwidth}
\includegraphics[width=1\textwidth]{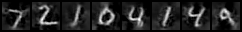}
		\end{subfigure}}
      {
		\begin{subfigure}{0.23\textwidth}
\includegraphics[width=1\textwidth]{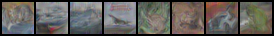}
		\end{subfigure}}
      {
		\begin{subfigure}{0.23\textwidth}
\includegraphics[width=1\textwidth]{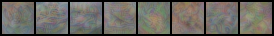}
		\end{subfigure}}
       {
		\begin{subfigure}{0.23\textwidth}
\includegraphics[width=1\textwidth]{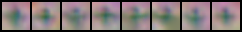}
		\end{subfigure}}
\vspace{5pt}

      {$r7$\begin{subfigure}{0.23\textwidth}
\includegraphics[width=1\textwidth]{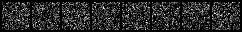}
		\end{subfigure}}
      {
		\begin{subfigure}{0.23\textwidth}
\includegraphics[width=1\textwidth]{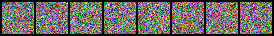}
		\end{subfigure}}
      {
		\begin{subfigure}{0.23\textwidth}
\includegraphics[width=1\textwidth]{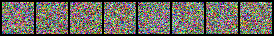}
		\end{subfigure}}
       {
		\begin{subfigure}{0.23\textwidth}
\includegraphics[width=1\textwidth]{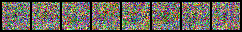}
		\end{subfigure}}
\vspace{5pt}

       {$r8$\begin{subfigure}{0.23\textwidth}
\includegraphics[width=1\textwidth]{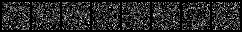}
		\end{subfigure}}
      {
		\begin{subfigure}{0.23\textwidth}
\includegraphics[width=1\textwidth]{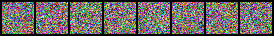}
		\end{subfigure}}
      {
		\begin{subfigure}{0.23\textwidth}
\includegraphics[width=1\textwidth]{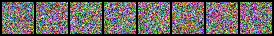}
		\end{subfigure}}
       {
		\begin{subfigure}{0.23\textwidth}
\includegraphics[width=1\textwidth]{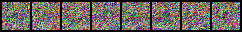}
		\end{subfigure}}

{
\begin{subfigure}{0.23\textwidth}
\subcaption*{LeNet-MNIST}
		\end{subfigure}
	}
      {
		\begin{subfigure}{0.23\textwidth}
\subcaption*{AlexNet-CIFAR10}
		\end{subfigure}
	}
      {
		\begin{subfigure}{0.23\textwidth}
\subcaption*{ResNet-CIFAR10}
		\end{subfigure}
	}
       {
		\begin{subfigure}{0.23\textwidth}
\subcaption*{LeNet-ModelNet}
		\end{subfigure}
	}
	\caption{Original images and images reconstructed by CAFE attack in \textbf{VFL} for different defense mechanisms on LeNet-MNIST, AlexNet-CIFAR10, ResNet-CIFAR10 and LeNet-ModelNet respectively. From top to bottom, a row represents original image ($r1$), no defense ($r2$), InstaHide ($r3$), DP with noise level $0.2$ ($r4$) and $2$ ($r5$), Sparsification with sparsification level $0.5$ ($r6$) and $0.05$ ($r7$), and FedAdOb ($r8$).}
	\label{fig:vis-whitebox}
\end{figure*}

\begin{figure}
\centering
	\centering
      		\begin{subfigure}{0.22\textwidth}
  		 	\includegraphics[width=1\textwidth]{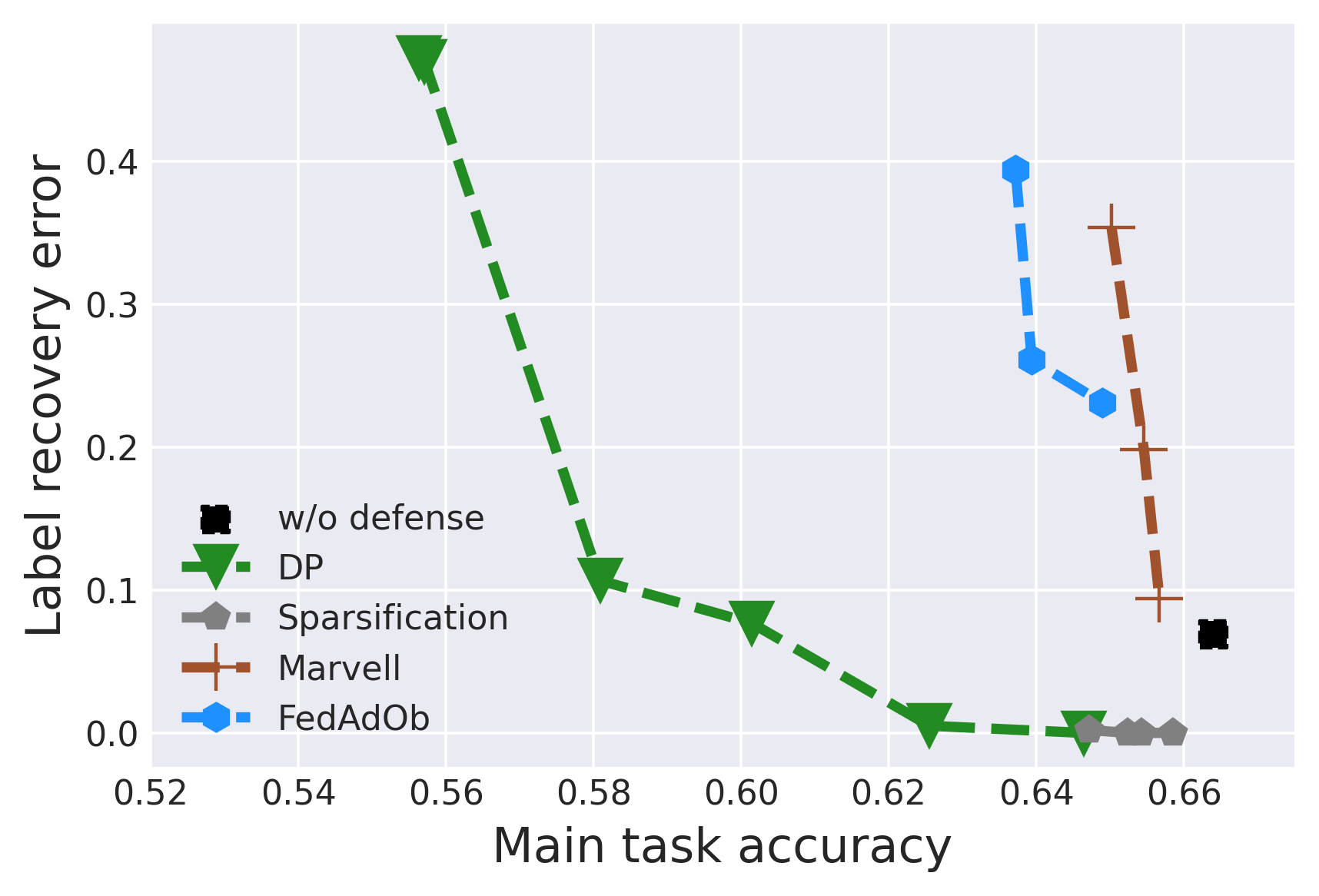}
      \subcaption{DS}
    		\end{subfigure}
    	\begin{subfigure}{0.22\textwidth}
  		 	\includegraphics[width=1\textwidth]{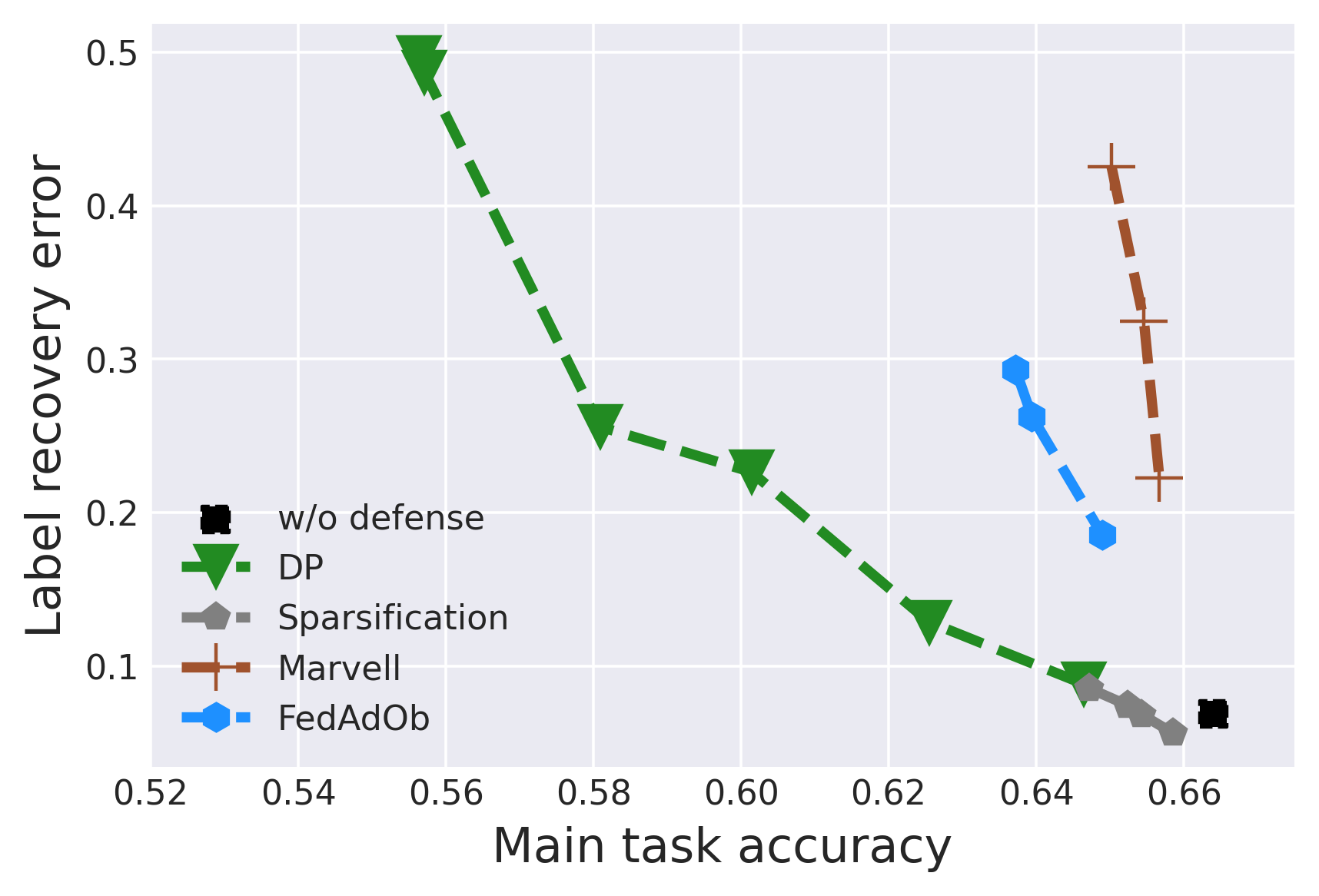}
            \subcaption{NS}
    		\end{subfigure}
\centering
                    \caption{Comparison of different defense methods in terms of their trade-offs between main task accuracy and label recovery error against DS and NS \cite{oscar2022split} attacks on DCN-Criteo in \textbf{VFL}.}
\label{fig:ds-ns}
\end{figure}

\subsection{Ablation Study} \label{sec:aba-tradeoff}
This section reports ablation study on two critical passport parameters.

As illustrated in main text, when the passports are embedded in a convolution layer or linear layer with $c$ channels\footnote{For the convolution layer, the passport $s\in \RR^{c\times h_1 \times h_2}$, where $c$ is channel number, $h_1$ and $h_2$ are height and width; for the linear layer, $s\in \RR^{ c\times h_1}$, where $c$ is channel number, $h_1$ is height.}, 
for each channel $j\in[c]$, 
the passport $s{(j)}$ (the $j_{th}$ element of vector $s$) is randomly generated as follows:
\begin{equation}\label{eq:sample-pst-app}
    s{(j)} \sim \calN(\mu_j, \sigma^2), \quad
    \mu_j \in \calU(-N, 0),
\end{equation}
where all $\mu_j, j=1,\cdots, c$ are different from each other, $\sigma^2$ is the variance of Gaussian distribution and $N$ is \textit{the range of passport mean}.  $\sigma$ and $N$ are two crucial parameters determining the privacy-utility trade-off of FedAdOb. This section analyzes how these two parameters influence the data recovery error and main task accuracy for FedAdOb. 

\begin{table}[htbp]
\centering
\resizebox{0.3\textwidth}{!}{\begin{tabular}{@{}c|c|c@{}}
\cmidrule(r){1-3}
         $N$ & \shortstack{Main task \\accuracy} & \shortstack{Feature recovery \\error}   \\ \cmidrule(r){1-3}
1   & 0.8798        & 0.0589                  \\
2   & 0.8771        & 0.0818                   \\
5   & 0.8758        & 0.1203           \\
10  & 0.8764        & 0.1204                   \\
\textbf{50}  & 0.8761        & 0.1205                   \\
100 & 0.8759        & 0.1205                   \\
200 & 0.8739        & 0.1205                   \\
\cmidrule(r){1-3}
\end{tabular}}
\caption{Influence of range of passports mean $N$ on main task accuracy and feature recovery error for FedAdOb against the CAFE attack in \textbf{VFL}.}
\label{tab:gaussian-mean}
\end{table}
\subsubsection{Influence of the range of Passport Mean $N$}
For the range of passport mean $N$, we consider the 2-party VFL scenario in AlexNet-CIFAR10 under feature recovery attack, i.e., CAFE \cite{jin2021cafe}, and MI \cite{he2019model} attacks. Fixed the $\sigma=1$, Table \ref{tab:gaussian-mean} provides the trade-off between main task accuracy and feature recovery error with different passport mean, which illustrates that when the range of passport mean increases, the privacy is preserved at the expense of the minor degradation of model performance less than 1\%. This phenomenon is consistent with our analysis of Theorem \ref{thm:thm1}. It is suggested that the $N$ is chosen as 100 when the model performance is less than 0.5\%.
\begin{table}[htbp]
\centering
\resizebox{0.3\textwidth}{!}{\begin{tabular}{c|c|c}
\toprule
$\sigma^2$ & \shortstack{Main task\\ accuracy} & \shortstack{Label recovery \\ error} \\ \midrule
0        & 0.877         & 0.368              \\
2        & 0.873         & 0.510               \\
\textbf{5}        & 0.870         & 0.543              \\
10       & 0.864         & 0.568              \\
50       & 0.858         & 0.577              \\
70       & 0.844         & 0.617              \\
100      & 0.791         & 0.751              \\ \bottomrule
\end{tabular}}
\caption{Influence of passports variance $\sigma^2$ on main task accuracy and label recovery error for FedAdOb in \textbf{VFL}.}
\label{tab:gaussian-var}
\end{table}

\subsubsection{Influence of Passport Variance $\sigma^2$}
For the passport variance $\sigma^2$, we consider the 2-party VFL scenario in AlexNet-CIFAR10 under label recovery attack \cite{fu2022label}, i.e., PMC attack \cite{fu2022label}. Fixed the range of passport mean $N =100$, Table \ref{tab:gaussian-var} shows that the main task accuracy decreases and label recovery error increases when the passport variance increases, which is also verified in Theorem \ref{thm2}. It is suggested that the $\sigma^2$ is chosen as 5 when the model performance is less than 0.5\%.

\subsection{Training and Inference Time}

TABLE \ref{tab:time} and \ref{tab:time-HFL} tests the training time (for one epoch) and inference time for FedAdOb and baseline defense methods in VFL and HFL setting. It shows that the FedAdOb is as efficient as the w/o defense for both training and inference procedures (the training time on MNIST for each epoch is 7.03s and the inference time is 1.48s in VFL) because embedding passport only introduces a few model parameters to train. It is worth noting that the training time of InstaHide is almost twice that of other methods because InstaHide involves mixing-up multiple feature vectors or labels, which is time-consuming.
\begin{table}[htbp]
\caption{Comparison of training time (for one epoch) and inference time among different defense mechanisms in \textbf{VFL}.}
\footnotesize
\center
\begin{tabular}{@{}lrlllll@{}}
\toprule
\multirow{2}{*}{\begin{tabular}[c]{@{}l@{}}Defense\\ Method\end{tabular}} &
  \multicolumn{2}{l}{\scriptsize LeNet-MNIST} &
  \multicolumn{2}{l}{\scriptsize AlexNet-Cifar10}  & \multicolumn{2}{l}{\scriptsize ResNet-Cifar10}\\ \cmidrule(l){2-7} 
 & Train & Infer & Train & Infer & Train & Infer \\ \hline \hline
w/o defense & 7.03  & 1.48 & 22.37 & 2.20 & 22.64 & 2.18 \\
\hline
CAE         & 7.30  & 1.48 & 22.71 & 2.27 & 23.02 & 2.21 \\
Sparsification & 6.93  & 1.45 & 22.39 & 2.12 & 22.61 & 2.21 \\
DP       & 7.01  & 1.49 & 22.24 & 2.23 & 22.63 & 2.16 \\
InstaHide   & 21.76 & 1.50 & 37.07 & 2.19 & 46.26 & 2.18 \\
\hline
\\[-1em]
FedAdOb (ours)    & 7.05  & 1.46 & 22.58 & 2.13 & 22.61 & 2.16 \\ \bottomrule
\end{tabular}
\label{tab:time}
\end{table}

\begin{table}[htbp]
\caption{Comparison of training time (for one epoch) and inference time among different defense mechanisms in \textbf{HFL}.}
\footnotesize
\center
\begin{tabular}{@{}lrlllll@{}}
\toprule
\multirow{2}{*}{\begin{tabular}[c]{@{}l@{}}Defense\\ Method\end{tabular}} &
  \multicolumn{2}{l}{\scriptsize LeNet-MNIST} &
  \multicolumn{2}{l}{\scriptsize AlexNet-Cifar10}  & \multicolumn{2}{l}{\scriptsize ResNet-Cifar100}\\ \cmidrule(l){2-7} 
 & Train & Infer & Train & Infer & Train & Infer \\ \hline \hline
w/o defense & 21.30  & 2.30 & 69.86 & 4.95 & 114.29 & 9.14 \\
\hline
SplitFed         & 22.05  & 2.66 & 70.51 & 5.18 & 115.37 & 9.13 \\
Sparsification & 21.76  & 2.38 & 70.59 & 5.34 & 121.13 & 9.45  \\
DP       & 21.53  & 2.32 & 70.54 & 5.63 & 120.82 & 9.54  \\
InstaHide   & 35.62 & 2.33 & 100.36 & 5.17 & 145.10 & 9.21 \\
\hline
\\[-1em]
FedAdOb (ours)    & 23.06  & 2.23 & 70.58 & 5.13 & 120.41 & 9.54 \\ \bottomrule
\end{tabular}
\label{tab:time-HFL}
\end{table}

\section{Conclusion}
In this paper, we propose a novel privacy-preserving federated deep learning framework called FedAdOb, which leverages adaptive obfuscation to protect the label and feature simultaneously. Specifically, the proposed adaptive obfuscation is implemented by embedding private passports in the bottom and top models to 
adapt the deep learning model such that the model performance is preserved. The extensive experiments on multiple datasets and theoretical analysis demonstrate that the FedAdOb can achieve significant improvements over the other protected methods in terms of model performance and privacy-preserving ability.

Besides the privacy attacks from the semi-honest adversary, defending against malicious attacks is another important research direction in federated learning \cite{blanchard2017machine,bagdasaryan2020backdoor}. In the future, we would like to investigate whether FedAdOb is effective in thwarting malicious attacks.


\begin{appendices}

\setcounter{equation}{0}
\setcounter{theorem}{0}
\setcounter{prop}{0}
\setcounter{definition}{0}

\section{Experimental Setting}

This section provides detailed information on our experimental settings. Table \ref{table:models-appendix} summarizes datasets and DNN models evaluated in all experiments and Table \ref{tab:train-params} summarizes the hyper-parameters used for training  models.



\subsection{Dataset \& Model Architectures}

We consider classification tasks by LeNet \cite{lecun1998gradient} on MNIST \cite{lecun2010mnist}, AlexNet \cite{NIPS2012_c399862d} on CIFAR10 \cite{krizhevsky2014cifar}, ResNet \cite{he2016deep} on CIFAR10 dataset and LeNet on ModelNet \cite{wu20153d,liu2022cross}. 
Specifically, the MNIST database of 10-class handwritten digits has a training set of 60,000 examples, and a test set of 10,000 examples. The CIFAR-10 dataset consists of 60000 $32\times 32$ colour images in 10 classes, with 6000 images per class. ModelNet is a widely-used 3D shape classification and shape retrieval benchmark, which currently contains 127,915 3D CAD models from 662 categories. We created 12 2D multi-view images per 3D mesh model by placing 12 virtual cameras evenly distributed around the centroid and partitioned the images into multiple (2 to 6) parties by their angles, which contains 6366 images to train and 1600 images to test. The Criteo dataset is used for predicting ads click-through rate. It has 26 categorical features and 13 continuous features. Categorical features are transformed into embeddings with fixed dimensions before we feed them into the model. In the 2-party VFL setting, both kinds of features are divided into two parts  and party A has the labels. The ratio of positive samples to negative samples is 1 : 16. To reduce the computational complexity, we randomly select 100000 samples as the training set and 20000 samples as the test set.
Moreover, we add the batch normalization layer of the LeNet and AlexNet for after the last convolution layer for existing defense methods.
\begin{table}[!h]
	\centering
	\footnotesize
	\begin{tabular}{c||c|c|c|c}
	        \hline
             \\[-1em]
		Scenario& \shortstack{Model \& \\ Dataset}  & \shortstack{Bottom \\ Model } & \shortstack{Top \\ Model}   & \# P \\
         \hline
         \hline
           \\[-1em]
            \multirow{3}{*}{HFL}&LeNet-MNIST & 1 Conv  &  1 Conv+ 3 FC  & 10 \\
		\cline{2-5}
       \\[-1em]
            & AlexNet-CIAFR10 & 1 Conv  &  4 Conv+ 1 FC  & 10 \\
		\cline{2-5}
		\\[-1em]
            & ResNet-CIFAR100 & 1 Conv  &   16 Conv+ 1 FC   & 10 \\
		\hline
        \hline
          \\[-1em]
              \multirow{5}{*}{VFL}& LeNet-MNIST & 2 Conv  &  3 FC  & 2 \\
		\cline{2-5}
          \\[-1em]
		&AlexNet-CIFAR10 & 5 Conv  &  1 FC & 2  \\
		\cline{2-5}
          \\[-1em]
		&ResNet18-CIAFR10 & 17 Conv & 1 FC & 2\\
		\cline{2-5}
     \\[-1em]
     	&LeNet-ModelNet & 2 Conv  &  3 FC  & 7 \\
		\cline{2-5}
       \\[-1em]
     	&DCN-Criteo  & 3 FC  &  1 FC  & 2 \\
		\hline
         
	\end{tabular}
	\caption{Models for evaluation in HFL and VFL. \# P denotes the number of parties. FC: fully-connected layer. Conv: convolution layer. }
\label{table:models-appendix}
\end{table}

\begin{table*}[!htbp]  \renewcommand\arraystretch{1.2}
\vspace{-3pt}
	\resizebox{1\textwidth}{!}{
		\begin{tabular}{l|cccc}
			\hline
			Hyper-parameter & LeNet-MNIST& AlexNet-CIFAR10 & ResNet-CIFAR10& LeNet-ModelNet \\ \hline
			Optimization method & SGD & SGD& SGD& SGD\\
			Learning rate & 1e-2 & 1e-2& 1e-2&1e-2\\
			Weight decay & 4e-5 & 4e-5& 4e-5&4e-5\\
			Batch size & 64 & 64& 64&64\\
			Iterations & 50 & 100& 100&100 \\
			The range of $N$ (passport mean) & [1, 50] & [1, 200] & [1, 100] & [1, 50] \\
			The range of $\sigma^2$ (passport variance) & [1, 9] & [1, 64] & [1, 25] & [1,  9] \\
			\hline
	\end{tabular}}
	
	\caption{Hyper-parameters used for training in FedAdOb.}
	\label{tab:train-params}
\end{table*}

\begin{table*}[!htbp]  \renewcommand\arraystretch{1.2}
\vspace{-3pt}
	\resizebox{1\textwidth}{!}{
		\begin{tabular}{l|cccc}
			\hline
			Hyper-parameter & LeNet-MNIST& AlexNet-CIFAR10 & ResNet-CIFAR10& LeNet-ModelNet \\ \hline
			Optimization method & SGD & SGD& SGD& SGD\\
			Learning rate & 1e-2 & 1e-2& 1e-2&1e-2\\
			Weight decay & 4e-5 & 4e-5& 4e-5&4e-5\\
			Batch size & 64 & 64& 64&64\\
			Iterations & 50 & 100& 100&100 \\
			The range of $N$ (passport mean) & [1, 50] & [1, 200] & [1, 100] & [1, 50] \\
			The range of $\sigma^2$ (passport variance) & [1, 9] & [1, 64] & [1, 25] & [1,  9] \\
			\hline
	\end{tabular}}
	
	\caption{Hyper-parameters used for training in FedAdOb.}
	\label{tab:HFL-train-params}
\end{table*}

\subsection{Federated Learning Settings}
In HFL, each client incorporates their passport information into the bottom model to safeguard the features. In VFL, 
we simulate a VFL scenario by splitting a neural network into a bottom model and a top model and assigning the bottom model to each passive party and the top model to the active party. For the 2-party scenario, the passive party has features when the active party owns the corresponding labels. For the 7-party scenario, i.e., the ModelNet, each passive party has two directions of 3D shapes when the active party owns the corresponding label following \cite{liu2022cross}. Table \ref{table:models-appendix} summarizes our FL scenarios. Also, the VFL framework follows algorithm 1 of \cite{liu2022vertical}.


\subsection{Privacy Attack Methods}
We investigate the effectiveness of FedAdOb on three attacks designed for VFL
\begin{itemize}
    \item Passive Model Completion (PMC) attack \cite{fu2022label} is the label inference attack. For each dataset, the attacker leverages some auxiliary labeled data (40 for MNIST and CIFAR10, 366 for ModelNet) to train the attack model and then the attacker predicts labels of the test data using the trained attack model. 
    \item CAFE attack \cite{jin2021cafe} is a feature reconstruction attack. We assume the attacker knows the forward embedding $H$ and passive party's model $G_{\theta}$, and we follow the white-box model inversion step (i.e., step 2) of CAFE to recover private features owned by a passive party.
    \item Model Inversion (MI) attack is a feature reconstruction attack. We follow the black-box setting of \cite{he2019model}, in which the attacker does not know the passive model. For each dataset, attacker leverages some labeled data (10000 for MNIST and CIFAR10, 1000 for ModelNet) to train a shallow model approximating the passive model, and then they use the forward embedding and the shallow model to infer the private features of a passive party inversely.  
\end{itemize}
For CAFE and MI attacks, we add the total variation loss~\cite{jin2021cafe} into two feature recovery attacks with regularization parameter $\lambda=0.1$.

\subsection{Baseline Defense Methods}

The details of baseline defense methods compared with FedAdOb is as follows:
\begin{itemize}
    \item For \textbf{Differential Privacy (DP)} \cite{abadi2016deep},  we experiment with Gaussian noise levels ranging from 5e-5 to 1.0. We add noise to gradients for defending against MC attack while add noise to forward embeddings for thwarting CAFE and MI attacks. 
    \item For \textbf{Sparsification} \cite{fu2022label,lin2018deep}, we implement gradient sparsification \cite{fu2022label} and forward embedding sparsification \cite{lin2018deep} for defending against label inference attack and feature reconstruction attack, respectively. Sparsification level are chosen from 0.1\% to 50.0\%. 
    \item  For \textbf{InstaHide}~\cite{huang2020instahide}, we mix up 1 to 4 of images to trade-off privacy and utility.
    \item For \textbf{Confusional AutoEncoder} (CAE)~\cite{zou2022defending}, we follow the implementation of the original paper. That is, both the encoder and decoder of CAE have the architecture of 2 FC layers. Values of the hyperparameter that controls the confusion level are chosen from 0.0 to 2.0.
\end{itemize}

\subsubsection{Implementation details of FedAdOb}
For \textbf{FedAdOb}, Passports are embedded in the last convolution layer of the passive party's model and the first fully connected layer of the active party's model. Table \ref{tab:train-params} summarizes the hyper-parameters for training FedAdOb.

\subsection{Notations}
\begin{table}[!htbp] 
  \renewcommand{\arraystretch}{1.05}
  \centering
  \setlength{\belowcaptionskip}{15pt}
  \label{table: notation}
    \begin{tabular}{c|p{5.5cm}}
    \toprule
    Notation & Meaning\cr
    \midrule\
    $F_\omega, G_{\theta_k}$ & Active model and $k_{th}$ passive model\cr \hline
     $W_p, W_a, W_{att}$ & The 2D matrix of the active model, passive model and attack model for the regression task \cr \hline
     $s_a, s_p$ & The passport of the active model and passive model \cr \hline
    $s_{a}^\gamma, s_{a}^\beta,s_p^\gamma, s_p^\beta$ & The passport of the active model and passive model w.r.t $\gamma$ and $\beta$ respectively \cr \hline
    $\gamma, \beta$ & The scaling and bias   \cr \hline
    $g_W$ & Adaptive obfuscations w.r.t model parameters $W$ \cr \hline
    $H_k$ & Forward embedding passive party $k$ transfers to the active party \cr \hline
    $\Tilde{\ell}$ & the loss of main task \cr \hline
    $\nabla_{H_k}\Tilde{l}$ & Backward gradients the active party transfers to the passive party $k$ \cr \hline
    $\calD_k= \{x_{k,i}\}_{i=1}^{n_i}$ & $n_i$ private features of passive party $k$ \cr \hline
     $y$ & Private labels of active party \cr \hline
    $K$ & The number of party\cr \hline
    $N$ & The range of Gaussian mean of passports sample \cr \hline
   $\sigma^2$ & The variance of passports sample \cr \hline
    $\Tilde{\ell}_a$ &  Training error on the auxiliary dataset for attackers \cr \hline
        $\Tilde{\ell}_t$ &  Test error on the test dataset for attackers \cr \hline
      $n_a$ & The number of auxiliary dataset to do PMC attack\cr \hline
   $T$ & Number of optimization steps for VFL \cr \hline
      $\eta$ & learning rate \cr \hline
 $\|\cdot \| $ & $\ell_2$ norm \cr 
    \bottomrule
    \end{tabular}
    \caption{Table of Notations}
\end{table}


	
			


\section{Formulation of FedAdOb in HFL and VFL}
Consider a neural network $f_\Theta(x):\calX \to \RR$, where $x \in \calX$, $\Theta$ denotes model parameters of neural networks. 

\noindent\textbf{VFL.} $K$ passive parties and one active party collaboratively optimize $\Theta = (\omega, \theta_1, \cdots, \theta_K)$ of network according to Eq. \eqref{eq:loss-VFL-app}.
\begin{equation}\label{eq:loss-VFL-app}
\begin{split}
        \min_{\omega, \theta_1, \cdots, \theta_K} &\frac{1}{n}\sum_{i=1}^n\ell(F_{\omega} \circ (G_{\theta_1}(x_{1,i}),G_{\theta_2}(x_{2,i}), \\
        & \cdots,G_{\theta_K}(x_{K,i})), y_{i}),
\end{split}
\end{equation}
where $\ell$ is the loss, e.g., the cross-entropy loss, passive party $P_k$ owns features $\calD_k = (x_{k,1}, \cdots, x_{k,n}) \in \mathcal{X}_k$ and the passive model $G_{\theta_k}$, the active party owns the labels $y \in \mathcal{Y}$ and active model $F_\omega$, $\mathcal{X}_k$ and $\mathcal{Y}$ are the feature space of party $P_k$ and the label space respectively. Furthermore, FedAdOb aims to optimize:
\begin{equation} \label{eq:loss-aof-vfl-app}
\begin{split}
        \min_{\omega, \theta_1,\cdots, \theta_K} &\frac{1}{N}\sum_{i=1}^N\ell(F_{\omega}g_{\omega} \circ (G_{\theta_1}(g_{\theta_1}(x_{1,i}, s_{p_1})), \\
        &\cdots, G_{\theta_K}(g_{\theta_K}(x_{K,i},s_{p_K})), y_{i}).
\end{split}
\end{equation}
Denote the composite function  $G_{\theta_j}g_{\theta_j}()$ as $G'_{\theta_j}()$, $j=1, \cdots, K$. Then we rewrite the Eq. \eqref{eq:loss-aof-vfl-app} as follows 
\begin{equation} 
\begin{split}
        \min_{\omega, \theta_1,\cdots, \theta_K} &\frac{1}{N}\sum_{i=1}^N\ell(F'_{\omega} \circ (G'_{\theta_1}(x_{1,i}, s^{p_1})), \\
        &\cdots, G'_{\theta_K}(x_{K,i}, s^{p_K})) , s^a, y_{i}),
\end{split}
\end{equation}

\noindent\textbf{HFL.} $K$  party collaboratively optimize $\Theta = (\omega, \theta_1, \cdots, \theta_K)$ of network according to Eq. \eqref{eq:loss-HFL-app}.
\begin{equation} \label{eq:loss-HFL-app}
    \min_{\theta_1, \cdots, \theta_K,\omega} \sum_{k=1}^K\sum_{i=1}^{n_k}\frac{\ell(F_\omega(x_{k,i}), y_{k,i})}{n_1+\cdots+n_K},
\end{equation}
where $\ell$ is the loss, e.g., the cross-entropy loss, $\calD_k=\{(x_{k,i}, y_{k,i})\}_{i=1}^{n_k}$ is the dataset with size $n_k$ owned by client $k$.
Furthermore, FedAdOb aims to optimize:
\begin{equation} \label{eq:loss-aof-app}
\begin{split}
        \min_{\omega, \theta_1,\cdots, \theta_K} &\frac{1}{N}\sum_{i=1}^N\ell(F_{\omega}g_{\omega} \circ (G_{\theta_1}(g_{\theta_1}(x_{1,i}, s_{p_1})), \\
        &\cdots, G_{\theta_K}(g_{\theta_K}(x_{K,i},s_{p_K})), y_{i}).
\end{split}
\end{equation}
Denote the composite function  $G_{\theta_j}g_{\theta_j}()$ as $G'_{\theta_j}()$, $j=1, \cdots, K$. Then we rewrite the Eq. \eqref{eq:loss-aof-app} as follows 
\begin{equation} \label{eq:loss-aof-HFL-app}
    \min_{\theta_1, \cdots, \theta_K,\omega} \sum_{k=1}^K\sum_{i=1}^{n_k}\frac{\ell(F_\omega \circ (G'_{\theta_k}(x_{k,i}, s_{p_k}), y_{k,i})}{n_1+\cdots+n_K},
\end{equation}

Then Proposition 1 ends the proof.

\section{Security Analysis for FedAdOb}
We investigate the privacy-preserving capability of FedAdOb against feature reconstruction attack and label inference attack. Note that we conduct the privacy analysis with linear regression models, for the sake of brevity. 

\begin{definition} \label{def:SplitFed-app}
Define the forward function of the bottom model $G$ and the top model $F$: 
\begin{itemize}
    \item For the bottom layer: $H = G(x) =  W_p s_p^\gamma \cdot W_p x + W_p s_p^\beta$.
    \item For the top layer: $y = F(H) =  W_a s_a^\gamma \cdot W_a  H + W_a s_a^\beta$.
\end{itemize}
where $W_p$, $W_a$ are 2D matrices of the bottom and top models; $\cdot$ denotes the inner product, $ s_p^\gamma,  s_p^\beta$ are passports embedding into the bottom layers, $ s_a^\gamma,  s_a^\beta$ are passports embedding into the top layers.

\end{definition}

\subsection{Hardness of Feature Restoration with FedAdOb}
Considering two strong feature restoration attacks, White-box Gradient Inversion (WGI) \cite{zhu2019dlg} and White-box Model Inversion (WMI) attack \cite{jin2021cafe,he2019model}, which aims to recover features $\hat{x}$ approximating original features $x$ according to the model gradients and outputs respectively. Specifically, for WMI, the attacker knows the bottom model parameters $W_p$, forward embedding $H$ and the way of embedding passport, but does not know the passport. For WGI, the adversary knows the bottom model gradients $\nabla W_p$, and the way of embedding passport, but does not know the passport.

\begin{lem} \label{lem1-app}
Suppose a bottom model as the Def. \ref{def:SplitFed-app} illustrates, an attack estimates the private feature by guessing passports $s_{\gamma'}^p$ and $s^{\beta'}_p$ via WMI attack. If $W_p$ is invertible, we could obtain the difference between $x$ and estimated $\hat{x}$ by the adversary in the following two cases:
\begin{itemize}
    \item When inserting the $s_\gamma^p$ for the passive party,
    \begin{equation}
    \|x - \hat{x}\|_2 \geq  \frac{\|(D_\gamma^{-1}-D_{\gamma'}^{-1}) H\|_2}{\| W_p\|_2}.
\end{equation}
\item When inserting the $s^\beta_p$ for the passive party,
\begin{equation} \label{eq:beta-lem-app}
     \|x - \hat{x}\|_2 = \|s^\beta_p - s^{\beta'}_p \|_2,
\end{equation}
\end{itemize}
where $D_\gamma=diag(W_ps_\gamma^p),D_{\gamma'}=diag(W_ps_{\gamma'}^p),D_\beta=diag(W_ps^\beta_p),D_{\beta'}=diag(W_ps^{\beta'}_p)$ and $W_p^\dag$ is the Moore–Penrose inverse of $W_p$. 
\end{lem}
\begin{proof}
For inserting the $s_\gamma^p$,
\begin{equation*}
    H = W_p s_\gamma^p * W_p x.
\end{equation*}
The attacker with knowing $W_p$ has
\begin{equation*}
    H =W_p s_{\gamma'}^p * W_p \hat{x}.
\end{equation*}
Denote $D_\gamma=diag(W_ps_\gamma^p),D_{\gamma'}=diag(W_ps_{\gamma'}^p)$. If $A$ is invertible, due to
\begin{equation*}
    \|A\|_2\|A^{-1}B\|_2 \leq \|AA^{-1}B\|_2 = \|B\|_2,
\end{equation*}
we have 
\begin{equation*}
    \|A^{-1}B\|_2 \geq \frac{\|B\|_2}{\|A\|_2}.
\end{equation*}   
Therefore, if $W_P^\dag$ is the Moore–Penrose inverse of $W_p$ and $W_p$ is invertible, we can obtain
\begin{equation}\label{eq:case1}
    \begin{split}
        \|\hat{x}-x\|&=\|W_p^\dag D_\gamma^{-1}H-W_p^\dag D_{\gamma'}^{-1}H\|_2 \\
        &= \| W_p^\dag(D_\gamma^{-1}-D_{\gamma'}^{-1})H \|_2 \\
        &\geq \frac{\|(D_\gamma^{-1}-D_{\gamma'}^{-1}) H\|_2}{\| W_p\|_2}.
    \end{split}
\end{equation}
On the other hand, for inserting the $s^\beta_p$ of passive party,
\begin{equation*}
     H =  W_p x + W_p s^\beta_p.
\end{equation*}
The attacker with knowing $W_p$ has
\begin{equation*}
    H  = W_p x + W_p s^{\beta'}_p.
\end{equation*}
We further obtain
\begin{equation}\label{eq:case2}
    \begin{split}
        \|\hat{x}-x\|&=\|W_p^\dag (H-D_\beta)-W_p^\dag (H-D_{\beta'})\|_2 \\
        &= \| W_p^\dag(W_p s^\beta_p- W_p s^{\beta'}_p) \|_2 \\
        &= \|s^\beta_p - s^{\beta'}_p \|_2.
    \end{split}
\end{equation}
\end{proof}
\begin{rmk}
If the adversary don't know the existence of passport layer, then $D_{\gamma'} = I$ and $s^{\beta'}_p =0$ in Eq. \eqref{eq:case1} and \eqref{eq:case2}.
\end{rmk}
\begin{rmk}
Eq. \eqref{eq:beta-lem-app} only needs $W_p$ is a left inverse, i.e., $W_p$ has linearly independent columns. 
\end{rmk}


\begin{lem} \label{lem3-app}
Suppose a bottom model as the Def. \ref{def:SplitFed-app} illustrates, an attack estimates the private feature by guessing passports $s_{\gamma'}^p$ and $s^{\beta'}_p$ via WGI attack. If $\nabla b$ is non-zero and inserting the $s^\beta_p$ for the bottom model, we could obtain 
\begin{equation} \label{eq:beta-lem-app1}
     \|x - \hat{x}\|_2 = \|s^\beta_p - s^{\beta'}_p \|_2.
\end{equation}
\end{lem}

\begin{proof}
For inserting the $s^\beta_p$ of passive model,
\begin{equation*}
     H =  W_p x + b_\beta= W_p x + W_p s^\beta_p.
\end{equation*}
The attacker with knowing $W_p$ has
\begin{equation*}
    H  = W_p x + b_{\beta'} = W_p x +W_p s^{\beta'}_p.
\end{equation*}
Denote $b_{\beta'} = W_p s^{\beta'}_p$ and $b_{\beta} = W_p s^\beta_p$, $\nabla W_p = \nabla_{W_p}\ell$ and $\nabla W_p = \nabla_{W_p}\ell$. According to chain rule, we can obtain
\begin{equation}
    \begin{split}
        \nabla W_p  = \nabla_H \ell (x + s_\beta)^T,
    \end{split}
\end{equation}
and
\begin{equation}
    \begin{split}
        \nabla b  = \nabla_H \ell I
    \end{split}
\end{equation}
Therefore, 
\begin{equation}
    \nabla W_p =  \nabla b (x+s^\beta_p)^T,
\end{equation}
and
\begin{equation}
    \nabla W_p =  \nabla b (\hat{x}+s^{\beta'}_p)^T,
\end{equation}
If $\nabla b$ is not zero, then 
\begin{equation}
   \hat{x}+s^{\beta'}_p = x +s^\beta_p, 
\end{equation}
which implies to
\begin{equation}
    \|x -\hat{x}\|_2 = \|s^\beta_p - s^{\beta'}_p \|_2.
\end{equation}
\end{proof}

Lemma \ref{lem1-app} and Lemma \ref{lem3-app} illustrate that the difference between estimated feature and original feature has the lower bound, which depends on the difference between original passport and inferred passport. Specifically, if the inferred passport is far away from $\|s^\beta_p - s^{\beta'}_p \|_2$ and $\|(D_\gamma^{-1}-D_{\gamma'}^{-1})\|$ goes large causing a large reconstruction error by attackers. Furthermore, we provide the analysis of the probability of attackers to reconstruct the $x$ in Theorem \ref{thm:thm1-app}. Let $m$ denote the dimension of the passport via flattening, $N$ denote the passport mean range formulated in Eq. (5) of the main text and $\Gamma(\cdot)$ denote the Gamma distribution.

\begin{theorem}\label{thm:thm1-app}
Suppose the passive party protects features $x$ by inserting the $s^\beta_p$. The probability of recovering features by the attacker via white-box MI and GI attack is at most $\frac{\pi^{m/2}\epsilon^m}{\Gamma(1+m/2)N^m}$ such that the recovering error is less than $\epsilon$, i.e., $\|x-\hat{x}\|_2\leq \epsilon$,
\end{theorem}

\begin{proof}
According to Lemma \ref{lem1-app}, the attacker aims to recover the feature $\hat{x}$ within the $\epsilon$ error from the original feature $x$, that is, the guessed passport needs to satisfy: 
    \begin{equation}
        \|s^{\beta'}_p - s^\beta_p  \|_2 \leq \epsilon
    \end{equation}
Therefore, the area of inferred passport of attackers is sphere with the center $s^\beta_p$ and radius $\epsilon$. And the volumes of this area is at most $\frac{\pi^{m/2}\epsilon^m}{\Gamma(1+m/2)}$, where $F$ represent the Gamma distribution and $m$ is dimension of the passport. Consequently, the probability of attackers to successfully recover the feature within $\epsilon$ error is:
\begin{equation*}
\begin{split}
        p \leq \frac{\pi^{m/2}\epsilon^m}{\Gamma(1+m/2)N^m},
        \end{split}
\end{equation*}
where the whole space is $(-N,0)^m$ with volume $N^m$.
\end{proof}
Theorem \ref{thm:thm1-app} demonstrates that the attacker's probability of recovering features within error $\epsilon$ is exponentially small in the dimension of passport size $m$. The successful recovering probability is inversely proportional to the passport mean range $N$, which is consistent with our ablation study in Sect. 6.3 in the main text.
\begin{rmk}
The attacker's behaviour we consider in Theorem \ref{thm:thm1-app} is that they guess the private passport randomly. 
\end{rmk}
\subsection{Hardness of Label Recovery with FedAdOb}
Consider the passive model competition attack \cite{fu2022label} that aims to recover labels owned by the active party. The attacker (i.e., the passive party) leverages a small auxiliary labeled dataset $\{x_i, y_i\}_{i=1}^{n_a}$ belonging to the original training data to train the attack model $W_{att}$, and then infer labels for the test data. Note that the attacker knows the trained passive model $G$ and forward embedding $H_i = G(x_i)$. Therefore, they optimizes the attack model $W_{att}$ by minimizing $\ell = \sum_{i=1}^{n_a}\|W_{att}H_i-y_i\|_2$.

\begin{assumption}\label{assum1-app}
Suppose the original main algorithm of VFL is convergent. For the attack model, we assume the error of the optimized attack model $W^*_{att}$ on test data $\tilde{\ell}_t$ is larger than that of the auxiliary labeled dataset $\Tilde{\ell}_a$.
\end{assumption}

\begin{rmk}
The test error is usually higher than the training error because the error is computed on an unknown dataset that the model hasn't seen.
\end{rmk}
\begin{theorem} \label{thm2-app}
Suppose the active party protect $y$ by embedding $s^a_\gamma$, and adversaries aim to recover labels on the test data with the error $\Tilde{\ell}_t$ satisfying:
\begin{equation}
        \Tilde{\ell}_t \geq \min_{W_{att}} \sum_{i=1}^{n_a}\|(W_{att}- T_i)H_i\|_2,
    \end{equation}
    where $T_i =diag(W_as_{\gamma,i}^a) W_a$ and $s_{\gamma,i}^a$ is the passport for the label $y_i$ embedded in the active model. Moreover, if $H_{i_1} = H_{i_2} = H$ for any $1\leq i_1,i_2 \leq n_a$, then
    \begin{equation} \label{eq:protecty-app}
        \Tilde{\ell}_t \geq \frac{1}{(n_a-1)}\sum_{1\leq i_1<i_2\leq n_a}\|(T_{i_1}-T_{i_2})H\|_2
    \end{equation}
\end{theorem}

\begin{proof}
For only inserting the passport $s^a_\gamma$ of active party, according to Assumption \ref{assum1-app}, we have
\begin{equation}
    y_i=   W_a s_\gamma^a \cdot W_a  H_i
\end{equation}
Moreover, the attackers aims to optimize
\begin{equation*}
\begin{split}
       &\min_W\sum_{i=1}^{n_a}\|WH_i-y_i\|_2  \\
       = &\min_W\sum_{i=1}^{n_a}\|WH_i - W_a s_\gamma^a \cdot W_a  H_i\|_2 \\
       =&\min_W\sum_{i=1}^{n_a}\|(W-T_i)H_i\|_2,
\end{split}   
\end{equation*}
where $T_i =diag(W_as_{\gamma,i}^a) W_a$. Therefore, based on Assumption \ref{assum1-app}, $\Tilde{\ell}_t\geq \min_W\sum_{i=1}^{n_a}\|(W-T_i)H_i\|_2$. Moreover, if $H_i = H_j = H$ for any $i_1,i_2 \in [n_a]$, then
\begin{equation*}
    \begin{split}
        \Tilde{\ell}_t& \geq \min_W\sum_{i=1}^{n_a}\|(W-T_i)H_i\|_2 \\
        &= \min_W \frac{1}{2(n_a-1)}\sum_{1\leq i_1,i_2\leq n_a}(\|(W-T_{i_1})H\|_2 \\
        &  \qquad  + \|(W-T_{i_2})H\|_2) \\
        &\geq \min_W \frac{1}{2(n_a-1)}\sum_{1\leq i_1,i_2\leq n_a}(\|(T_{i_1}-T_{i_2})H\|_2) \\
       & =  \frac{1}{(n_a-1)}\sum_{1\leq i_1<i_2\leq n_a}(\|(T_{i_1}-T_{i_2})H\|_2)
    \end{split}
\end{equation*}
\end{proof}

\begin{prop}\label{prop1-app}
Since passports are randomly generated and $W_a$ and $H$ are fixed, if the $W_a = I, H=\Vec{1}$, then it follows that:
\begin{equation}
    \Tilde{\ell}_t \geq \frac{1}{(n_a-1)}\sum_{1\leq i_1<i_2\leq n_a}\|s_{\gamma,i_1}^a-s_{\gamma,i_2}^a\|_2)
\end{equation}
\end{prop}
\begin{proof}
  When $W_a =I$ and $H = \Vec{1}$, $(T_{i_1}-T_{i_2})H = s_{\gamma,i_1}^a-s_{\gamma,i_2}$ . Therefore, based on Theorem \ref{thm2-app}, we obtain
  \begin{equation}
          \Tilde{\ell}_t \geq \frac{1}{(n_a-1)}\sum_{1\leq i_1<i_2\leq n_a}\|s_{\gamma,i_1}^a-s_{\gamma,i_2}^a\|_2)
  \end{equation}
\end{proof}

Theorem \ref{thm2-app} and Proposition \ref{prop1-app} show that the label recovery error $\Tilde{\ell}_t$ has a lower bound,
which deserves further explanations. 
\begin{itemize}
    \item First, when passports are randomly generated for all data, i.e., $s_{\gamma,i_1}^a \neq s_{\gamma,i_2}^a$, then a non-zero label recovery error is guaranteed no matter how adversaries attempt to minimize it. The recovery error thus acts as a protective random noise imposed on true labels.
    \item Second, the magnitude of the recovery error monotonically increases with the variance $\sigma^2$ of the Gaussian distribution passports sample from (It is because the difference of two samples from the same Gaussian distribution $\calN(\mu, \sigma^2)$  depends on the variance $\sigma^2$), which is a crucial parameter to control privacy-preserving capability. Experiments on Sect. 6.4 of the main text also verify this phenomenon. 
    \item Third, it is worth noting that the lower bound is based on the training error of the auxiliary data used by adversaries to launch PMC attacks. Given possible discrepancies between the auxiliary data and private labels, e.g., in terms of distributions and the number of dataset samples, the actual recovery error of private labels can be much larger than the lower bound. 
\end{itemize}
 

\end{appendices}
\newpage
\bibliographystyle{IEEEtran}
\bibliography{Reference/FL,Reference/attack,Reference/cryptography,Reference/model,Reference/protection}

\begin{thebibliography}{10}
\providecommand{\url}[1]{#1}
\csname url@samestyle\endcsname
\providecommand{\newblock}{\relax}
\providecommand{\bibinfo}[2]{#2}
\providecommand{\BIBentrySTDinterwordspacing}{\spaceskip=0pt\relax}
\providecommand{\BIBentryALTinterwordstretchfactor}{4}
\providecommand{\BIBentryALTinterwordspacing}{\spaceskip=\fontdimen2\font plus
\BIBentryALTinterwordstretchfactor\fontdimen3\font minus
  \fontdimen4\font\relax}
\providecommand{\BIBforeignlanguage}[2]{{%
\expandafter\ifx\csname l@#1\endcsname\relax
\typeout{** WARNING: IEEEtran.bst: No hyphenation pattern has been}%
\typeout{** loaded for the language `#1'. Using the pattern for}%
\typeout{** the default language instead.}%
\else
\language=\csname l@#1\endcsname
\fi
#2}}
\providecommand{\BIBdecl}{\relax}
\BIBdecl

\bibitem{zhu2019dlg}
L.~Zhu, Z.~Liu, , and S.~Han, ``Deep leakage from gradients,'' in \emph{Annual
  Conference on Neural Information Processing Systems (NeurIPS)}, 2019.

\bibitem{konevcny2015federated}
J.~Kone{\v{c}}n{\`y}, B.~McMahan, and D.~Ramage, ``Federated optimization:
  Distributed optimization beyond the datacenter,'' \emph{arXiv preprint
  arXiv:1511.03575}, 2015.

\bibitem{konevcny2016federated}
J.~Kone{\v{c}}n{\`y}, H.~B. McMahan, F.~X. Yu, P.~Richt{\'a}rik, A.~T. Suresh,
  and D.~Bacon, ``Federated learning: Strategies for improving communication
  efficiency,'' \emph{arXiv preprint arXiv:1610.05492}, 2016.

\bibitem{mcmahan2017communication}
B.~McMahan, E.~Moore, D.~Ramage, S.~Hampson, and B.~A. y~Arcas,
  ``Communication-efficient learning of deep networks from decentralized
  data,'' in \emph{Artificial Intelligence and Statistics}.\hskip 1em plus
  0.5em minus 0.4em\relax PMLR, 2017, pp. 1273--1282.

\bibitem{yang2019federated}
Q.~Yang, Y.~Liu, T.~Chen, and Y.~Tong, ``Federated machine learning: Concept
  and applications,'' \emph{ACM Transactions on Intelligent Systems and
  Technology (TIST)}, vol.~10, no.~2, pp. 1--19, 2019.

\bibitem{jin2021cafe}
X.~Jin, P.-Y. Chen, C.-Y. Hsu, C.-M. Yu, and T.~Chen, ``Cafe: Catastrophic data
  leakage in vertical federated learning,'' \emph{NeurIPS}, vol.~34, pp.
  994--1006, 2021.

\bibitem{fu2022label}
C.~Fu, X.~Zhang, S.~Ji, J.~Chen, J.~Wu, S.~Guo, J.~Zhou, A.~X. Liu, and
  T.~Wang, ``Label inference attacks against vertical federated learning,'' in
  \emph{31st USENIX Security Symposium (USENIX Security 22), Boston, MA}, 2022.

\bibitem{abadi2016deep}
M.~Abadi, A.~Chu, I.~Goodfellow, H.~B. McMahan, I.~Mironov, K.~Talwar, and
  L.~Zhang, ``Deep learning with differential privacy,'' in \emph{Proceedings
  of the 2016 ACM SIGSAC conference on computer and communications security},
  2016, pp. 308--318.

\bibitem{lin2018deep}
Y.~Lin, S.~Han, H.~Mao, Y.~Wang, and B.~Dally, ``Deep gradient compression:
  Reducing the communication bandwidth for distributed training,'' in
  \emph{International Conference on Learning Representations}, 2018.

\bibitem{thapa2020splitfed}
C.~Thapa, P.~C.~M. Arachchige, S.~Camtepe, and L.~Sun, ``Splitfed: When
  federated learning meets split learning,'' in \emph{Proceedings of the AAAI
  Conference on Artificial Intelligence}, vol.~36, no.~8, 2022, pp. 8485--8493.

\bibitem{huang2020instahide}
Y.~Huang, Z.~Song, K.~Li, and S.~Arora, ``Instahide: Instance-hiding schemes
  for private distributed learning,'' in \emph{International conference on
  machine learning}.\hskip 1em plus 0.5em minus 0.4em\relax PMLR, 2020, pp.
  4507--4518.

\bibitem{zhang2018mixup}
H.~Zhang, M.~Cisse, Y.~N. Dauphin, and D.~Lopez-Paz, ``mixup: Beyond empirical
  risk minimization,'' in \emph{International Conference on Learning
  Representations}, 2018.

\bibitem{liu2021defending}
Y.~Liu, Z.~Yi, Y.~Kang, Y.~He, W.~Liu, T.~Zou, and Q.~Yang, ``Defending label
  inference and backdoor attacks in vertical federated learning,'' \emph{arXiv
  preprint arXiv:2112.05409}, 2021.

\bibitem{ghazi2021deep}
B.~Ghazi, N.~Golowich, R.~Kumar, P.~Manurangsi, and C.~Zhang, ``Deep learning
  with label differential privacy,'' \emph{Advances in neural information
  processing systems}, vol.~34, pp. 27\,131--27\,145, 2021.

\bibitem{yang2022differentially}
X.~Yang, J.~Sun, Y.~Yao, J.~Xie, and C.~Wang, ``Differentially private label
  protection in split learning,'' \emph{arXiv preprint arXiv:2203.02073}, 2022.

\bibitem{dryden2016communication}
N.~Dryden, T.~Moon, S.~A. Jacobs, and B.~Van~Essen, ``Communication
  quantization for data-parallel training of deep neural networks,'' in
  \emph{2016 2nd Workshop on Machine Learning in HPC Environments
  (MLHPC)}.\hskip 1em plus 0.5em minus 0.4em\relax IEEE, 2016, pp. 1--8.

\bibitem{aji2017sparse}
A.~F. Aji and K.~Heafield, ``Sparse communication for distributed gradient
  descent,'' in \emph{Proceedings of the 2017 Conference on Empirical Methods
  in Natural Language Processing}, 2017, pp. 440--445.

\bibitem{liu2020secure}
Y.~Liu, Y.~Kang, C.~Xing, T.~Chen, and Q.~Yang, ``A secure federated transfer
  learning framework,'' \emph{IEEE Intelligent Systems}, vol.~35, no.~4, pp.
  70--82, 2020.

\bibitem{zhang2022trading}
X.~Zhang, Y.~Kang, K.~Chen, L.~Fan, and Q.~Yang, ``Trading off privacy, utility
  and efficiency in federated learning,'' \emph{ACM Trans. Intell. Syst.
  Technol.}, may 2023.

\bibitem{kang2022framework}
Y.~Kang, J.~Luo, Y.~He, X.~Zhang, L.~Fan, and Q.~Yang, ``A framework for
  evaluating privacy-utility trade-off in vertical federated learning,''
  \emph{arXiv preprint arXiv:2209.03885}, 2022.

\bibitem{fan2021deepip}
L.~Fan, K.~W. Ng, C.~S. Chan, and Q.~Yang, ``Deepip: Deep neural network
  intellectual property protection with passports,'' \emph{IEEE Transactions on
  Pattern Analysis \& Machine Intelligence}, no.~01, pp. 1--1, 2021.

\bibitem{li2022fedipr}
B.~Li, L.~Fan, H.~Gu, J.~Li, and Q.~Yang, ``Fedipr: Ownership verification for
  federated deep neural network models,'' \emph{IEEE Transactions on Pattern
  Analysis and Machine Intelligence}, 2022.

\bibitem{geiping2020inverting}
J.~Geiping, H.~Bauermeister, H.~Dr{\"o}ge, and M.~Moeller, ``Inverting
  gradients-how easy is it to break privacy in federated learning?''
  \emph{Advances in Neural Information Processing Systems}, vol.~33, pp.
  16\,937--16\,947, 2020.

\bibitem{he2019model}
Z.~He, T.~Zhang, and R.~B. Lee, ``Model inversion attacks against collaborative
  inference,'' in \emph{Proceedings of the 35th Annual Computer Security
  Applications Conference}, 2019, pp. 148--162.

\bibitem{gu2021federated}
H.~Gu, L.~Fan, B.~Li, Y.~Kang, Y.~Yao, and Q.~Yang, ``Federated deep learning
  with bayesian privacy,'' \emph{arXiv preprint arXiv:2109.13012}, 2021.

\bibitem{hardy2017private}
S.~Hardy, W.~Henecka, H.~Ivey-Law, R.~Nock, G.~Patrini, G.~Smith, and
  B.~Thorne, ``Private federated learning on vertically partitioned data via
  entity resolution and additively homomorphic encryption,'' \emph{arXiv
  preprint arXiv:1711.10677}, 2017.

\bibitem{SecShare-Adi79}
\BIBentryALTinterwordspacing
A.~Shamir, ``How to share a secret,'' \emph{Commun. ACM}, vol.~22, no.~11, p.
  612–613, nov 1979. [Online]. Available:
  \url{https://doi.org/10.1145/359168.359176}
\BIBentrySTDinterwordspacing

\bibitem{oscar2022split}
\BIBentryALTinterwordspacing
O.~Li, J.~Sun, X.~Yang, W.~Gao, H.~Zhang, J.~Xie, V.~Smith, and C.~Wang,
  ``Label leakage and protection in two-party split learning,'' in
  \emph{International Conference on Learning Representations}, 2022. [Online].
  Available: \url{https://openreview.net/forum?id=cOtBRgsf2fO}
\BIBentrySTDinterwordspacing

\bibitem{secureboost}
K.~Cheng, T.~Fan, Y.~Jin, Y.~Liu, T.~Chen, D.~Papadopoulos, and Q.~Yang,
  ``Secureboost: A lossless federated learning framework,'' \emph{IEEE
  Intelligent Systems}, vol.~36, no.~06, pp. 87--98, nov 2021.

\bibitem{zou2022defending}
T.~Zou, Y.~Liu, Y.~Kang, W.~Liu, Y.~He, Z.~Yi, Q.~Yang, and Y.-Q. Zhang,
  ``Defending batch-level label inference and replacement attacks in vertical
  federated learning,'' \emph{IEEE Transactions on Big Data}, 2022.

\bibitem{zhu2019deep}
L.~Zhu, Z.~Liu, and S.~Han, ``Deep leakage from gradients,'' \emph{Advances in
  neural information processing systems}, vol.~32, 2019.

\bibitem{fan2019rethinking}
L.~Fan, K.~W. Ng, and C.~S. Chan, ``Rethinking deep neural network ownership
  verification: Embedding passports to defeat ambiguity attacks,''
  \emph{Advances in neural information processing systems}, vol.~32, 2019.

\bibitem{lecun2010mnist}
Y.~LeCun, C.~Cortes, and C.~Burges, ``Mnist handwritten digit database,''
  \emph{ATT Labs [Online]. Available: http://yann.lecun.com/exdb/mnist},
  vol.~2, 2010.

\bibitem{krizhevsky2009learning}
A.~Krizhevsky and G.~Hinton, ``Learning multiple layers of features from tiny
  images,'' \emph{Master's thesis, Department of Computer Science, University
  of Toronto}, 2009.

\bibitem{wu20153d}
Z.~Wu, S.~Song, A.~Khosla, F.~Yu, L.~Zhang, X.~Tang, and J.~Xiao, ``3d
  shapenets: A deep representation for volumetric shapes,'' in
  \emph{Proceedings of the IEEE conference on computer vision and pattern
  recognition}, 2015, pp. 1912--1920.

\bibitem{wang2017deep}
R.~Wang, B.~Fu, G.~Fu, and M.~Wang, ``Deep \& cross network for ad click
  predictions,'' in \emph{Proceedings of the ADKDD'17}, 2017, pp. 1--7.

\bibitem{lecun1998gradient}
Y.~LeCun, L.~Bottou, Y.~Bengio, and P.~Haffner, ``Gradient-based learning
  applied to document recognition,'' \emph{Proceedings of the IEEE}, vol.~86,
  no.~11, pp. 2278--2324, 1998.

\bibitem{NIPS2012_c399862d}
\BIBentryALTinterwordspacing
A.~Krizhevsky, I.~Sutskever, and G.~E. Hinton, ``Imagenet classification with
  deep convolutional neural networks,'' in \emph{Advances in Neural Information
  Processing Systems}, F.~Pereira, C.~Burges, L.~Bottou, and K.~Weinberger,
  Eds., vol.~25.\hskip 1em plus 0.5em minus 0.4em\relax Curran Associates,
  Inc., 2012. [Online]. Available:
  \url{https://proceedings.neurips.cc/paper/2012/file/c399862d3b9d6b76c8436e924a68c45b-Paper.pdf}
\BIBentrySTDinterwordspacing

\bibitem{he2016deep}
K.~He, X.~Zhang, S.~Ren, and J.~Sun, ``Deep residual learning for image
  recognition,'' in \emph{Proceedings of the IEEE conference on computer vision
  and pattern recognition}, 2016, pp. 770--778.

\bibitem{fan2020rethinking}
L.~Fan, K.~W. Ng, C.~Ju, T.~Zhang, C.~Liu, C.~S. Chan, and Q.~Yang,
  ``Rethinking privacy preserving deep learning: How to evaluate and thwart
  privacy attacks,'' in \emph{Federated Learning}.\hskip 1em plus 0.5em minus
  0.4em\relax Springer, 2020, pp. 32--50.

\bibitem{blanchard2017machine}
P.~Blanchard, E.~M. El~Mhamdi, R.~Guerraoui, and J.~Stainer, ``Machine learning
  with adversaries: Byzantine tolerant gradient descent,'' \emph{Advances in
  neural information processing systems}, vol.~30, 2017.

\bibitem{bagdasaryan2020backdoor}
E.~Bagdasaryan, A.~Veit, Y.~Hua, D.~Estrin, and V.~Shmatikov, ``How to backdoor
  federated learning,'' in \emph{International Conference on Artificial
  Intelligence and Statistics}.\hskip 1em plus 0.5em minus 0.4em\relax PMLR,
  2020, pp. 2938--2948.

\bibitem{krizhevsky2014cifar}
A.~Krizhevsky, V.~Nair, and G.~Hinton, ``The cifar-10 dataset,'' \emph{online:
  http://www. cs. toronto. edu/kriz/cifar. html}, vol.~55, no.~5, 2014.

\bibitem{liu2022cross}
Y.~Liu, X.~Liang, J.~Luo, Y.~He, T.~Chen, Q.~Yao, and Q.~Yang, ``Cross-silo
  federated neural architecture search for heterogeneous and cooperative
  systems,'' in \emph{Federated and Transfer Learning}.\hskip 1em plus 0.5em
  minus 0.4em\relax Springer, 2022, pp. 57--86.

\bibitem{liu2022vertical}
Y.~Liu, Y.~Kang, T.~Zou, Y.~Pu, Y.~He, X.~Ye, Y.~Ouyang, Y.-Q. Zhang, and
  Q.~Yang, ``Vertical federated learning,'' \emph{arXiv preprint
  arXiv:2211.12814}, 2022.

\end{thebibliography}



\end{document}